%% file: DivRAMmain.tex
\documentclass[a4paper,10pt]{article}
\usepackage[top=2cm, bottom=1.5cm]{geometry} %
\usepackage[utf8]{inputenc}
\usepackage[american]{babel}  
\usepackage{amsmath,amssymb}
\usepackage{amsthm}
\usepackage{mathtools}
\usepackage{hyperref}
\usepackage{enumerate}
\usepackage{paralist}
\usepackage{graphics}
\usepackage[commentColor=black]{algpseudocodex}
\usepackage{algorithm}
\usepackage{float}
\usepackage{todonotes}
\usepackage[many]{tcolorbox}

\newtcolorbox{header}{
  fontupper = \it\color{black}, 
  boxrule = 0.5pt,
  colframe = black,
  colback = white,
  rounded corners,
  arc = 5pt   
}

\newcommand{\almosttextwidth}{0.95\textwidth}

\newcommand{\From}{~\textbf{from}~ }
\newcommand{\To}{~\textbf{to}~ }
\newcommand{\Downto}{~\textbf{downto}~ } 
\newcommand{\fun}[1]{\textsc{#1}} 
\newcommand{\arr}[1]{\texttt{#1}} 
\newcommand{\var}[1]{\texttt{#1}} 

\newenvironment{nosAlgos}[1][]{
  \begin{minipage}[t]{0.00\textwidth}
    ~
  \end{minipage}
  \begin{minipage}[t]{0.85\textwidth}
    \begin{algorithm}[H]
      \caption{#1}
      \begin{algorithmic}[1]
}{
      \end{algorithmic}
    \end{algorithm}
  \end{minipage}
  \medskip
  }

\newcommand{\largeB}{B}
\newcommand{\smallB}{K}

\usepackage{tikz}
\usetikzlibrary{arrows}
\newtheorem{lemma}{Lemma}
\newtheorem{definition}{Definition}
\newtheorem{theorem}{Theorem}
\newtheorem{proposition}{Proposition}
\newtheorem{corollary}{Corollary}
\newtheorem{remark}{Remark}
\newtheorem*{notation}{Notation}
\newtheorem*{convention}{Convention}
\newtheorem{example}{Example}
\newtheorem{claim}{Claim}
\newtheorem{openpb}{Open problem}
\newtheorem{conjecture}{Conjecture}

\def\lc{\left\lceil}   
\def\rc{\right\rceil}
\def\lf{\left\lfloor}   
\def\rf{\right\rfloor}

\newcommand{\predtab}{\mathtt{PRED}}

\newcommand{\divtab}{\mathtt{DIV}}
\newcommand{\modtab}{\mathtt{MOD}}

\newcommand{\modop}{\;\mathtt{mod}\;}
\newcommand{\divop}{\;\mathtt{div}\;}

\newcommand{\xorop}{\mathtt{xor}}
\newcommand{\orop}{\mathtt{or}}
\newcommand{\andop}{\mathtt{and}}

\newcommand{\op}{\mathtt{op}}
\newcommand{\Input}{\mathfrak{I}}

\newcommand{\Op}{\mathtt{Op}}
\newcommand{\RAM}[1][{\Op}]{\mathcal{M}[#1]}

\newcommand{\lin}{\mathtt{LIN}}

\DeclareMathOperator{\linsp}{\ensuremath{\mathtt{LIN}^{\mathtt{succ,pred}}}}

\newcommand{\dlin}{\mathtt{DLIN}}

\DeclareMathOperator{\cdlin}{\ensuremath{\mathtt{CONST}}-\ensuremath{\mathtt{DELAY}_{\mathtt{lin}}}}

\DeclareMathOperator{\lindlin}{\ensuremath{\mathtt{LIN}}-\ensuremath{\mathtt{DELAY}_{\mathtt{lin}}}}

\DeclareMathOperator{\cdlinlinspace}{\ensuremath{\mathtt{CONST}}-\ensuremath{\mathtt{DELAY}}-
\ensuremath{\mathtt{LIN}}-\ensuremath{\mathtt{SPACE}_{\mathtt{lin}}}}

\DeclareMathOperator{\cdlinsp}{\ensuremath{\mathtt{CONST}}-\ensuremath{\mathtt{DELAY}}_{\mathtt{lin}}^{\mathtt{succ,pred}}}

\DeclareMathOperator{\clin}{\ensuremath{\mathtt{CONST}}_{\mathtt{lin}}}

\DeclareMathOperator{\clinsp}{\ensuremath{\mathtt{CONST}}_{\mathtt{lin}}^{\mathtt{succ,pred}}}

\begin{document}

\title{Which arithmetic operations can be performed \\
 in constant time in the RAM model with addition?}


\date{\today}

\author{Étienne Grandjean,\\ 
{GREYC, Normandie Univ, UNICAEN, ENSICAEN, CNRS,
   Caen, France}\\
  Louis Jachiet,\\
 {LTCI, Télécom Paris, Institut polytechnique de Paris, France}
}

\maketitle
\tableofcontents

\input{1_Introduction.tex}

\input{2_Preliminaries.tex}

\input{3_classes.tex}

\input{4_BaseK.tex}

\input{5_Division.tex}

\input{7_ExpLog.tex}

\input{6_Roots.tex}

\input{8_Others.tex}

\input{9_SuccRAMweak.tex}

\input{10_Conclusion.tex}

\bigskip
\noindent
\emph{Acknowledgments:} 
This work was partly supported by the PING/ACK
project of the French National Agency for Research
(ANR-18-CE40-0011). 
We thank Antoine Amarilli and Luc Segoufin for
their proofreading of the manuscript. 
In particular, we thank Luc for his comment that led to the drafting of
Subsection~\ref{subsec:RAMwithLargerSpace} 
``RAM models with a different space usage''.

\bibliographystyle{plain}
\bibliography{biblioDivRAM}

\input{Appendix.tex}

\end{document}

%% file: 1_Introduction.tex
\paragraph{Abstract:}

The Church-Turing thesis~\cite{WikipediaChurchTuring} states that the
set of computable functions is the same for any ``reasonable''
computation model. When looking at complexity classes such
as \textsc{P}, \textsc{Exptime}, \textsc{Pspace}, etc., once again, a
variation of this thesis~\cite{WikipediaChurchTuring} states that these
classes do not depend on the specific choice of a computational model as
long as it is somewhat ``reasonable''\footnote{We are talking about classical computers here, to the exclusion of quantum computers which can, for example, factorize any integer in polynomial time. The question of knowing if the factorization problem can be solved in polynomial time is a major open problem in the classical framework!}. In contrast to that, when talking
of a function computable in linear time, it becomes crucial to specify
the model of computation used and to define what is the time required
for each operation in that model.

In the literature of algorithms, the specific model is often not
explicit as it is assumed that the model of computation is the RAM
(Random Access Machine) model. However, the RAM model itself is
ill-founded in the literature, with disparate definitions and no
unified results.  

The ambition of this paper is to found the RAM model from scratch by exhibiting a RAM model that enjoys interesting
algorithmic properties and the robustness of its complexity classes. With that goal in mind, we have tried to make this article as progressive, self-contained, and comprehensive as possible so that it can be read by graduate or even undergraduate students.

The computation model that we define is a RAM in which the values
contained in registers as well as the addresses of registers
(including input registers) are $O(N)$, where $N$ is the size
(number of registers) of the input, and where the time cost of
each instruction is 1 (unit cost criterion).

The key to the foundation of our RAM model will be to prove that even
if addition is the only primitive operation, such a RAM can still
compute all the arithmetic operations taught at school, subtraction, multiplication and
above all division, in constant time after a linear-time
preprocessing.

Moreover, while handling only $O(N)$ integers in each register is a strict limit of
the model, we will show that our RAM can handle $O(N^d)$ integers, for any fixed~$d$,
by storing them on $O(d)$ registers and we will have surprising
algorithms that compute many operations acting on these
``polynomial'' integers -- addition, subtraction, multiplication,
division, exponential, integer logarithm, integer square root (or $c$th root, for any integer~$c$),
bitwise logical operations, and, more generally, any operation computable in linear time on a cellular automaton -- in constant time after a linear-time preprocessing.

The most important consequence of our work is that, on the RAM model,
the $\lin$ class of linear-time computable problems, or the now
well-known $\cdlin$ class of enumeration problems computable with
constant delay after linear-time preprocessing, are invariant with
respect to which subset of these primitive operations the RAM provides
as long as it includes addition.

\section{Introduction and discussion of the RAM model}

While the Turing machine is the ``standard'' model in computability theory and polynomial time complexity theory, the standard character of this model is not essential due to the Church-Turing thesis and the ``complexity-theoretic'' Church-Turing thesis~\cite{WikipediaChurchTuring}: the polynomial time complexity class does not depend on the precise model of computation provided it is ``reasonable''.

Nevertheless, it is well-known that the Turing machine is too rudimentary to model the precise functioning of the vast majority of concrete algorithms. Instead, the complexity of algorithms is measured up to a constant factor using the RAM model (Random Access Machine)~\cite{AhoHU74,WikipediaRAM}.

However, there is a strange paradox about the RAM model, which has persisted since the model was introduced more than half a century ago~\cite{ShepherdsonS63, ElgotR64, Hartmanis71, CookR73, AhoHU74}: 
\begin{itemize}
\item on the one hand, the RAM is recognized as the appropriate machine model ``par excellence" for the implementation of algorithms and the precise analysis of their time and space complexity~\cite{CookR73, AhoHU74, WagnerW86, Papadimitriou94, Grandjean94-cc, Grandjean94, Grandjean96, Schwentick97, Schwentick98, GrandjeanSchwentick02, GrandjeanOlive04, FlumG06, CormenLRS09, WikipediaRAM};
\item on the other hand, unfortunately, the algorithmic literature does not agree on a standard definition of RAMs and the time costs of their instructions.
\end{itemize}
 
\paragraph{A confusing situation.}
The situation looks like a ``Babel tower": to describe and analyze algorithms, while some manuals on algorithms~\cite{Knuth68, PapadimitriouS82, DasguptaPV08, SedgewickW16} only introduce a pidgin programming language or a standard one (Pascal, C, Java, Python, etc.) instead of the RAM model, many other books on algorithms~\cite{AhoHU74, Tarjan83, Kozen92, CrochemoreR94, MotwaniR95, Niedermeier06, CormenLRS09, Skiena20} or on computational complexity~\cite{HopcroftU79, WagnerW86, Papadimitriou94, FlumG06} ``define'' and discuss ``their" RAM models.

The problem is that the RAM model, whose objective is to model the algorithms in a ``standard" way, presents a multitude of variants in the literature. While the representation of memory as a sequence of registers with indirect addressing ability is a common and characteristic feature of all versions of RAM -- random access means indirect addressing --, there is no agreement on all other aspects of the definition of this model:

\begin{description}

\item[Contents of registers:] the data stored in registers are either (possibly negative) integers~\cite{CookR73, AhoHU74, Papadimitriou94}, or natural numbers~\cite{WagnerW86, FlumG06}; in a few rare references~\cite{Tarjan83}, real numbers are also allowed.

\item[Input/output:]
an input (or an output) is either a word (on a fixed finite alphabet)
which is read (or written) letter by letter~\cite{WagnerW86,
CrochemoreR94,FlumG06}, or a sequence of integers read (or written)
sequentially one by one~\cite{CookR73, AhoHU74}; in several
references~\cite{Papadimitriou94, Schwentick97, Schwentick98,
GrandjeanSchwentick02, GrandjeanOlive04}, an input is a
one-dimensional array of integers (a list of input registers) accessed
by an index (i.e.\ they can be access in any order).

\item[Operations:] the set of authorized operations consists sometimes of the addition and the subtraction~\cite{CookR73, Grandjean94-cc, Grandjean94, Grandjean96, Schwentick97, Schwentick98, GrandjeanSchwentick02, GrandjeanOlive04}, sometimes of the four operations $+,-,\times,/$, see~\cite{AhoHU74}; in other cases, it is a subset of either of these two sets plus possibly some specific operations~\cite{Tarjan83, WagnerW86, Papadimitriou94, FlumG06, Skiena20}: division by 2, bitwise Boolean operations, etc.

\item[Sizes of registers:] in most references, see e.g.~\cite{CookR73, AhoHU74, HopcroftU79, WagnerW86, FlumG06, Papadimitriou94}, the integer contents of registers are a priori unbounded; in other papers, they must be ``polynomial", i.e.\ at most $O(N^d)$,  see~\cite{AngluinV79, GurevichS89, CrochemoreR94, CormenLRS09}, for a constant $d$, or ``linear" $O(N)$, see~\cite{Grandjean94-cc, Grandjean94, Grandjean96, Schwentick97, Schwentick98, LautemannW99, GrandjeanSchwentick02, GrandjeanOlive04}, where $N$ is the ``size'' of the input.

\item[Costs of instructions:] the literature presents the following two measures defined in the pioneering paper by Cook and Reckhow~\cite{CookR73} and the reference book by Aho, Hopcroft and Ullman~\cite{AhoHU74}: 

\begin{description}
\item[Unit cost criterion:] the execution time of each instruction is uniformly 1 (therefore, it is also called ``uniform cost criterion'', e.g.\ in~\cite{AhoHU74}); this makes sense if the integers allowed are bounded, linear, or polynomial in the size of the input;
\item[Logarithmic cost criterion:] the execution time of each instruction is the sum of the lengths of the binary representations of the integers handled by the instruction (addresses and contents of registers), or is 1 if the instruction uses no integer.
\end{description}
It should be noted that in practice, the books and papers on algorithms, see e.g.~\cite{AhoHU74, PapadimitriouS82,  Tarjan83, Kozen92, CrochemoreR94, MotwaniR95, Niedermeier06, FlumG06, CormenLRS09, Skiena20}, almost \emph{exclusively} apply the unit cost criterion, which is more intuitive and easier to handle than the logarithmic criterion.

\end{description}

The confusing situation of the RAM model that we have just described is all the more damaging since the new algorithmic concepts that have appeared in recent years, see for example~\cite{Durand20}, require a finer complexity analysis than for traditional algorithms.

\paragraph{Renew the understanding of the RAM model for the analysis of ``dynamic'' algorithms.}

Since the early 2000s, new concepts have been developed in artificial intelligence (AI), logic, database theory and combinatorics. Basically, they relate to ``dynamic'' problems whose resolution works in two or more phases.

As a significant example, a \emph{knowledge compilation problem} in
AI~\cite{CadoliD98, DarwicheM02, Marquis15} is a query whose input
breaks down into a ``fixed part'' $F$ and a (much smaller) ``variable
part" $V$.  An algorithm solving such a query operates in two
successive phases: a (long) ``preprocessing phase'' acting on the
``fixed part'' $F$; a ``response phase", which reads the ``variable
part" $V$ of the input and answers (very quickly) the query on the
$(F, V)$ input.

A very active field of research for the past fifteen years concerns \emph{enumeration algorithms} in logic, the theory of queries in databases, and combinatorics, see e.g.\ the surveys~\cite{Segoufin13, Segoufin14, Wasa16, Strozecki19, BerkholzGS20, Durand20}. A ``fine-grained" and now widespread complexity class (which has some variants) has emerged for enumeration problems, see~\cite{Bagan06, DurandG07, BaganDG07, BaganDGO08, Courcelle09, Bagan09, KazanaS11, KazanaS13, DurandSS14, Uno15, AmarilliBJM17, BerkholzKS18, KazanaS19, CarmeliZBKS20}: the class of problems whose set of solutions can be enumerated with \emph{constant delay} (between two consecutive solutions) after a \emph{linear-time} preprocessing. Each of these papers is based on the RAM model and uses the ``minimal" (or quasi-minimal) complexity classes of this model: mainly, linear time, for decision or counting problems, and, for enumeration problems, constant delay after linear (or quasi-linear, or polynomial) time preprocessing. 

More recently, several authors in database theory have introduced what they call the ``dynamic query evaluation algorithms"~\cite{BerkholzKS17,BerkholzKS18}. The purpose of these algorithms is to modify the answer of a query to a database, in \emph{constant} (or \emph{quasi-constant}) \emph{time} when the database is updated by inserting or removing one tuple.

The ``fine-grained" complexity classes associated with these new algorithmic concepts (dynamic algorithms) make a thorough understanding of the RAM computation model all the more urgent. 
The main objective of this paper is to elucidate and justify precisely what ``the" RAM model ``is" or ``should be". 

\paragraph{What should the RAM model be?}
Our requirements are as follows:
\begin{itemize}
\item the RAM model we ``re-define" must be \emph{simple} and \emph{homogeneous}: in particular, the RAM input and the RAM memory must have the same structure; also, the set of allowed register contents must equal the set of allowed addresses; 
\item the RAM model must ``stick'' closely to the  \emph{functioning of algorithms on discrete structures} (integers, graphs, etc.);
\item the ``fine-grained" complexity classes~\cite{Durand20} defined in the RAM model must be ``robust": it must be proven that these complexity classes are largely  \emph{independent of the precise definition} of the RAM and, in particular, of the \emph{arithmetic operations} allowed. 
\end{itemize}

 \begin{remark}[What the RAM model does not model]
 
 In the 80s and 90s, a few authors, see e.g.~\cite{AggarwalCS87,Regan96}, introduced several versions of the RAM model in order to model the memory hierarchy of real computers: fast memory (registers) and slower memory with cache memory hierarchy, etc. 
 
Nevertheless, we require that the RAM model be homogeneous, thus without memory hierarchy, as almost all the literature on algorithms does. As the reference book of Cormen et al.~\cite{CormenLRS09} argues, there are two reasons: first, the analysis of an algorithm is very complicated in memory hierarchy models, which, moreover, are far from being unified in a standard form; the second and main reason is that the analysis of algorithms based on the RAM model without memory hierarchy ``generally gives excellent predictions of the performances obtained on real computers''~\cite{CormenLRS09}.

 \end{remark}
 
Our precise definition of the RAM model will be given in the preliminary section. Let us present below the main characteristics of our model with some elements of justification.
 
 \paragraph{The RAM model we choose}
 
 \begin{enumerate}

\item {\bf Contents of registers:}
each register contains a \emph{natural number}; a RAM with possibly
negative integers can be easily simulated by representing an integer
$x$ by the pair of natural numbers $(s_x,\vert x \vert)$ with $s_x=0$
if $x\geq 0$ and $s_x=1$ if $x<0$; real numbers are excluded for the
sake of homogeneity.

\item {\bf Input/output:}
in order for the inputs/outputs to be of the same type as the RAM
memory, a (standard) input is an \emph{integer}\footnote{From now on,
we write ``integer" in place of ``natural number".}  $N>0$ or
a \emph{one-dimensional array of integers} $(N,I[0],\ldots,I[N-1])$;
similarly, an output is an integer or a list of integers. For the sake
of simplicity, even when the input is solely composed of $N$, we call
the positive integer $N$, the ``size" of the input.

\item {\bf Operations:} our RAM has a basic instruction set plus a set  of
      primitive operations. In this paper the set of primitive operations
      consists only of the \emph{addition}.

\item {\bf Sizes of registers:}
the contents of each register (including each input register $I[j]$)
is $O(N)$, where $N$ is the ``size'' of the input, or equivalently,
its length in binary notation is at most $\log_2(N)+O(1)$.

\item {\bf Costs of instructions:} the execution time of each instruction is 1 (\emph{unit cost} criterion).

\end{enumerate}

Our choice of a one-dimensional array of integers $(N,I[0],\ldots,I[N-1])$ of size $N$ as standard input (item 2) is convenient for representing usual data structures such as trees, circuits, (valued) graphs, hypergraphs, etc., while respecting their size. For example, it is easy to represent a tree of $n$ nodes as a standard input of size $\Theta(n)$, or a graph of $n$ vertices and $m$ edges 
as a standard input
\footnote{For example, a graph without isolated vertex whose vertices are $1,\ldots, n$ and edges are $(a_1,b_1),\ldots, (a_m,b_m)$ is ``naturally represented'' by the standard input $\mathcal{I}=(N,I[0],\ldots,I[N-1])$ where $N\coloneqq 2m+2$,  $I[0]\coloneqq m$,  $I[1]\coloneqq n$, and finally $I[2j]\coloneqq a_j$ and $I[2j+1]\coloneqq b_j$ for $j=1,\ldots, m$. Note that since the maximum number of vertices of a graph having $m$ edges and no isolated vertices is $2m$
, we have $n\leq 2m=N-2$, 
so that each $I[j]$ is less than $N$.}
of size $\Theta(m+n)$.

We do not yet substantiate items (3-5) of our RAM model. Most of this paper is devoted to showing that limiting primitive operations to addition alone, which acts on integers in $O(N)$, is not a limitation of the computing power of RAMs. In fact, we will prove that each of the usual arithmetic operations, subtraction, multiplication, division, and many others, acting on integers in $O(N)$, even on ``polynomial'' integers in $O(N^d)$ represented in $d$ registers (for a fixed integer~$d$), can be simulated in \emph{constant time} in our RAM model with only the addition operation. 
This will show the ``robustness'' and high computing power of our ``simple" RAM model, which is the main focus of this paper.  

\paragraph{Structure of the paper and reading tips:} 
We have tried to make this long paper as self-contained, educational and modular as possible. 
For this, we have given complete, detailed and elementary proofs\footnote{The only exception is the proof of Theorem~\ref{th:linTimeCA} which requires implementing elaborate simulations of cellular automata acting in linear time.} and structured our many sections so that they can be read, most of the time, independently of each other. 

Sections~\ref{sec:Prelimin} and~\ref{sec:RAMcomplexityclasses} define
our RAM model -- with three equivalent instruction sets -- and their
complexity classes\footnote{Moreover,
Subsection~\ref{subsec:RAMwithLargerSpace} studies the properties of
the RAM model when it is extended by allowing $k$-dimensional arrays
(for a fixed $k$), that is, with $O(N^k)$ registers available.}.
Those sections also introduce two kinds of RAM simulations with their
transitivity properties: first, the \emph{faithful simulation} which
is very sensitive to the primitive arithmetic operations of RAMs as
Proposition~\ref{prop:multWithAdd} states; second, the \emph{faithful
simulation after linear-time initialization} (also
called \emph{linear-time preprocessing}) which we prove is the right
notion to enrich the set of primitive operations of RAMs while
preserving their complexity classes, see
Corollary~\ref{cor:transitivity} (transitivity corollary).

Although Sections~\ref{sec:Prelimin}
and~\ref{sec:RAMcomplexityclasses} present the precise concepts and
their properties to formally ground the results of the paper, the
reader can skip them on the first reading of the main results, which
are given in Sections~\ref{sec:SumDiffProdBase}
to~\ref{sec:bitwiseOperConstTime}, while contenting herself with an
intuitive notion of simulation.  These sections show how to enrich the
RAM model endowed with only the addition operation, by a multitude of
operations -- the most significant of which are \emph{Euclidean
division} (Section~\ref{section:division}) and any
operations \emph{computable in linear time on cellular automata}
(Subsection~ \ref{subsec:TuringCA} and Appendix)~-- while preserving
its complexity classes.

Third, Section~\ref{sec: MinAdd} proves several results that argue for the ``minimality'' of the RAM model with addition.

Section~\ref{sec:conclusion} (concluding section) recapitulates and completes our results establishing the invariance/robustness of our ``minimal'' complexity classes, mainly $\lin$, $\clin$ and $\cdlin$, according to the set of primitive operations (Corollaries~\ref{cor:invOp+} and~\ref{cor:invRAM+}). 
We also explain why and how our RAM model ``faithfully'' models the computations of algorithms on combinatorial structures (trees, graphs, etc.) and discuss the more subtle case of string/text algorithms. 
Finally, Section~\ref{sec:conclusion} gives a list of open problems, supplementing those presented throughout the previous sections.

\paragraph{Small digression: a comparison with Gurevich’s Abstract State Machines.}
In addition to the sequential computation models that are the Turing machine, the ``Storage Modification Machine'' (or ``pointer machine'', see Section~\ref{sec:conclusion}) and the RAM that we study here, we must mention the model of ``Abstract State Machines'' (ASM) introduced by Gurevich in the 80s~\cite{Gurevich93} and since studied in many publications~\cite{BoergerH98, Gurevich00, BlassGRR07, BlassGRR07a, Gurevich12}. 
The ambition of Gurevich and his collaborators is to faithfully model, that is to say step by step, the functioning of all the algorithms (sequential algorithms, parallel algorithms, etc.).

Although the goals of the ASM model are foundational, like ours, they are much broader in a sense:
\begin{itemize}
\item its objectives are to refound the theory of computation, algorithms and programming on a single model, the ASM;
\item this unique model includes/simulates in a faithful and natural way the functioning of all algorithms and computation models: Turing machines, RAMs, etc. (see e.g.~\cite{BlassGRR07, BlassGRR07a}).
\end{itemize}

Our objectives in this paper are different:
\begin{itemize}
\item we limit our paper to a systematic study of the RAM model alone and we do not try to model
parallel algorithms, quantum computers, etc.;
\item while we are also interested in the fine functioning of algorithms, our main objective is to study the properties of fine-grained complexity classes, in particular minimal classes; 
\item for this, our main tool is not the \emph{faithful simulation}, similar to the \emph{lock-step simulation} of the ASM (see e.g.~\cite{DexterDG97}), because it is very sensitive to the primitive arithmetic operations of the RAMs (by our Proposition~\ref{prop:multWithAdd}), 
but it is the \emph{faithful simulation after linear-time initialization}, which is not allowed in the ASM framework.
\end{itemize}

%% file: 2_Preliminaries.tex
\section{Instruction sets for the RAM model}\label{sec:Prelimin}

In this section, we define our computation model. This is the standard
RAM with addition using only integers $O(N)$: here, $N$ is the integer
corresponding to the ``size" of the input, i.e.\ the nonnegative number
of registers occupied by the input or, when an arithmetic operation is
simulated, an upper bound on the value of the operands. Note that this
$O(N)$ limit on the contents of registers also places a limit on the
memory since only cells whose addresses are $O(N)$ can be accessed.

To convince the reader of the flexibility of the model, we define a RAM model
with a minimal instruction set, called $\RAM{}$, where $\mathtt{Op}$
denotes a set of allowed operations (e.g.\ addition, subtraction,
etc.). We will also introduce two other minimal instruction sets and
show that all these models are in fact equivalent.

For the sake of simplicity, since in what follows we are interested in
``minimal" complexity classes, namely \emph{linear time} complexity
and \emph{constant time} complexity after \emph{linear time
preprocessing}, 
our RAM model also uses \emph{linear space}.
However, it would be possible to define the RAM model for any
time $T(N)$ and space~$S(N)$ complexity class, provided that the
functions $N\mapsto T(N)$ and $N\mapsto S(N)$ meet some ``natural"
conditions: constructibility (e.g., see~\cite{BalcazarDG88}, section~2.4), etc.


\subsection{The RAM model with minimal instruction set}\label{subsec: AB-inst}

A \emph{RAM} $\RAM{}$ has two special registers $A$ (the
\emph{accumulator}) and $B$ (the \emph{buffer}), \emph{work registers}
$R[j]$ (the \emph{memory}), \emph{input registers} $I[j]$, for all
$j\geq 0$, and an \emph{input size register}~$N$.  (For the sake of
simplicity, throughout this paper, we identify a RAM with its
program.)

\paragraph{Syntax.} A \emph{program/RAM} ($AB$-\emph{program}, $AB$-\emph{RAM})
of $\RAM{}$ is a sequence $I_0,\ldots, I_{r-1}$ of $r$
\emph{instructions}, also called $AB$-\emph{instructions}, where each of the $I_i$ has one of the following
forms:

\begin{table}[H]
  \centering
\begin{tabular}{|clr|}
  \hline
  \bf{Instruction} &  \bf{Meaning} &  \bf{Parameter} \\
  \hline
  $\mathtt{CST}$ $j$ & $A \gets j$ & for some constant integer $j\geq 0$ \\
  $\mathtt{Buffer}$  & $B \gets A$ & \\ 
  $\mathtt{Store}$ & $R[A] \gets B$ & \\ 
  $\mathtt{Load}$ & $A \gets R[A]$ &  \\
  $\mathtt{Jzero} \; {\ell_0} \; {\ell_1}$ &  $\mathtt{if}\; A=0\;\mathtt{then}\;\mathtt{goto}\;\ell_0\;\mathtt{else}\;\mathtt{goto}\;\ell_1$& for $0\leq \ell_0,\ell_1 \leq r$
  \\
  $\op$ & see semantics below & for $\op$ in $\Op$ \\
  $\mathtt{getN}$ & $A\gets N$ & \\
  $\mathtt{Input}$ & $A \gets I[A]$ & \\ 
  $\mathtt{Output}$ & output $A$ & \\
  \hline
\end{tabular}
\caption{Possible instructions for $\RAM{}$}
\label{table:instRAM}
\end{table}


\paragraph{Semantics.}
The semantics, called \emph{computation}, of an $AB$-program $I_0,
\dots, I_{r-1}$ on a given input is defined rather
intuitively. Formally, it is defined as the (possibly infinite)
sequence of the \emph{configurations} ($AB$-\emph{configurations})
$\mathcal{C}_0,\mathcal{C}_1,\dots$ where each
$\mathcal{C}_i\coloneqq(A_i,B_i,\lambda_i,D_i)$ consists of the
contents $A_i,B_i$ of the two registers $A,B$, the index $\lambda_i$
of the current instruction, as well as the description $D_i$ of the
memory (the contents of the $R[j]$) at instant~$i$.

The \emph{initial configuration} is $A_0=B_0=\lambda_0=0$ with the array $D_0$
initialized to 0. Then, from a configuration $\mathcal{C}_i=(A_i,B_i,\lambda_i,D_i)$,
the next configuration -- if it exists -- is deduced from the instruction $I_{\lambda_i}$ using
the semantics given in Table~\ref{table:instRAM}.
The instruction $\mathtt{Jzero}\;\ell_0\;\ell_1$ is a conditional jump for which we have 
$\lambda_{i+1}=\ell_0$ when $A_i=0$, and $\lambda_{i+1}=\ell_1$ otherwise; 
for all other instructions, we have $\lambda_{i+1}=\lambda_i+1$. 
Note that we allow $\lambda_{i+1}=r$: in that case, 
we say that $C_{i+1}=(A_{i+1},B_{i+1},\lambda_{i+1},D_{i+1})$ where $A_{i+1}=A_i$, $B_{i+1}=B_i$, $D_{i+1}=D_i$ and $\lambda_{i+1}=r$ is a \emph{final configuration}, it is the last of the computation. 
For an instruction $\op$, its exact semantics depends on the particular operand but for
all classical unary or binary operations (addition, multiplication,
etc.) we assume that it reads $A$ (and $B$ when binary) and
stores its result in $A$, i.e., its semantics is $A \gets \mathtt{op}(A)$ 
(resp.\ $A \gets \mathtt{op}(A,B)$) when $\mathtt{op}$ is unary (resp. binary). 

The \emph{output} of a program that stops is the list of values that
it outputs in the order it produces them. The \emph{running time} of a
program is the number of configurations it goes through before ending
in its final configuration.

\paragraph{Input.}
In this paper, we will consider two forms of random access inputs:
\begin{itemize}
\item a list of integers $(N,I[0],\ldots,I[N-1])$, for $N>0$ --
  contained in the registers having the same names -- where each
  $I[j]$, $0\leq j <N$, is such that $0\leq I[j] \leq cN$, for some
  constant integer $c\geq 1$;
\item a single integer $N>0$, contained in the register $N$.
\end{itemize}

Note that these types of inputs impose a limit on what we can model.
For instance, it does not allow us to model streaming algorithms or
online algorithms where the algorithm is required to read the data in
a particular order because, in our model, the input is read in random
order.
We believe this design is adequate for our purpose (grounding the RAM model) but this limitation is not fundamental and the input/output model could easily be tweaked for many specific needs. 

\paragraph{An illustrative program.}
We now demonstrate how to use $\RAM[\{+\}]$, the RAM equipped with
addition (denoted $+$ or $\mathtt{add}$), with the repetitive but
simple $AB$-program given by Table~\ref{table:exampleProgram}.
We claim that the $AB$-program is equivalent to the following ``high-level'' program (in pseudo-code format), which reads the input size $N$, assumed to be greater than 1, sets the registers $1$ to $N-1$ to the value $I[0]$ and outputs/returns the value $I[1]$. 

For this, it uses a ``while'' construct, which does not exist in our RAM, 
and we also adopt the convention that the variable~$i$ is stored within the $R[0]$ register.

\begin{nosAlgos}[Example program in pseudo-code format]
\State $R[N] \gets 0$ \Comment{lines 0 - 3}
\State $i\gets 1$ \Comment{expressed by $R[0] \gets 1$ since $i$ is stored in $R[0]$, see lines 4 - 7} 
\While{$R[N] = 0$} \Comment{loop test at lines 8 - 10}
  \State $R[i] \gets I[0]$ \Comment{lines 11 - 16}
  \State $R[i+1] \gets i+1$ \Comment{lines 17 - 21}
  \State $i \gets i + 1$ \Comment{lines 22 - 23}
\EndWhile
\State \Output{$I[1]$} \Comment{lines 25 - 27}
\end{nosAlgos}

For better readability, each instruction of Algorithm 1 is
annotated with the lines it corresponds to in our $AB$-program;
additionally, the $AB$-instructions of Table~\ref{table:exampleProgram} are grouped according to each basic instruction of Algorithm 1 that they simulate together 
(bold lines beginning by~$\vartriangleright$)
and the column ``Effect'' also details what each $AB$-instruction does.

\begin{table}[H]
  \centering
  \begin{tabular}{clr}
    \textbf{\#}  & \textbf{Instruction} & \textbf{Effect} \\
    \\
 $\vartriangleright$ & $\mathbf{R[N]\gets 0}$   &   \\
    $0$  & $\mathtt{CST} \; 0 $   & $A \gets 0 $        \\
    $1$  & $\mathtt{Buffer} $   & $B \gets A=0 $        \\
    $2$  & $\mathtt{getN}  $   & $A \gets N $        \\
    $3$  & $\mathtt{Store} $   & $R[A]\gets B$, which means $R[N] \gets 0 $        \\
    \\
    $\vartriangleright$ & $\mathbf{i \gets 1}$   &   \\
    $4$  & $\mathtt{CST} \; 1 $   & $A \gets 1 $        \\
    $5$  & $\mathtt{Buffer} $   & $B \gets A= 1 $        \\
    6 & $\mathtt{CST} \; 0 $   & $A \gets 0 $        \\
    7 & $\mathtt{Store} $   & $R[A] \gets B$, which means $R[0] \gets 1$, or $i \gets 1$        \\

    \\ 
    $\vartriangleright$ & \bf{if} $\mathbf{R[N]=0}$ \bf{then enter the loop}  &   \\
    8 & $\mathtt{getN}  $   & $A \gets N $        \\
    9 & $\mathtt{Load}  $   & $A \gets R[A]=R[N]$      \\
   10 & 
   $\mathtt{Jzero} \; 11 \; 25  $   & Enter loop when $A=0$, i.e.\ when $R[N]= 0$,
   \\ & & else go to loop end    \\
    \\
    $\vartriangleright$ & \bf{loop body start}   &   \\
    \\
    $\vartriangleright$ & $\mathbf{R[i]\gets I[0]}$   &   \\
   11            & $\mathtt{CST} \; 0 $   & $A \gets 0 $        \\
    $12$  
    & $\mathtt{Input} $   & $A \gets I[A]=I[0] $     \\
    $13$  
   & $\mathtt{Buffer}$     & $B \gets A= I[0]$ \\
    $14$  
   & $\mathtt{CST} \; 0$     & $A \gets 0$ \\
    $15$  
    & $\mathtt{Load}$     & $A \gets R[A]$, which means $A \gets R[0] = i$ \\
    $16$  
    & $\mathtt{Store}$     & $R[A]\gets B$, which means $R[i] \gets I[0]$ \\
    \\
    $\vartriangleright$ & $\mathbf{R[i+1] \gets i+1}$   &   \\
    $17$  
    & $\mathtt{Buffer}$     & $B \gets A= i$ \\
    $18$  
    & $\mathtt{CST} \; 1$     & $A \gets 1$ \\
    $19$  
    & $\mathtt{add}$     & $A \gets A+B= 1+i$ \\
    $20$  
    & $\mathtt{Buffer}$     & $B \gets A= i+1$ \\
    $21$  
    & $\mathtt{Store}$     & $R[A]\gets B$, which means $R[i+1] = i+1$ \\
    \\
    $\vartriangleright$ & $\mathbf{i\gets i+1}$   &   \\
    $22$  
    & $\mathtt{CST} \; 0$     & $A \gets 0$ \\
    $23$  
    & $\mathtt{Store}$     & $R[A]\gets B$, which means $i = R[0] \gets i+1$ \\
    \\
    $\vartriangleright$ & \bf{back to the loop test}   &   \\
    $24$   
    &$\mathtt{JZero} \; 10 \; 10$     & unconditional jump to the conditional ``go to'' \\
    \\
    $\vartriangleright$ & \bf{loop end}   &   \\
    \\
    $\vartriangleright$ & $\mathbf{output\; I[1]}$   &   \\
    $25$  
    & $\mathtt{CST} \; 1$     & $A \gets 1$ \\
    $26$  
    & $\mathtt{Input}$     & $A \gets I[A]=I[1]$ \\
    $27$  
    & $\mathtt{Output}$   & $\mathtt{Output}\; A$, which means  $\mathtt{Output}\; I[1]$ \\    
\end{tabular}
\caption{Example of $AB$-program of $\RAM[\{+\}]$}
\label{table:exampleProgram}
\end{table}

\paragraph{Minimality of the instruction set.}
No $AB$-instruction of Table~\ref{table:instRAM} can be removed
without diminishing the expressive power of the model: $\mathtt{CST}$
$j$ must be used to compute the constant function $N \mapsto j$ when
$\mathtt{Op}$ is the empty set; if we remove $\mathtt{Load}$
(resp.\ $\mathtt{Buffer}$ or $\mathtt{Store}$), then the memory
registers $R[j]$ cannot be read (resp.\ modified); if we remove the
jump instruction $\mathtt{Jzero} \; \ell_0 \; \ell_1$ then any program
performs a constant number of instructions, and finally, it is obvious
that the input/output instructions $\mathtt{getN}$, $\mathtt{Input}$,
and $\mathtt{Output}$ cannot be removed.
Note that the minimality we claim is that no instruction from our RAM
can be removed without decreasing expressiveness. This does not mean
that another RAM needs at least as many instructions to be as
expressive, in particular, this statement is not contradicted by the
existence of one-instruction set computers~\cite{oisc}.


\subsection{Two more instruction sets}

We have just observed that the RAM with the $AB$-instructions is
minimal in the following sense: removing one of the $AB$-instructions
listed in Table~\ref{table:instRAM} strictly weakens the computational
power of the model.  Recall that our main goal in this paper is to
exhibit a ``minimal'' RAM in the sense of a RAM endowed with the least
number of operations on the registers: test at zero and addition, but
neither multiplication, nor comparisons $=$, $<$, etc. The second
(complementary) objective of this article is to establish the
``robustness'' of the RAM model: its computational power and
complexity classes are invariant under many changes to its
definition. To this end, this subsection will present two alternative
instruction sets and Subsection~\ref{sub:equivInst} will show that all
programs can be transformed from one instruction set to another with
time and space bounds multiplied by a constant.

\subsubsection*{RAM with assembly-like instruction set}
A RAM with $\mathtt{Op}$ operations and \emph{assembly-like instruction
set}, called $R$-RAM, has \emph{work registers} $R[j]$, \emph{input registers} $I[j]$, for all $j\geq
0$, and an \emph{input size register}~$N$.  An $R$-\emph{program} is a sequence
$I_0, \dots, I_{r-1}$ of labeled instructions, called
$R$-\emph{instructions}, of the following forms with given intuitive
meaning, where $i,j$ are explicit integers, $\ell_0, \ell_1$ are valid indexes
of instructions and $\op$ is an operator from $\Op$ of arity $k$:

\begin{table}[H]
  \centering
\begin{tabular}{|clr|}
  \hline
   \bf{Instruction} &  \bf{Meaning} & \\ 
  \hline
  $\mathtt{CST} \; i \; j$ & $R[i] \gets j$ & \\
  $\mathtt{Move} \; i \; j$ & $R[i] \gets R[j]$  & \\ 
  $\mathtt{Store}\; i \; j$ & $R[R[i]] \gets R[j]$  & \\ 
  $\mathtt{Load}\; i \; j$  & $R[i] \gets R[R[j]]$  & \\ 
  $\mathtt{Jzero} \; i \; {\ell_0} \; {\ell_1}$  & $\mathtt{if}\; R[i]=0\;\mathtt{then}\;\mathtt{goto}\;\ell_0\;\mathtt{else}\;\mathtt{goto}\;\ell_1$
  & \\ 
  $\op$ & $R[0]\gets \op (R[0],\ldots, R[k-1])$  & \\
  $\mathtt{getN}\; i$ & $R[i]\gets N$  & \\ 
  $\mathtt{Input}\; i \; j$ & $R[i] \gets I[R[j]]$  & \\ 
  $\mathtt{Output} \; i$ & output $R[i]$ & \\ 
  \hline
\end{tabular}
\caption{Set of $R$-instructions for $\RAM{}$}
\label{table:instRAM2}
\end{table}

\noindent
The semantics given in Table~\ref{table:instRAM2} is very similar to
that of our original set of instructions.  Formally, the
\emph{computation} of an $R$-\emph{program} $I_0, \dots, I_{r-1}$ on a
given input is defined as the sequence of the \emph{configurations}
($R$-\emph{configurations}) $\mathcal{C}_0,\mathcal{C}_1,\dots$ where
each $\mathcal{C}_i\coloneqq(\lambda_i,D_i)$ consists of the index $\lambda_i$
of the current instruction and the description $D_i$ of the RAM memory
(contents of the $R[j]$) at instant~$i$; in particular,
$\mathcal{C}_0=(\lambda_0,D_0)$ is the initial configuration with
$\lambda_0=0$ and the RAM memory $D_0$ initialized to 0.

Noticing that we can store the value of $A$ into $R[0]$ and the value
of $B$ into $R[1]$, it is relatively clear that any $AB$-program can
be simulated with an $R$-program. The third definition of RAM data
structures and instructions that follows is closer to those of a
programming language.

\subsubsection*{Multi-memory RAM}
 A multi-memory RAM with $\mathtt{Op}$-operations uses a fixed number $t$ ($t\geq 1$) of one-dimensional \emph{arrays of integers}
$T_0[0,\ldots],\ldots,T_{t-1}[0,\ldots]$, a fixed number of \emph{integer
variables} $C_0,\ldots,C_{v-1}$, an \emph{input array} $I[0,\ldots]$ of integers
and the \emph{input size register}~$N$.  The \emph{program} of $M$ uses
$\mathtt{Op}$-\emph{expressions} or, for short, \emph{expressions}, which are
defined recursively as follows.

\begin{definition}[$\mathtt{Op}$-\emph{expressions}]
\begin{itemize}
\item Any fixed integer $i$, the size integer $N$, and any integer
  variable $C_j$ are $\Op$-expressions.
\item If $\alpha$ is an $\Op$-expression, then an array
  element $T_j[\alpha]$ or $I[\alpha]$ is an $\Op$-expression.
\item If $\alpha_1, \dots, \alpha_k$ are $\Op$-expressions and $\mathtt{op}$ is a $k$-ary operation in $\mathtt{Op}$, then $\mathtt{op}(\alpha_1, \dots, \alpha_k)$ is an
  $\Op$-expression.
\end{itemize}
\end{definition}
\noindent
\emph{Example:} $T_1[T_3[2]+N]\times I[C_1]$ is a $\{+,\times\}$-expression.

\medskip
The \emph{program} of a multi-memory RAM with $\Op$ operations is
a sequence $I_0, \dots, I_{r-1} $ of instructions, called
\emph{array-instructions}, of the following forms with intuitive meaning, where $\alpha$ and~$\beta$ are $\Op$-expressions and  $\ell_0, \ell_1$ are valid indexes of instructions.

\begin{table}[H]
  \centering
\begin{tabular}{|lr|}
  \hline
   \bf{Instruction} &  \bf{Meaning} \\
  \hline
  $T_j[\alpha] \gets \beta$  & \\
  $C_j \gets \alpha$ & \\
  $\mathtt{Jzero} \; \alpha \; \ell_0 \; \ell_1$  & $\mathtt{if}\; \alpha=0\;\mathtt{then}\;\mathtt{goto}\;\ell_0\;\mathtt{else}\;\mathtt{goto}\;\ell_1$ \\
  $\mathtt{Output} \; \alpha$
  & \\
  \hline
\end{tabular}
\caption{Set of array-instructions for $\RAM{}$}
\label{table:instRAM3}
\end{table}

The \emph{computation} of a program with array-instructions (called \emph{array-program}) $I_0, \dots, I_{r-1}$ on a given input is defined as the sequence of the \emph{configurations} (\emph{array-configurations}) $\mathcal{C}_0,\mathcal{C}_1,\dots$ where the configuration
$\mathcal{C}_i\coloneqq(\lambda_i,D_i)$ at instant~$i$ consists of the index $\lambda_i$ of the current instruction ($\lambda_0=0$, initially) and the description $D_i$ of the memory (the values of the integer variables and arrays, all initialized to 0).


\subsection{Equivalence of our three instruction sets}\label{sub:equivInst}

The following lemma states that our three instructions sets (our three classes of RAMs) are ``equivalent'', within the meaning of Definition~\ref{def:equivalence} of Subsection~\ref{subsec:emul}, when the addition is available.
(Note that, for pedagogical reasons, the notions of ``simulation'' and ``equivalence'' involved in the statement of Lemma~\ref{lemma:lockstepEquivInst} are first presented here in an informal/intuitive way. Section~\ref{sec:RAMcomplexityclasses} will give formal definitions and complete evidence.)
\begin{lemma}
\begin{enumerate}
\item Any $R$-RAM can be simulated by an $AB$-RAM.
\item Any $AB$-RAM can be simulated by a multi-memory RAM.
\item If addition belongs to the set $\mathtt{Op}$ of allowed
  operations then any multi-memory RAM can be simulated by an $R$-RAM.
\end{enumerate}

\noindent 
Therefore, by transitivity, the three classes of $AB$-RAMs, $R$-RAMs,
and multi-memory RAMs are equivalent when
$\mathtt{add}\in\mathtt{Op}$.

\label{lemma:lockstepEquivInst} 
\end{lemma} 

\begin{proof}[Proof of item 1]
As representative examples, Tables~\ref{tab:emulOP},
\ref{tab:emulOutput}, \ref{tab:emulStore}, \ref{tab:emul_Input}
and~\ref{tab:emulJzero} give the simulation (by a sequence of
$AB$-instructions) of each of the $R$-instructions $\op$, for a binary
$\op$, $\mathtt{Output}\; i$, $\mathtt{Store} \; i \; j$,
$\mathtt{Input}\; i\; j$ and $\mathtt{Jzero}\; i \; \ell_0\;\ell_1$.

\begin{minipage}[b]{0.45\textwidth}
\begin{table}[H]
  \centering
  \begin{tabular}{lr}
    \textbf{Instruction} & \textbf{Effect} \\ \hline 
  $\mathtt{CST} \; 1$ & $A \gets 1$ \\ 
  $\mathtt{Load}$ & $A \gets R[1]$ \\
  $\mathtt{Buffer}$ & $B \gets R[1]$ \\
  $\mathtt{CST} \; 0$ & $A\gets 0$ \\ 
  $\mathtt{Load}$ & $A\gets R[0]$ \\
  $\op$ & $A \gets \op(R[0],R[1])$ \\
  $\mathtt{Buffer}$ & $B \gets \op{}( R[0], R[1])$ \\
  $\mathtt{CST} \; 0$ & $A\gets 0$ \\ 
  $\mathtt{Store}$ & $R[0]\gets \op{}(R[0],R[1])$ \\
\end{tabular}
\caption{Emulation 
of $\op$}
\label{tab:emulOP}
\end{table}
\begin{table}[H]
  \centering
  \begin{tabular}{lr}
    \textbf{Instruction} & \textbf{Effect} \\ \hline 
    $\mathtt{CST}\;i$ & $A\gets i$ \\
    $\mathtt{Load}$ & $A\gets R[i]$ \\
    $\mathtt{Output}$ & output $R[i]$ \\
  \end{tabular}
\caption{Emulation of $\mathtt{Output}\; i$}
\label{tab:emulOutput}
\end{table}  
\end{minipage}
\begin{minipage}[b]{0.45\textwidth}
\begin{table}[H]
  \centering
\begin{tabular}{lr}
    \textbf{Instruction} & \textbf{Effect} \\ \hline
  $\mathtt{CST} \; j$ & $A\gets j$ \\ 
  $\mathtt{Load}$ & $A\gets R[j]$\\ 
  $\mathtt{Buffer}$ & $B\gets R[j]$ \\
  $\mathtt{CST} \; i$ & $A\gets i$ \\
  $\mathtt{Load}$ & $A\gets R[i]$ \\ 
  $\mathtt{Store}$ & $R[R[i]]\gets R[j]$ \\
\end{tabular}
\caption{Emulation of $\mathtt{Store} \; i \; j$}
\label{tab:emulStore}
\end{table}
\begin{table}[H]
  \centering
  \begin{tabular}{lr}
    \textbf{Instruction} & \textbf{Effect} \\ \hline 
    $\mathtt{CST}\;j$ & $A\gets j$ \\
    $\mathtt{Load}$ & $A\gets R[j]$ \\
    $\mathtt{Input}$ & $A\gets I[R[j]]$ \\
    $\mathtt{Buffer}$ & $B\gets I[R[j]]$ \\
    $\mathtt{CST}\;i$ & $A\gets i$ \\
    $\mathtt{Store}$ & $R[i]\gets I[R[j]]$ \\
  \end{tabular}
\caption{Emulation of $\mathtt{Input}\; i\; j$}
\label{tab:emul_Input}
\end{table}  
\end{minipage}
\begin{table}[H]
  \centering
  \begin{tabular}{lr}
    \textbf{Instruction} & \textbf{Effect} \\ \hline 
    $\mathtt{CST}\;i$ & $A\gets i$ \\
    $\mathtt{Load}$ & $A\gets R[i]$ \\
    $\mathtt{Jzero}\;\ell_0\;\ell_1$ & 
    $\mathtt{if}\; R[i]=0\;\mathtt{then}\;\mathtt{goto}\;\ell_0\;\mathtt{else}\;\mathtt{goto}\;\ell_1$ \\
  \end{tabular}
\caption{Emulation of $\mathtt{Jzero}\; i \; \ell_0\;\ell_1$}
\label{tab:emulJzero}
\end{table} 
\end{proof}

\begin{proof}[Proof of item 2]
  This item is the easiest to prove: we can consider the special
  registers $A,B$ as integer variables and the sequence of work
  registers $R[0],R[1],\ldots$ as the array of integers $R[0\ldots]$
  and then any $AB$-program becomes a multi-memory program.
\end{proof}

\begin{proof}[Proof of item 3]
Similarly to a compiler that works in two ``passes'' with an
intermediate language, our simulation of array-instructions with
$R$-instructions is done in two steps with an intermediate instruction
set that we call ``elementary array-instructions'' depicted in
Table~\ref{tab:elementaryArrayInstructions} where $h,i,j,j_1,\ldots,
j_k$ are explicit integers, $\ell_0, \ell_1$ are valid indexes of
instructions and $\op$ is a $k$-ary operator from $\Op$.

\begin{table}[H]
  \centering
 \begin{tabular}{|lr|}
 \hline
 \textbf{Instruction} &  \textbf{Meaning} \\
 \hline
  $C_i\gets j$  & \\
  $C_i\gets N$  & \\
  $C_i\gets \op{}(C_{j_1}, \dots, C_{j_k})$  &\\
  $C_i\gets I[C_j]$  &\\
  $C_h\gets T_j[C_i]$  &\\
  $T_j[C_i]\gets C_h$  &\\
  $\mathtt{Jzero} \; C_i \; {\ell_0} \; \ell_1$  &$\mathtt{if}\; C_i=0\;\mathtt{then}\;\mathtt{goto}\;\ell_0\;\mathtt{else}\;\mathtt{goto}\;\ell_1$ \\
  $\mathtt{Output}$ $C_i$  &\\ 
  \hline
\end{tabular}

\caption{Elementary array-instructions}
\label{tab:elementaryArrayInstructions}
\end{table}

\paragraph{From array-instructions to elementary array-instructions.}
To achieve this first compilation step, the idea is to create a new integer variable $C_j$ for each sub-expression $\alpha'$ appearing in any of the $\mathtt{Op}$-expression $\alpha$ that occurs in the initial multi-memory RAM program.
So, any array-instruction will be simulated (replaced) by a sequence of elementary array-instructions.

As a representative example, the array-instruction $T_2[C_1*T_1[4]] \gets T_1[T_3[2]+N]$, which is of the form $T_j[\alpha]\gets\beta$, is simulated by the following sequence of elementary array-instructions:
\begin{itemize}
\item Storing $C_1*T_1[4]$ into $D_0$: 
$D_1\gets 4$ ; $D_2\gets T_1[D_1]$ ; $D_0\gets C_1*D_2$;
\item 
Storing $T_1[T_3[2]+N]$ into $D_3$: \\
$D_4\gets 2$ ; $D_5\gets T_3[D_4]$ ; $D_6\gets N$ ; $D_7\gets D_5 +D_6$ ; $D_3 \gets T_1[D_7]$;
\item Storing $T_1[T_3[2]+N]$ into  $T_2[C_1*T_1[4]]$ : 
$T_2[D_0]\gets D_3$.
\end{itemize}

Note that the resulting instructions have the required form: to see
it, rename each variable $D_j$ to $C_{j'}$ for some $j'$. This process
can be adapted for any array-instruction using $\Op$-expressions $\alpha$
and $\beta$ (see Table~\ref{table:instRAM3}) by adding a number of instructions and variables roughly
equal to the sum of the sizes of the expressions $\alpha$ and $\beta$.

\paragraph{From elementary array-instructions to $R$-instructions.}

Let us suppose that we want to translate a program with elementary array-instructions into an $R$-program. Let $C_0,\ldots,C_{v-1}$ be the list of variables of the array-program and $T_0,\ldots,T_{t-1}$ be its list of arrays, and finally let $k$ be an integer greater than the maximal arity of an $\op$ (which is at least 2 because we have $\mathtt{add}$).

The idea for building an $R$-configuration from an array-configuration
is to segment the memory in the following way: the $k$ first registers
$R[0], \dots, R[k-1]$ will be used for intermediate computations, the
$v$ next registers $R[k], \dots, R[k+v-1]$ will store the contents of
the variables, and finally the rest of the memory will be dedicated to
store the contents of the arrays.  More precisely, the register
$R[k+i]$ for $i<v$ will store the content of the variable $C_i$ and
$T_j[i]$ will be stored in the register $R[k+v+t*i+j]$.
Clearly, there is no interference between the roles of registers $R[y]$.

Since each integer variable $C_i$ is translated as $R[j]$ where $j$ is
the explicit value of $k+i$, each elementary array-instruction turns
trivially into an $R$-instruction, with the exception of the three
instructions $C_h \gets T_j[C_i]$, $T_j[C_i] \gets C_h$ and
$C_i\gets \op{}(C_{j_1}, \dots, C_{j_k})$.

The first two of these instructions can be translated as $R[k+h] \gets
R[k+v+t*R[k+i]+j]$ and $R[k+v+t*R[k+i]+j] \gets R[k+h]$, respectively,
which must be transformed. The two transformations are similar, so we
will detail the first one. By using the buffers $R[0],R[1]$ and the
addition operation (but not the multiplication!) we can transform the
assignment $R[k+h] \gets R[k+v+t*R[k+i]+j]$ into an equivalent
sequence of $R$-instructions as shown in the program of
Table~\ref{tab:emulArrayAccess} with $t+3$ instructions. (Note that
$h, i, j, k,v$ and $t$ are all explicit integers and therefore $t+3$,
$k+i$, $k+v+j$, $k+h$, etc., are also explicit integers whose value
can be computed before the RAM is launched.)

\begin{table}[H]
  \centering
  \begin{tabular}{clr}
    \textbf{\#}  & \textbf{Instruction} & \textbf{Effect} \\
  $\ell+0$  & $\mathtt{Move} \; 1 \; (k+i)$ & $R[1] \gets R[k+i]$        \\
  $\ell+1$  & $\mathtt{CST} \; 0 \; (k+v+j)$ & $R[0] \gets k+v+j$
  \\
  $\ell+2$  & $\mathtt{add}$ & $R[0]\gets R[0]+R[1]$, meaning $R[0] \gets k+v+j+R[k+i]$        \\
  $\ell+3$  & $\mathtt{add}$ & $R[0]\gets R[0]+R[1]$, meaning $R[0] \gets k+v+j+2*R[k+i]$        \\
    $\dots$  & $\dots$ & $\dots$         \\
  $\ell+t+1$  & $\mathtt{add}$ & $R[0]\gets R[0]+R[1]$, meaning $R[0] \gets k+v+j+t*R[k+i]$        \\
  $\ell+t+2$  & $\mathtt{Load} \; (k+h) \; 0$ & $R[k+h]\gets R[R[0]]$, or $R[k+h] \gets R[k+v+j+t*R[k+i]]$        \\
\end{tabular}
\caption{Emulation of $C_h  \gets  T_j[C_i]$}
\label{tab:emulArrayAccess}
\end{table}

For the $\op$ operation $C_i \gets \op(C_{j_1}, \dots, C_{j_k})$, the
translation is simpler: we first store each $C_{j_l}$ into
$R[l-1]$ with $\mathtt{Move} \; (l-1) \; (k+j_l)$, then we apply the 
$R$-instruction $\op$, which computes $R[0]\gets \op(R[0], \dots, R[k-1])$; 
finally, we store back the result into $C_i$ with
 $\mathtt{Move}\;(k+i) \; 0$.

This completes the proof of item 3. Lemma~\ref{lemma:lockstepEquivInst} is proved.
\end{proof}

\subsection{Richer instruction set}

At the moment we saw three instructions sets and showed that they are
equivalent. Since they are all equivalent, in the following sections,
our procedures will be written in the third instruction set, which is
the easiest to work with.

While this multi-memory model is the closest to a standard language
among our three instruction sets, it still lacks a lot of the
constructs available in most programming languages such as loops,
functions, etc. This section is dedicated to show that we can easily
compile a language with those features to our assembly.

Most of the techniques we use here to enrich the instruction set are
very well known and very well studied for compilers. While they are
interesting in themselves there are not the novelty of the paper,
which is why we will only give the intuition of what is happening and
point the reader to the existing literature.

\paragraph{From indexed programs to labeled programs.}

At this point in the paper, a program is still viewed as a sequence
$I_0, \dots, I_{r-1}$ of instructions where the index of each
instruction is important as a $\mathtt{Jzero}$ instruction will jump
to a given index. We can replace this scheme with a label scheme and a
program becomes a set of labeled instructions $I_a, I_b, \dots$ where
each instruction has one or several labels and each instruction $e$
which is not $\mathtt{Jzero}$ has a next label $next(e)$ (determining
what instruction should be executed after).  We also have a special
label $\mathtt{end}$ for the label $next(e)$ where $e$ is the final instruction
and a special label  $\mathtt{begin}$ to determine the first instruction.  Such a
labeled program can easily be translated into an indexed program.

\paragraph{Composing programs.}
Once we have labeled programs, it is easy to talk about the
composition of two programs $P_1$ and $P_2$: we relabel the labels of
$P_1$ and $P_2$ to make sure that they are unique in both, with the
$\mathtt{end}$ of $P_1$ being the $\mathtt{begin}$ instruction of
$P_2$; then we replace the output instructions of $P_1$ by write
instructions in an array $O$ and the input instructions of~$P_2$ are
replaced by read instructions of the said array $O$.

\paragraph{Equality conditions.}
Using labeled programs with array-instructions, we can allow the equality
condition $\alpha=\beta$ for all expressions $\alpha,\beta$,
i.e.\ allow a new instruction
$\mathtt{Jequal}\;\alpha\;\beta\;\ell_0\;\ell_1$ (extending the
$\mathtt{Jzero}$ instruction), which means
``$\mathtt{if}\;\alpha=\beta\;\mathtt{then}\;\mathtt{goto}\;\ell_0\;\mathtt{else}\;\mathtt{goto}\;\ell_1$''.
Actually, we claim that the instruction
$\mathtt{Jequal}\;\alpha\;\beta\;\ell_0\;\ell_1$ is equivalent to the
following sequence of array-instructions using the new array
$\mathtt{Eq}$ with suggestive name:
\begin{center}
$\mathtt{Eq}[\alpha]\gets 1$ ;
$\mathtt{Eq}[\beta]\gets 0$ ;
$\mathtt{Jzero}\;\mathtt{Eq}[\alpha]\;\ell_0\;\ell_1$
\end{center}
\begin{proof}
Clearly, if $\alpha=\beta$ then the assignment $\mathtt{Eq}[\beta]\gets 0$ gives 
$\mathtt{Eq}[\alpha]=\mathtt{Eq}[\beta]=0$; otherwise, we still have 
$\mathtt{Eq}[\alpha]=1\neq 0$. In other words, we have the equivalence $\alpha=\beta \iff \mathtt{Eq}[\alpha]=0$.
\end{proof}

\paragraph{Advanced Boolean manipulation.}
One of the features of a programming language is to be able to perform
complex conditions such as $(a=b) \lor (b\neq c \land b\neq d)$. At
the moment, the only condition allowed is the equality comparison but
it is easy to convert such a comparison $\alpha=\beta$ into an integer $c$
which is 0 when $\alpha \neq \beta$ and 1 when $\alpha=\beta$. 
This also allows manipulation of Boolean logic: the negation of
$\text{cond}_1$ and the conjunction $\text{cond}_1 \land
\text{cond}_2$ can be transformed into $c_1=0$ and $c_1+c_2=2$,
respectively, where $c_i$ is the integer associated with
$\text{cond}_i$. This allows us to manipulate any Boolean combinations
of conditions (the basic conditions only allow equalities for the
moment, inequalities $<$ will be introduced in Subsection~\ref{sec:sumdiffproduct}).

\paragraph{Basic constructions: if-then-else, loops, etc.}
Our procedures and programs will freely use the basic constructs and
syntax of a programming language. Once you have moved from an indexed program
to a labeled program, it is easy to translate an ``$\mathtt{if} \;
\alpha=\beta \; \mathtt{then} \; \text{body}_0 \; \mathtt{else} \;
\text{body}_1 $'' into a $\mathtt{Jequal} \; \alpha \; \beta \; \ell_0
\; \ell_1$ where $\ell_i$ points to the first instruction of
$\text{body}_i$ and $\text{body}_0$ finishes with a jump to the end of
$\text{body}_1$.

The reader can also easily check that a loop ``$\mathtt{while} \;
\alpha=\beta \; \mathtt{do}\; \text{body} $'' can be translated with
one conditional test at the beginning, then the translation of the
body, then one unconditional test to go back to the conditional test.
In the same manner, we can have the for-loop construct. We will use
the following format ``$\mathtt{for} \; var \; \mathtt{from}\; \alpha \; \mathtt{to} \; 
\beta \; \mathtt{do} \; \text{body} \; \mathtt{done}$'' meaning that
$var$ will range from $\alpha$ up to $\beta$. At this stage, this
construct is only defined for $\alpha\leq \beta$.

\paragraph{Dynamic arrays of bounded size.}
In our simulation of array-instructions with $R$-instructions we use a
constant number of arrays with no bound on the size. As in most actual
programming languages, it is often useful to be able to allocate
arrays dynamically (i.e.\ during the execution).

In our RAM model, we can do this with the addition operation using the
following memory allocator scheme. The idea to allocate dynamic arrays
uses one static array called \texttt{DATA} storing the data of all the
dynamic arrays and one variable \texttt{nbCellsUsed} storing the
number of cells used for dynamic arrays in \texttt{DATA}. To allocate
a dynamic array of size $k$ we increase \texttt{nbCellsUsed} by $k$
and return the value $p=$ \texttt{nbCellsUsed}$-k$. The value $p$
identifies the array and is a sort of ``pointer'' towards the memory
of the newly allocated array, very similarly to what happens in
languages such as C. Each read or write access to $i$-th cell of the
array $T$ is then replaced with an access to $\mathtt{DATA}[T+i]$.

Note that this creates two kinds of arrays, the static ones (available
in array instructions) and the dynamic ones but for the sake of
simplicity we will denote by $T[i]$ the $i$-th cell of $T$ no matter
whether $T$ is a value corresponding to a dynamically allocated array
or a static one.

\paragraph{Two-dimensional arrays.}
The allocation scheme presented above is very useful to allocate many
arrays (i.e.\ more than a constant number) as long as the total
allocated size is $O(N)$. For instance, using this we can
allocate a dynamic array $T$ of size $R$ and store within each $T[i]$
an array of size $C$. This creates a two-dimensional array of
dimension $R\times C$. Note that this is only allowed if the total
allocated memory is $O(N)$ and thus we need to have $R\times C=O(N)$.

Note that we can also generalize this method to any fixed number of
dimensions as long as the total number of cells is $O(N)$.



\paragraph{Functions.}
Most programming languages are equipped with functions. It is well
known that one can ``de-functionalize'' a program with a stack. Here,
we informally explain one such way to de-functionalize programs so
that we can notice that the addition is all we need on a RAM to make
this ``de-functionalization'' work.

Each active function call will have its memory stored in an array $S$
(the stack) and we will have a variable $F$ (the frame pointer)
pointing to the point in the array $S$ where is the memory for the
current function call. The memory used by one function call will span
cells $S[F], \dots, S[F+v]$, where $v$ is the number of registers
used. The first cell, $S[F]$, designates the point in the program
where the function call originated from (so we can return to it), the
cell $S[F+1]$ will store the value of $F$ for the caller (so we can
return to it with $F$ set to the correct value), $S[F+2]$ will store
the number of registers used by the current call. When calling a
function, we need to compute $F=F+S[F+2]$, then we set $S[F]$ to mark
the current function call as the caller and finally, we jump to the
function. When returning from a function call we determine the caller
from $S[F]$ and then we set $F=S[F+1]$.

Since the memory of the function called is not overwritten until we
make a new function call, we can simulate a ``return'' instruction by
reading the value of the variables of the function we have just
called.  Note that our scheme allows for function to have arrays and
we will often handle integers spanning several registers (that we
later call ``polynomial'' integers) and will freely use \Return
$\arr{x}$ for $\arr{x}$ an array when the size of the array is known
in advance.



\subsection{RAM models with a different space usage}\label{subsec:RAMwithLargerSpace}

We have seen three different RAM models which are all equally
expressive when the addition is present (equally expressive in the
sense that they define the same complexity classes). These three
models have three different instructions sets that all manipulate a
memory that is zero-initialized, which leads to our first question:

\textit{Does the expressiveness of the model change if we suppose that the memory is
  not zero-initialized?} 
  
 A second property of our model is that it depends on a parameter~$N$ and we suppose that all registers can only contain a value $O(N)$ which, in turn, limits the available memory to
$O(N)$. This leads to our second natural question:

\textit{Can we change our model to take more memory into account? 
And if so, how does that impact expressiveness?}

\paragraph{The initialization problem.}
As in the original RAM model~\cite{CookR73,AhoHU74}, we have assumed
that all registers are initialized to zero before the start of the
computation. However, as is the rule in the Python language, it is a
principle of good programming that any memory variable (integer,
array, etc.) be explicitly initialized as soon as it is declared/used.

This requirement can be satisfied for any program of the ``linear
time'' complexity classes that interest us (see
Subsection~\ref{subsec:linComplexityClasses}) provided that the
addition is available. Indeed, such a program uses only a linear
number of memory registers $R[0],R[1],\ldots,R[c \times N]$, for a
constant integer $c$. Therefore, the program can be transformed into
an equivalent program of the same complexity by adding at the start of
the program the explicit initialization. For this, we can use a program
very similar to the illustrative program given in Table~\ref{table:exampleProgram} but where, instead of setting and reading $R[N]$, we set and read $R[N\times c]$ where $N\times c$ is
$N$ added $c$ times to $0$.

If we suppose that when reading uninitialized memory it might return
any value, another way to deal with the initialization problem is the
``lazy array evaluation technique'' described below.

\paragraph{RAM model with more memory.}
\label{sec:more_memory}
All the RAM models presented above suffer from the same limit on space usage. 
Since the integers stored in the registers are $O(N)$ and all native arrays are one-dimensional, it means that the memory accessed by any program is always $O(N)$. 

One way to overcome this limit is to consider that we can have
$k$-dimensional arrays available, for some fixed integer $k$, and in
this case a $k$-dimensional array would allow $O(N^k)$ registers.
Several recent papers about queries in logic and databases, see
e.g.\ ~\cite{DurandSS14, Segoufin14, BerkholzKS17,BerkholzKS17b},
present algorithms inside such an extended RAM model. Significantly,
Berkholz, Keppeler and Schweikardt write in~\cite{BerkholzKS17}: ``A
further assumption that is unproblematic within the RAM-model, but
unrealistic for real-world computers, is that for every fixed
dimension $d\in \mathbb{N}_{>1}$ we have available an unbounded number
of $d$-ary arrays $\mathbf{A}$ such that for given
$(n_1,\dots,n_d)\in\mathbb{N}^d$ the entry $\mathbf{A}[n_1,\dots,n_d]$
at position $(n_1,\dots,n_d)$ can be accessed in constant time.''

We know by the space hierarchy theorem that a machine with a memory of
size $O(N^k)$ cannot be emulated (in general) by a machine with only
$O(N)$ memory for any $k>1$. This means that the models with or
without $k$-dimensional arrays are incomparable and, more generally,
for all $k$ the models differ. If we really want to emulate ``true''
$k$-dimensional arrays we need to set the parameter $N'$ to $N^k$ but
this does not provide much hindsight in the problem at hand.  An
interesting case is when the program has access to $k$-dimensional
arrays but actually only uses only $O(N)$ memory registers. Consider
first the problem of initializing $k$-dimensional arrays.

\paragraph{The ``lazy array evaluation technique''}(according to Moret and Shapiro's book~\cite{MoretS1991}, Subsection 4.1, Program 4.2): 
Let $P$ be a program using $k$-dimensional arrays \emph{assumed to be initialized to zero for free}.
We want to construct a program $P'$ simulating $P$ with $k$-dimensional arrays, for the same $k$, which are  \emph{not initialized}.
For each $k$-dimensional array $\arr{A}[1..N]\ldots [1..N]$ of $P$, we create a variable $\var{countA}$ counting the number of initialized cells $(i_1,\ldots, i_k)$ of $\arr{A}$,
a $k$-dimensional array $\arr{rankInitA}[1..N]\ldots [1..N]$ and its ``inverse'' one-dimensional array 
$\arr{rankInitA}^{-1}$, called \emph{lazy arrays}, such that $\arr{rankInitA}[i_1]\ldots [i_k]=y$ and, conversely, 
$\arr{rankInitA}^{-1}[y]=(i_1,\ldots, i_k)$ if  $\arr{A}[i_1]\ldots [i_k]$ 
is the $y$th initialized cell of $\arr{A}$. 
The important thing is to maintain the following equivalence, called equivalence~$(E)$, for each 
$(i_1,\ldots, i_k)\in [1,N]^k$:
\begin{center}
the cell $\arr{A}[i_1]\ldots [i_k]$ has been given a value, i.e.\ has been initialized, \emph{iff}
$\arr{rankInitA}^{-1}[\arr{rankInitA}[i_1]\ldots [i_k]]= (i_1,\ldots, i_k)$, 
for $\arr{rankInitA}[i_1]\ldots [i_k]\in [1,\var{countA}]$.
\end{center}
This gives the following procedures (freely inspired by Program 4.2
in~\cite{MoretS1991}), which we present, to simplify the notation,
with $k=2$. Also, we omit mentioning $\arr{A}$ in the
identifiers of the variable $\var{count}$, of the arrays $\arr{rankInit}$ and
$\arr{rankInit}^{-1}$, and in the procedures $\fun{initCount}$,
$\fun{initCell}$, $\fun{store}$ and $\fun{fetch}$ corresponding to array $\arr{A}$.

\begin{minipage}{\almosttextwidth}

  \begin{algorithm}[H]
    \caption{Lazy evaluation procedures for a 2-dimensional array $\arr{A}$}
    
\begin{algorithmic}[1]
\Procedure{initCount}{\mbox{}}
  \State $\var{count} \gets 0$ 
\EndProcedure
\end{algorithmic}
      
 \vspace{1em}
    
\begin{algorithmic}[1] 
  \Procedure{initCell}{$\var{i},\var{j}$}
  \LComment{answers the question: ``Has cell $\arr{A}[\var{i}][\var{j}] $ been initialized?''}
  \State \Return $(1\leq \arr{rankInit}[\var{i}][\var{j}]\leq \var{count}) \; \mathtt{and} \;
  (\arr{rankInit}^{-1}[\arr{rankInit}[\var{i}][\var{j}]]= (\var{i},\var{j}))$
\EndProcedure
\end{algorithmic}
      
 \vspace{1em}
      
\begin{algorithmic}[1]
\Procedure{store}{$\var{i},\var{j}, \var{x}$}
 \LComment{runs $\arr{A}[\var{i}][\var{j}] \gets \var{x}$ and validates that cell $(\var{i},\var{j})$ is now initialized}
   \State $\arr{A}[\var{i}][\var{j}] \gets \var{x}$
   \If {$\mathtt{not}~\fun{initCell}(\var{i},\var{j})$}
      \State $\var{count} \gets \var{count}+1$
      \State $\arr{rankInit}[\var{i}][\var{j}] \gets \var{count}$
      \State $\arr{rankInit}^{-1}[\var{count}] \gets (\var{i},\var{j})$ 
      \Comment{now, $\fun{initCell}(\var{i},\var{j})=\mathtt{true}$}
   \EndIf 
 \EndProcedure
 \end{algorithmic}
\vspace{1em}

\begin{algorithmic}[1]
\Procedure{fetch}{$\var{i},\var{j}$}
\LComment{returns $\arr{A}[\var{i}][\var{j}]$ if cell $(\var{i},\var{j})$ has been initialized, and $0$ otherwise}
\If{$\fun{initCell}(\var{i},\var{j})$}
  \State \Return $\arr{A}[\var{i}][\var{j}]$
\Else
    \State \Return $0$
\EndIf
\EndProcedure
\end{algorithmic}
 
  \end{algorithm}
\end{minipage}

\begin{proof} Lines 3-6 of the procedure $\fun{store}$, with the procedures $\fun{initCell}$ and $\fun{initCount}$, justify that the function $(i,j)\mapsto \arr{rankInit}[\var{i}][\var{j}]$ is always a bijection from the set of ``initialized'' pairs  $(i,j)\in[1,N]^2$, i.e.\ for which $\fun{initCell}(\var{i},\var{j})$ is true, to the integer interval $[1,\var{count}]$, and 
$y \mapsto \arr{rankInit}^{-1}[y]$, for $y\in[1,\var{count}]$, is its inverse bijection. 
The fact that the set of images $\arr{rankInit}^{-1}([1,\var{count}])$ is equal to the set of ``initialized'' pairs $(i,j)\in[1,N]^2$ justifies that the equivalence~$(E)$ seen above is true for all $(i,j)\in[1,N]^2$.

\medskip
As a consequence, in the general case, a program $P$ using $k$-dimensional arrays 
assumed to be initialized to zero for free, is simulated by a program $P'$ constructed from $P$ as follows:
\begin{itemize}
\item put the procedure call $\fun{initCount}(\mbox{})$ (i.e.\ the assignment $\var{count}\gets 0$), for each $k$-dimensional array $\arr{A}$ of $P$, at the beginning of the program $P'$;
\item replace each \emph{assignment} of the form $\arr{A}[t_1]\ldots[t_k] \gets t$ (resp.\ each \emph{argument} of the form $\arr{A}[t_1]\ldots [t_k]$) in $P$ by the call instruction 
$\fun{store}(t_1,\ldots,t_k,t)$ 
(resp.\ by the term \linebreak
$\fun{fetch}(t_1,\ldots,t_k)$) corresponding to $\arr{A}$.
\end{itemize}

Now, it is easy to verify that if the RAM program $P$ using $k$-dimensional arrays $N\times \cdots \times N$, assumed to be initialized to zero for free, runs in time $O(T)$, then its simulating program $P'$ using only \emph{uninitialized} $k$-dimensional (and one-dimensional) arrays $N\times \cdots \times N$, including the additional \emph{lazy arrays}, runs also in time $O(T)$.
In conclusion, the lazy array evaluation technique (to simulate multidimensional arrays initialized to zero for free) does not lead to a real complexity overhead: time and space are only multiplied by a (small) constant.
\end{proof}


\paragraph{Multidimensional arrays are equivalent to 2-dimensional arrays.}
We now give a (very) partial answer to our second question: how does the expressiveness of the RAM model depend on the dimensions of the allowed arrays?
The following theorem states that the RAM model with arrays of arbitrary dimensions is equivalent to the RAM model using only 2-dimensional arrays, and even, more strongly, using arrays of dimensions 
$N \times N^\epsilon$, for any $\epsilon>0$.

\begin{theorem}
  \label{thm:mem}
For any integers $k>1$ and $d\geq 1$, a RAM model equipped with the basic arithmetic operations and $k$-dimensional arrays but using only $O(N)$ registers and time $O(t(N))$, can be simulated in time $O(t(N))$ and using $O(N)$ registers by a RAM model equipped with the same set of operations but using only $1$-dimensional arrays and one 2-dimensional array 
$N \times  \lf N^{1/d} \rf$. 
\end{theorem}

\begin{proof}
For the sake of simplicity, we will suppose that we want to emulate
one $k$-dimensional array $N\times\dots\times N$.  Indeed, if we have
multiple $k$-dimensional arrays $\arr{T}_0, \dots, \arr{T}_{\ell-1}$,
each of dimensions $cN\times\dots\times cN$, for a fixed integer $c$,
we can replace each $k$-dimensional key $v_1,\dots,v_k$ ($v_j<cN$) of
$\arr{T}_i$ with one $(2k+1)$-dimensional key
$i,u_1,u_2,\dots,u_{2k-1},u_{2k}$, with $u_{2j-1}\coloneqq v_j \divop
c$ and $u_{2j}\coloneqq v_j \modop c$: the arrays $\arr{T}_0, \dots,
\arr{T}_{\ell-1}$ are emulated by one $(2k+1)$-dimensional
array~$\arr{T}$ of dimensions $N\times\dots\times N$ such that
$\arr{T}[i][u_1][u_2]\dots[u_{2k-1}][u_{2k}]=\arr{T}_i[v_1]\dots[v_k]$,
assuming \linebreak $N\geq \mathtt{max}(c,\ell)$; thus, the dimension
$k$ is increased to $2k+1$. Note that, in the $2k+1$, the factor $2$
comes from reducing $cN$ to $N$ while the $1$ comes from the reduction
of $\ell$ arrays to $1$.
  
Now, the general idea of the emulation of a $k$-dimensional
$N\times\dots\times N$ array $\arr{T}$ with a 2-dimensional array
$cN\times \lf N^{1/d} \rf$ (for some fixed integer $c$) called
$\arr{node}$ and a one-dimensional array called~$\arr{A}$ consists in
representing all the data in all the $k$-dimensional keys as paths in
a tree where the leaves store the values of the original array
$\arr{T}$.  Each node in the tree will have an arity bounded by $\lf
N^{1/d} \rf$ and all leaves will be at depth $\leq 2dk$, which is
possible by the inequality\footnote{Indeed, for $r=\lf N^{1/d} \rf$
and $N\geq 2^d$, we have $N<(r+1)^d<r^{2d}$ with $r\geq 2$, and
therefore $\lf N^{1/d} \rf^{2d} > N$.}  $\lf N^{1/d} \rf^{2dk} \geq
N^k$. Since we only use $O(N)$ registers in~$\arr{T}$, only $O(N)$
leaves and thus $O(N)$ nodes will be present in that tree.

Now, to find the leaf associated with a key $(v_1, \dots, v_k)\in [0,N[^k$, we transform the sequence 
$(v_1,\dots, v_k)$, seen as the sequence of digits of an integer in base $N$, into the sequence 
$(u_1, \dots, u_{k'})\in [0,\var{B}[^{k'}$ representing the same integer in base 
$\var{B}\coloneqq  \lf N^{1/d} \rf$. 
  (Note that $k' < 2dk$.)
Now, we can access the value stored in the leaf corresponding to $(u_1, \dots, u_{k'})$ 
by starting from the root $n_0 = 0$ and repetitively using $n_{i} = \arr{node}[n_{i-1}][u_i]$, 
for $1\leq i \leq k'$. Thus, the time of this access is $O(k')=O(kd)$: it is a constant time.
When the value/node returned by $\arr{node}[n_i][u_i]$ is undefined, we create a new node in the tree and set $\arr{node}[n_i][u_i]$ to correspond to this node.
Finally, the value ``stored'' at the leaf $n_{k'}$ is $\arr{A}[n_{k'}]=\arr{T}[v_1]\dots[v_k]$.

Note that the values (nodes) $n_i$ are less than $cN$, for some fixed integer $c$, and therefore the arrays used by the simulation are $\arr{node}[0..cN-1][0..\var{B}-1]$ and $\arr{A}[0..cN-1]$.
Then, replace (encode) the $cN\times \var{B}$ array 
$\arr{node}$ by the $N\times \lf N^{2/d} \rf$ array~$\arr{E}$ defined by 
$\arr{E}[x \divop c][c\times y +x \modop c]\coloneqq \arr{node}[x][y]$, for $x<cN$ and $y< \var{B}$: 
this is possible because we have $\lf N^{1/d} \rf \geq c$, for $N\geq c^d$, which implies
$c \times \var{B} = c \times \lf N^{1/d} \rf \leq \lf N^{1/d} \rf^2 \leq \lf N^{2/d} \rf$. The theorem is obtained by replacing $d$ by $2d$: this means that we set $\var{B}\coloneqq  \lf N^{1/(2d)} \rf$ instead of 
$\var{B}\coloneqq  \lf N^{1/d} \rf$.
\end{proof}

\begin{remark}[Essential point]
Note that it is assumed in the proof of Theorem~\ref{thm:mem} that the RAM has access to the basic arithmetic operations (essentially, computation of $\lf N^{1/d} \rf$, division by a constant $c$ and base change) but since the RAM model with $k$-dimensional arrays is more powerful than the model we study in the rest of this paper (which allows to execute those operations in constant time with a linear-time preprocessing, as we will establish), we know that \emph{a fortiori} those operations (for example, the conversion from base~$N$ to base $\var{B}$, used in the above simulation) can be computed in constant time with a linear-time preprocessing in the more liberal model.
 \end{remark}

From Theorem~\ref{thm:mem}, we can deduce that two models using $O(N)$
registers with $k$-dimensional or $k'$-dimensional arrays are
equivalent (for $k>k'>1$) but there remains the question of whether
the RAM model with $k$-dimensional arrays but using only $O(N)$
registers is equivalent to the RAM model with only $1$-dimensional
arrays.

\medskip
To our knowledge, even the following problem (asked by Luc Segoufin in
a personal communication) is open.

\begin{openpb}\label{conj:problemeDeLuc}
Given a set $S$ of $N$ pairs of integers $(a,b)$ with $a,b<N$, can we
preprocess~$S$ in $O(N)$ time and memory, using only $1$-dimensional
arrays, so that when given $(c,d)\in [0,N[^2$ we can answer in
constant time whether $(c,d)\in S$?
\end{openpb}

Note that this problem is easily solved with a Boolean $2$-dimensional array $N\times N$. 
Conversely, if we can positively solve the open problem~\ref{conj:problemeDeLuc},
this does not necessarily mean that we can have 
$k$-dimensional arrays: arrays are dynamic (with interleaving of reads and writes) and
store integer values, not just Boolean information.

%% file: 3_Classes.tex
\section{Complexity classes induced by the RAM model}\label{sec:RAMcomplexityclasses}

Given a RAM machine $\RAM{}$, we can define the class $TIME(f(N))$ as
the class of problems for which there exists a program running on
$\RAM{}$ and solving the problem in time $O(f(N))$ (where $N$ is the
parameter of the RAM, that is both the value or the size of the input and the number such that the values stored within registers must be less than $cN$, for a fixed $c$).

This definition of $TIME$ depends on the specific RAM we start
with. For instance, if $\Op$ includes the division, it is pretty clear
that the problem of dividing two integers of value $O(N)$ can be done
in time $O(1)$ but if $\Op$ only includes the addition, it is not
clear that division is in $TIME(1)$.

In Subsection~\ref{subsec:emul}, we define a notion of ``faithful simulation'' of a RAM by another RAM which preserves their ``dynamic'' complexity in a very precise sense, see Lemma~\ref{lemma: faithful simul}. 
This allows to define a very strict notion of ``equivalence'' between classes of RAMs.
Building on that, we will introduce in Subsection
~\ref{subsec:linComplexityClasses}, Definition~\ref{def:equivOp,equivRAM}, a second, weaker, criterion for equivalence that states that two RAM models are equivalent when they compute the same operations in constant time \emph{after a linear-time initialization}.

We believe that this second equivalence criterion is the ``right'' one
to study complexity classes that interest us (i.e.\ where a $O(N)$
computation is allowed). The criterion is strong enough to show that
two equivalent RAM models induce the same class $TIME(N)$ (and
similarly for many other complexity classes) but will also allow us to
show that a model simply equipped with addition is equivalent to a
model with subtraction, multiplication, division, etc.

\subsection{Emulation, faithful simulation and equivalence of classes of RAMs}\label{subsec:emul}

The notions that we used in the previous section to show that our
three instruction sets are equivalent are intuitive notions but, for
the sake of rigor, we introduce below the technical notion of
``emulation'' as well as the ``faithful simulation'' and ``equivalence''
relations and their properties (Definition~\ref{def:emul},
Lemma~\ref{lemma: faithful simul} and
Definition~\ref{def:equivalence}). First, let us introduce some useful
conventions and notation.

\begin{convention}[registers of a RAM]
  We consider that the \emph{registers} of an $AB$-RAM (resp.\ a
  multi-memory RAM) are $A,B$ and the $R[j]$ (resp.\ the integer
  variables $C_i$ and the array elements $T_j[i]$).
\end{convention}

\begin{notation}[value in a configuration]
Let $\mathcal{C}$ be any configuration ($AB$-configuration,
$R$-configuration, or array-configuration). The value associated with
a register or an $\mathtt{Op}$-expression $\rho$ in $\mathcal{C}$ is
denoted $\mathtt{value}(\rho,\mathcal{C})$.
\end{notation}

\begin{convention}[Output instructions]
The instructions $\mathtt{Output}$ (meaning output $A$),
$\mathtt{Output}\; i$ (meaning output $R[i]$) and $\mathtt{Output}\;
\alpha$ (for an $\mathtt{Op}$-expression $\alpha$) are called
\emph{output instructions}. Each one means ``output the value
$\mathtt{value}(\rho,\mathcal{C})$'' for $\mathcal{C}$ the current
configuration and $\rho$ the argument of output (i.e.\ $A$, $R[i]$ or
$\alpha$ depending on the type of RAM).

\end{convention}

\begin{notation}[transitions between configurations] We write 
$\mathcal{C}\vdash^{M}_{\mathfrak{I},O} \mathcal{C}'$
  (resp.\ $\mathcal{C}\vdash^{M,j}_{\mathfrak{I},O} \mathcal{C}'$) to
  signify that for an input $\mathfrak{I}$, a configuration
  $\mathcal{C}$ is transformed into a configuration $\mathcal{C'}$ by
  one computation step (resp.\ $j$ computation steps) of the RAM $M$,
  while running with the input $\mathfrak{I}$ and writing the output
  $O$.
\end{notation}

We are now ready to give the formal definition of a faithful
simulation:

\begin{definition}[emulation, faithful simulation]\label{def:emul}
Let $M_1$ and $M_2$ be two RAMs (among $AB$-RAMs, $R$-RAMs, or
multi-memory RAMs).  An \emph{emulation} $E$ of $M_1$ by $M_2$ is a
triplet of maps $(\mathcal{I} \mapsto \mathcal{I}^E,\; \rho \mapsto
\rho^E,\; \mathcal{C} \mapsto \mathcal{C}^E)$ where 
\begin{itemize}
\item the map $\mathcal{I} \mapsto \mathcal{I}^E$ associates each
  instruction $\mathcal{I}$ of $M_1$ a sequence of at most $k$
  instructions $\mathcal{I}^E$ of $M_2$, for some constant $k$,
\item the map $\rho \mapsto \rho^E$ associates to each register $\rho$
  of $M_1$ a register $\rho^E$ of $M_2$,
\item the map $\mathcal{C} \mapsto \mathcal{C}^E$ associates to each
  configuration $\mathcal{C}$ of $M_1$ a configuration $\mathcal{C}^E$
  of $M_2$.
\end{itemize}
satisfying the following properties:
\begin{enumerate}
\item for each configuration $\mathcal{C}$ and each register $\rho$ of
  $M_1$, we have $\mathtt{value}(\rho,
  \mathcal{C})=\mathtt{value}(\rho^E, \mathcal{C}^E)$;
\item if $\mathcal{C}'$ is the configuration obtained from any
  configuration $\mathcal{C}$ of $M_1$ by an instruction $\mathcal{I}$
  (of~$M_1$) for an input $\mathfrak{I}$, then $\mathcal{C}'^E$ is
  obtained from the configuration $\mathcal{C}^E$ by the sequence of
  instructions $\mathcal{I}^E$ (of $M_2$) for the same input
  $\mathfrak{I}$;
\item moreover, if $\mathcal{I}$ is an output instruction, meaning
  ``output $\rho$'', which executes on a configuration $\mathcal{C}$
  of $M_1$, then the sequence of instructions $\mathcal{I}^E$ outputs
  $\rho^E$ from the configuration $\mathcal{C}^E$, which gives the
  same outputted value $\mathtt{value}(\rho^E,
  \mathcal{C}^E)=\mathtt{value}(\rho, \mathcal{C})$; conversely, if
  $\mathcal{I}$ is not an output instruction, then the sequence of
  instructions $\mathcal{I}^E$ outputs no value;
\item if $\mathcal{C}$ is the initial (resp.\ a final) configuration of
  $M_1$, then $\mathcal{C}^E$ is the initial (resp.\ a final)
  configuration of $M_2$.
\end{enumerate}
When properties (1-4) are satisfied we say that $M_2$ \emph{faithfully
simulates} $M_1$ by the emulation~$E$.

\end{definition}

\begin{remark} 
Our notion of faithful simulation is a very precise variant, for RAMs, of the different notions of ``real-time simulation'' (or real-time reducibility) in the literature for different machine models:  see, e.g.,~\cite{Schonhage80}, or~\cite{DexterDG97} who calls it a ``lock-step simulation''.
\end{remark}

The following result which derives from Definition~\ref{def:emul}
establishes that an emulation of a RAM~$M_1$ by another RAM $M_2$
transforms a computation of $M_1$ into a computation of $M_2$ which
computes exactly the \emph{same thing} (the same outputs for the same
inputs) and in the \emph{same way} (the number of steps between the
start and the first output, or between two outputs is the same, up to
the constant factor $k$).

\begin{lemma}[faithful simulation of a RAM]\label{lemma: faithful simul}
Let $M_1,M_2$ be two RAMs such that $M_2$ faithfully simulates $M_1$
by an emulation~$E$.
 
\begin{enumerate}
\item Then, for any input~$\mathfrak{I}$, if the computation of $M_1$
  on input~$\mathfrak{I}$ stops in time $\tau$ and outputs~$O$
  (written $\mathcal{C}_0 \vdash^{M_1,\tau}_{\mathfrak{I},O}
  \mathcal{C}_{\tau}$, where $\mathcal{C}_0$
  (resp.\ $\mathcal{C}_{\tau}$) is the initial (resp. a final)
  configuration of $M_1$), then on the same input~$\mathfrak{I}$ the
  RAM $M_2$ stops in time~$\tau'\leq k*\tau$ (where $k$ is the
  constant integer mentioned in Definition~\ref{def:emul}), which is
  $O(\tau)$, and produces the same output~$O$, which is written
  $\mathcal{C}^E_0 \vdash^{M_2,\tau'}_{\mathfrak{I},O}
  \mathcal{C}^E_{\tau}$.

\item Moreover, if the computation of $M_1$ can be divided into intervals of
  durations $\tau_1,\tau_2,\ldots,\tau_\ell$, so that within the $j$th
  computation interval, $M_1$ reads $\mathfrak{I}_j$ and outputs
  $O_j$, which is written \linebreak $\mathcal{C}_0
  \vdash^{M_1,\tau_1}_{\mathfrak{I}_1,O_1} \mathcal{C}_1
  \vdash^{M_1,\tau_2}_{\mathfrak{I}_2,O_2} \dots
  \vdash^{M_1,\tau_\ell}_{\mathfrak{I}_\ell,O_\ell}
  \mathcal{C}_{\ell}$, where $\mathcal{C}_0$
  (resp.\ $\mathcal{C}_{\ell}$) is the initial (resp. a final)
  configuration of $M_1$, then we get $\mathcal{C}^E_0
  \vdash^{M_2,\tau'_1}_{\mathfrak{I}_1,O_1} \mathcal{C}^E_1
  \vdash^{M_2,\tau'_2}_{\mathfrak{I}_2,O_2} \dots
  \vdash^{M_2,\tau'_\ell}_{\mathfrak{I}_\ell,O_\ell}
  \mathcal{C}^E_{\ell}$, with $\tau'_j\leq k * \tau_j$.

\end{enumerate}
 
\end{lemma}

\begin{proof}[Proof of item 1]
Let $\mathcal{C}_0,\mathcal{C}_1,\dots,\mathcal{C}_{\tau}$ be the computation of $M_1$ on input~$\mathfrak{I}$. That means that we have 
$\mathcal{C}_0 \vdash^{M_1}_{\mathfrak{I},O_1} \mathcal{C}_1 \vdash^{M_1}_{\mathfrak{I},O_2} 
\dots
\vdash^{M_1}_{\mathfrak{I},O_\tau} \mathcal{C}_{\tau}$, 
where $\mathcal{C}_0$ (resp.\ $\mathcal{C}_{\tau}$) is the initial (resp. a final) configuration of $M_1$
and $O_1,O_2,\dots, O_\tau$ are the successive outputs produced (notice that a ``non output'' instruction produces an empty output).
Then, by definition of the emulation~$E$ (Definition~\ref{def:emul}), we also have 
$\mathcal{C}^E_0 
\vdash^{M_2,i_1}_{\mathfrak{I},O_1} 
\mathcal{C}^E_1 
\vdash^{M_2,i_2}_{\mathfrak{I},O_2} \dots
\vdash^{M_2,i_\tau}_{\mathfrak{I},O_\tau}  
\mathcal{C}^E_{\tau}$, 
for some positive integers $i_1,i_2,\ldots,i_\tau \leq k$, where $\mathcal{C}^E_0$ 
(resp.\ $\mathcal{C}^E_\tau$) is the initial (resp. a final) configuration of $M_2$. 
Therefore, we have 
$\mathcal{C}^E_0 \vdash^{M_2,\tau'}_{\mathfrak{I},O} \mathcal{C}^E_{\tau}$, with the output
$O=O_1,O_2,\dots, O_\tau$ of $M_1$, and the computation time \linebreak
$\tau' = i_1+i_2+\cdots +i_\tau \leq k*\tau=O(\tau)$. 
\end{proof}

\begin{proof}[Proof of item 2]
From $\mathcal{C}_{j-1} \vdash^{M_1,\tau_j}_{\mathfrak{I}_j,O_j}
\mathcal{C}_j$, for $j=1,\ldots,\ell$, we deduce (by a proof similar
to that~of~item~1) that, for each $j$, we have $\mathcal{C}^E_{j-1}
\vdash^{M_2,\tau'_j}_{\mathfrak{I}_j,O_j} \mathcal{C}^E_j$, for some
$\tau'_j\leq k * \tau_j$. This is the expected result.
\end{proof}

\begin{remark}[constant time after linear-time preprocessing]
Lemma~\ref{lemma: faithful simul} has many applications.  For example,
let $M_1$ be a RAM whose computation can be divided into two phases:

\begin{enumerate}
\item $M_1$ reads an input $\mathfrak{I}_1$ and computes an output
  $O_1$ in time $T_1$;
\item after reading a second input $\mathfrak{I}_2$, the RAM $M_1$
  computes in time $T_2$ a second output $O_2$ (depending on~$O_1$ and
  $\mathfrak{I}_2$).
\end{enumerate}

\noindent
Let $M_2$ be a RAM which faithfully simulates $M_1$. Then the
computation of $M_2$ is exactly similar to that of $M_1$. This means
that the computation of $M_2$ divides into the same two phases~(1,2)
with the same outputs $O_1,O_2$ for the same inputs
$\mathfrak{I}_1,\mathfrak{I}_2$. There is only one difference: the
durations of phases 1 and 2 of $M_2$ are respectively $O(T_1)$ and
$O(T_2)$ (instead of $T_1$ and $T_2$).

In this paper, we will be particularly interested in the complexity
class called ``\emph{constant time after linear-time preprocessing}'',
defined in Subsection~\ref{subsec:linComplexityClasses}. This means
that the time bound $T_1$ of the first phase of the RAM computation is
\emph{linear}, i.e.\ is $O(N)$, where $N$ is the size of the first
input~$\mathfrak{I}_1$ (the first phase is called ``linear-time
preprocessing''), and the time bound $T_2$ of the second phase, which
reads a second input $\mathfrak{I}_2$ of constant size, is
\emph{constant}.

\end{remark}

Let us define what it means that two classes of RAMs are equivalent:

\begin{definition}[faithful simulation relation and equivalence of classes of RAMs]\label{def:equivalence}
Let $\mathcal{R}_1$ and $\mathcal{R}_2$ be two classes of RAMs.  We
say that $\mathcal{R}_2$ \emph{faithfully simulates} $\mathcal{R}_1$
if for each RAM $M_1\in \mathcal{R}_1$ there are a RAM
$M_2\in\mathcal{R}_2$ and an emulation of $M_1$ by $M_2$. Note that
the faithful simulation relation is reflexive and transitive and thus
when $\mathcal{R}_1$ also \emph{faithfully simulates} $\mathcal{R}_2$,
then we say that $\mathcal{R}_1$ and $\mathcal{R}_2$ are
\emph{equivalent}.
\end{definition}

\begin{remark}
  The class of $AB$-RAMs, or $R$-RAMs, or multi-memory RAMs are all
  equivalent.
\end{remark}
\subsection{Complexity classes in our RAM model}\label{subsec:linComplexityClasses}



The \emph{equivalence} we have defined, using the notion the
\emph{faithful simulation}, is quite natural but it is very
strong. Using this notion we can show that two equivalent RAMs define
the same complexity classes but the converse is not necessarily
true. For instance, as we will show in
Proposition~\ref{prop:multWithAdd}, the RAM equipped with addition
\emph{cannot} faithfully simulate the RAM equipped with addition and
the multiplication, but we will see in the next sections that those
diverse models define the same complexity classes of ``minimal time'':
linear time, etc.

In this subsection, we will introduce a somewhat weaker notion of
equivalence that makes two classes of RAMs ``equivalent'' if one class can simulate the
other (and vice versa) \emph{after a linear time initialization}. While this notion might seem
less natural, we can still show that the interesting complexity
classes are the same for ``equivalent'' RAMs and we can prove our
(surprising!) results making the RAM with addition equivalent to the
RAM equipped with all the usual arithmetic primitives.

First, the following result shows that the notion of faithful simulation turns out to be too restrictive:

\begin{proposition}\label{prop:multWithAdd}
$\RAM[\{+\}]$ \emph{cannot} faithfully simulate $\RAM[\{+,\times\}]$.
\end{proposition}

\begin{proof}
  To prove that $\RAM[\{+\}]$ cannot faithfully simulate
  $\RAM[\{+,\times\}]$, we only need to show that given two inputs
  $I[0]$ and $I[1]$, strictly less than $N$, a program in
  $\RAM[\{+\}]$ cannot output $I[0]\times I[1]$ in constant time.

  For that, we will consider such a program $P$ and suppose that the RAM
  starts initialized with zeros. We will consider $M_k$, the largest
  value strictly below $N$ that the RAM has access to after the $k$-th
  step (i.e.\ any constant in the program, any value in the input or in
  any register).

First, $M_0$ is the maximum value among $I[0]$, $I[1]$, all the
constants $c_1,\dots,c_r$ in the program and the maximum of registers
(which are all set to $0$).  For $M_{k+1}$ we see that the only way of
creating a larger maximum value is by summing two values and in that
case we still have $M_{k+1} \leq 2\times M_k \leq 2^{k+1} M_0$ (note
that by summing $N$ and anything else we get something bigger than $N$
and therefore it will not change the value of $M_k$).

  Overall, we have $M_k \leq 2^k \times M_0$ and thus by setting $I[0]
  = I[1] = \lfloor N^{1/3} \rfloor$ and taking $N$ big enough, we have
  $I[0] \times I[1] = \lfloor N^{1/3} \rfloor^2 > 2^\ell \times M_0
  \geq M_\ell$, where the \emph{constant} integer $\ell$ is the
  maximum number of steps executed by program $P$, because $M_0 =
  \max(I_0,I_1,c_1,\dots,c_r) = O(N^{1/3})$.  Therefore, the product
  $I[0] \times I[1]$ cannot be obtained at the $\ell$-th (final) step
  or before.
\end{proof}

\begin{remark}
Proposition~\ref{prop:multWithAdd} is not an isolated result. Several
variants can by proved by adapting the above proof, for instance:
\begin{itemize}
\item $\RAM[\{+,\dot{-}\}]$ \emph{cannot} faithfully simulate
  $\RAM[\{+,\dot{-},\times\}]$ where $\dot{-}$ is the ``subtraction''
  defined by $x\dot{-}y\coloneqq \max(0,x-y)$;

\item \label{item:proved}  $\RAM[\{+\}]$ \emph{cannot} faithfully simulate
  $\RAM[\{+,\dot{-}\}]$.
\item $\RAM[\{+,\times\}]$ cannot faithfully simulate
  $\RAM[\{+,\times,/ \}]$ where~$/$ denotes the Euclidean division.
\end{itemize}
  
Indeed consider, e.g., a program in $\RAM[\{+\}]$ trying to output
$I[0]\dot{-} I[1]$ in constant time. At each step of a constant-time
program in $\RAM[\{+\}]$, each register content is equal to a linear
combination of the form $a_0 I[0] + a_1 I[1] + a_2 N + b_1 c_1 + \dots
+ b_r c_r$, where $c_1,\dots,c_r$ are the constants of the program and
the integer coefficients $a_i,b_j$ depend on the current step and the
register involved; the important point is that the number of possible
tuples of coefficients $(a_0,a_1,a_2,b_1,\dots,b_r)$ after a constant
number of steps is \emph{bounded by a constant}. It implies that, for
$N$ large enough, such a linear combination cannot be equal to
$I[0]\dot{-} I[1]$ when we take $I[0]=\lf N^{1/2} \rf$ and $I[1]=\lf
N^{1/3} \rf$.

\end{remark}

\paragraph{Operations computed in constant time after linear-time initialization:}
We say that a RAM $\RAM{}$ computes an operation $\op$ \emph{in constant time after a linear-time initialization} if it works in two successive phases:
\begin{enumerate}
\item \emph{A pre-computation phase} where the RAM reads the integer
  $N\geq 1$ and computes an ``index" structure $p(N)$ in time $O(N)$.
  The pre-computed ``index" $p(N)$ will be a finite set of arrays
  (tables) and constants depending only on $N$, not on the rest of the
  input.
\item \emph{A computation phase} where the RAM reads the input
  $\mathcal{X}\coloneqq (x_1,\ldots,x_k)\in[0,cN]^k$, for fixed
  positive integers $k,c$, and outputs in constant time the result
  $\mathtt{op}(x_1,\ldots,x_k)\in[0,cN]$.
\end{enumerate}

In this paper, we will mainly focus on the following problem: what
operations can a RAM with addition (i.e., with
$\mathtt{Op}\coloneqq\{+\}$) compute in constant time after
linear-time initialization? We will show how this can be generalized
to RAMs which compute operations acting on ``polynomial" operands with
a ``polynomial" result: a ``polynomial" integer is an integer smaller
than $N^d$, for a constant integer $d>1$, and it is naturally
represented by the sequence of its $d$ ``digits" in the ``standard"
base $N$.

\begin{definition}[set of operations $\RAM{}$]\label{def:operation_Set_M_Op}

By abuse of notation, we call $\RAM{}$ the set of
operations $\op$ for which there exists a multi-memory RAM program
using only the $\Op$ operations that computes $\op$ in constant time
after linear initialization. 
\end{definition} 

\begin{definition}[equivalence of sets of operations, equivalence of RAM models]\label{def:equivOp,equivRAM}
Two sets of operations $\Op_1$ and $\Op_2$ are \emph{equivalent} when
the sets of operations $\mathcal{M}[\Op_1]$ and $\mathcal{M}[\Op_2]$
are equal; we will also say that the RAM models $\mathcal{M}[\Op_1]$
and $\mathcal{M}[\Op_2]$ are \emph{equivalent}.
\end{definition}

\noindent
These definitions will be justified by Corollary~\ref{cor:robustness} below.

\medskip
In this paper, we will mainly consider the two following complexity
classes: \emph{Linear time} and \emph{Constant time with linear-time initialization}. Recall that constant-time complexity is not a robust notion, as explained above, unless the constant-time computation is preceded by a linear-time initialization (also called linear-time preprocessing).

\paragraph{Linear time complexity:}
A RAM with an integer $N>0$ as input (resp.\ with a list of integers
$\Input\coloneqq (N,I[0],\ldots,I[N-1])$, $N>0$, as input) \emph{works
in linear time} if it stops after $O(N)$ steps. Note that, by
definition, such a RAM uses only integers~$O(N)$ as addresses and
contents of registers.

\paragraph{Constant time after linear-time preprocessing:}

A RAM with a list of integers \linebreak $\Input \coloneqq
(N,I[0],\ldots,I[N-1])$, $N>0$, as input and using only
integers~$O(N)$ (as addresses and contents of registers) \emph{works
in constant time after linear-time preprocessing} if it works in two
successive phases:
\begin{enumerate}
\item \emph{Pre-computation:} the RAM computes an ``index''
  $p(\Input)$ in time $O(N)$;
\item \emph{Computation:} from $p(\Input)$ and after reading a second
  input $\mathcal{X}$ of constant size, the RAM returns in constant
  time the corresponding result $\mathtt{output}(\Input,\mathcal{X})$.
\end{enumerate}

\begin{remark}
In our concluding section (Section~\ref{sec:conclusion}), we will also
study a third class of problems of ``minimal'' computational
complexity: the now well-known class of problems \emph{enumerable with
constant delay after linear-time preprocessing}.
\end{remark}

\paragraph{}
The main result of this paper is the proof of the robustness of our three ``minimal" complexity classes with respect to the set 
$\mathtt{Op}$ of arithmetic operations allowed, provided that $\mathtt{Op}$ contains addition. 
To begin with, the following proposition expresses that the two classes defined above are robust with respect to other variations of the instruction set.

\begin{proposition}\label{prop: robustInstRAM}
Let $\mathtt{Op}$ be a finite set of operations \emph{including addition}.
Then the class of problems computable in linear time (resp.\ in constant time after linear-time preprocessing) on a RAM with 
$\mathtt{Op}$ operations does not change depending on whether:
\begin{itemize}
\item the program is given by a sequence of $AB$-instructions;
\item the program is given by a sequence of $R$-instructions;
\item the program is given by a sequence of array-instructions.
\end{itemize}
\end{proposition}

\begin{proof}
It is a direct consequence of Lemma~\ref{lemma:lockstepEquivInst} and its proof.
The only condition to check is that in the three simulations (1-3) of this lemma, each integer handled is $O(N)$. This is evident for simulations 1 and 2. In simulation 3 of an array-instruction by a sequence of $R$-instructions (the only one of the three simulations that explicitly uses addition), it suffices to notice that the address $y$ of a register $R[y]$ which simulates an array element $T_j[x]$ is linear in its index $x$.
\end{proof}

As shown by the proofs of Lemma~\ref{lemma:lockstepEquivInst} and
Proposition~\ref{prop: robustInstRAM}, the availability of addition
seems essential to implement (one-dimensional) arrays in the RAM
model.  Most of this paper is devoted to showing that a RAM with only
addition can implement the other usual arithmetic operations,
subtraction, multiplication, division, etc., in constant time after
linear-time preprocessing.

\subsection{Reductions between operations in the RAM model}\label{subsec:reductions}
To show the robustness of the RAM model $\RAM{}$ with respect to the set $\mathtt{Op}$ of allowed operations, it will be convenient to introduce a notion of \emph{reduction} between operations or sets of operations.

By Definition~\ref{def:operation_Set_M_Op}, we have $\op\in\RAM{}$, 
for an operation $\op$ of arity $k$, if
there are two RAM procedures in $\RAM{}$, a ``preprocessing''
procedure, suggestively called $\fun{LinPreproc}()$ and acting in
time $O(N)$, and a ``computation'' procedure, called
$\fun{CstProc}(\mathbf{x})$, for $\mathbf{x}=(x_1,\ldots,x_k)$,
and acting in constant time, such that after executing
$\fun{LinPreproc}()$, executing $\fun{CstProc}(\mathbf{x})$
returns $\op(\mathbf{x})$, for each $k$-tuple of integers $\mathbf{x}$ in
the definition domain of the $\op$ operation.

\begin{example}\label{ex:pred}
Let $c\geq 1$ be a fixed integer. (It will be convenient to choose $c$ such that $cN$ is an upper bound of the integers allowed as contents of the RAM registers.) To define the arithmetic operations by induction we will need the predecessor function $\mathtt{pred}:
x\mapsto x-1$ from $[1,cN]$ to $[0,cN[$. To demonstrate $\mathtt{pred}
    \in \RAM[\{+\}]$, we use the following procedures of
    $\RAM[\{+\}]$, where $cN$ denotes the sum $N+\cdots +N$ ($c$ times):
 
      \begin{minipage}{0.05\textwidth}
        ~
      \end{minipage}
      \begin{minipage}{0.85\textwidth}
        \begin{algorithm}[H]
          \caption{Computation of \fun{pred}}
      \begin{minipage}[t]{0.50\textwidth}
      \begin{algorithmic}[1]
        \Procedure{LinPreprocPred}{\mbox{}}
  \State $\var{y} \gets 0$
  \While{$\var{y} \neq c\var{N}$}
    \State $\arr{PRED}[\var{y}+1] \gets \var{y}$
    \State $\var{y} \gets \var{y}+1$
  \EndWhile
\EndProcedure
      \end{algorithmic}
      \end{minipage}
      \begin{minipage}[t]{0.45\textwidth}
      \begin{algorithmic}[1]
\Procedure{Pred}{$\var{x}$}
\State \Return $\arr{PRED}[\var{x}]$
\EndProcedure
      \end{algorithmic}
      \end{minipage}
    \end{algorithm}
      \end{minipage}
\medskip

\noindent
Clearly, running $\fun{LinPreprocPred}()$ computes in $O(N)$
time the array $\predtab[1..cN]$ defined by $\predtab[y]\coloneqq
y-1$, for $y \in [1,cN]$, so that any (constant time) call
$\fun{Pred}(\alpha)$, for $\alpha \in [1,cN]$, returns
$\mathtt{pred}(\alpha)$.
This proves $\mathtt{pred} \in \RAM[\{+\}]$. 
\end{example}

\begin{lemma}\label{lemma:+=+AndPred}
 Any program $P\in \RAM[\{+,\mathtt{pred}\}]$, i.e.\ using the
 operations $+$ and $\mathtt{pred}$ is faithfully simulated by a
 program $P'\in \RAM[\{+\}]$ after a linear-time initialization also
 using only the operation $+$.
\end{lemma}

\begin{proof}[Proof of item~1]
$P'$ starts by calling the procedure $\fun{LinPreprocPred}$
  (linear-time initialization) and then calls $P$ but where each
  occurrence of the $\mathtt{pred}$ operation in $P$ is replaced with
  a call to the $\fun{Pred}$ procedure.
\end{proof}

Lemma~\ref{lemma:+=+AndPred} implies the equality
$\RAM[\{+,\mathtt{pred}\}]=\RAM[\{+\}]$. This means that the two sets
of operations computed in constant time after linear-time
initialization by a RAM using the operation $+$ and \emph{using} or
\emph{not using} the $\mathtt{pred}$ operation are equal.

The following general lemma establishes a ``transitivity'' property.

\begin{lemma}[Fundamental Lemma]\label{lemma:fund}
Let $\mathtt{Op}$ be a set of operations including $+$.  Let~$P$ be a
program in $\RAM[\mathtt{Op}\cup\{\op\}]$, i.e.\ using the operations
of $\mathtt{Op}\cup\{\op\}$, where $\op$ is an operation in $\RAM{}$,
$\op\not\in \mathtt{Op}$.

Then $P$ can be faithfully simulated by a program $P'$ in $\RAM[\mathtt{Op}]$, i.e.\ using only the operations in $\mathtt{Op}$ after a linear-time initialization, also using only the operations of $\mathtt{Op}$. 
\end{lemma}

\begin{proof}
Let $P$ be a program in $\RAM[\mathtt{Op}\cup\{\op\}]$.  
Intuitively, the idea of the proof is to create a program $P'$ that starts by calling 
$\fun{LinPreproc}$, the initialization procedure for $\op$, and then calls $P$ but where each occurrence of the $\op$ operation in $P$ is replaced with a call to $\fun{CstProc}$, the computation procedure of~$\op$.
This is the essence of our proof but we also need to take care of the following  problem: 



\begin{itemize} 
 \item 
The definition of $\op \in \RAM{}$ does not guarantee that it is possible to make two (or more) calls of
$\fun{CstProc}$. For instance, we might imagine that $\fun{CstProc}$ is destructive and prevents us from executing a second call to $\fun{CstProc}$ or that the answer might be wrong.
\end{itemize}
To solve this problem, we can modify the procedure $\fun{CstProc}$ to ensure that several successive calls are possible. 
For that, we use a modified procedure $\fun{CstProc}'$ where we store the value of each array cell or integer variable that is overwritten and at the end of the call to $\fun{CstProc}$ we restore the memory back to its state before the call. 
This is possible by maintaining a complete ``log" (a history) of the successive modifications (assignments/writes) of the memory. 
At the end of the call to $\fun{CstProc}'$, before its return instruction, the initial state of the memory is restored, step by step, by reading the ``log'' in the anti-chronological direction.

For that, we introduce three new arrays called
  $\arr{OldVal}$, $\arr{Index}$ and $\arr{Dest}$, as well as two new variables called
  $\var{NbWrite}$ and $\var{ReturnValue}$. 
  The modified procedure $\fun{CstProc}'$ starts by setting $\mathtt{NbWrite}=0$,
  and then, for each instruction in $\fun{CstProc}$ that writes
  something in an individual variable $x_i$ or in an array cell $A_i[t]$,
  i.e.\ an instruction of the form $x_i\gets u$ or $A_i[t]\gets u$, we write
  into 
  $\arr{OldVal}[\var{NbWrite}]$
  the current value $v$ of $x_i$ or $A_i[t]$ that will be overwritten;
  afterwards, we write into $\mathtt{\arr{Dest}}[\var{NbWrite}]$ the integer $i$
  that indicates\footnote{It is convenient to assume that the set of integers $i\leq k$ which identify the individual variables $x_i$ is disjoint from the set of integers $i>k$ that number the arrays $A_i$. Note that the test $i\leq k$ does not actually use the order symbol $\leq$ since it is equivalent to the disjunction of $k+1$ equalities $\bigvee_{j=0}^k i = j$.}
  which variable $x_i$ or array $A_i$ is being written into, and, if the write was a write into an array $A_i$ (i.e.\ if the integer $i$ is greater than~$k$), we note the value of the index $t$ of the cell $A_i[t]$ overwritten into $\arr{Index}[\var{NbWrite}]$; 
 finally, we write the new value $u$ into $x_i$ or $A_i[t ]$ and we increment $\var{NbWrite}$ by one.
 
 At the end of the procedure, instead of returning directly the expected value, we store it in the
 $\var{ReturnValue}$ variable. 
 Then, we restore the memory by replaying the ``log'' (write history) in reverse. 
 To do this, we decrease the counter variable $\var{NbWrite}$ by one as long as it is strictly positive and, each time, we restore the value $v=\arr{OldVal}[\var{NbWrite}]$ into the array cell (resp.\ variable) indicated by $i=\arr{Dest}[\var{NbWrite}]$, at index 
  $\var{index}=\arr{Index}[\var{NbWrite}]$ for an array cell. 
  This means that we execute the assignment $A_i[\var{index}]\gets v$ if $i>k$ 
  (resp.\ $x_i\gets v$ if $i\leq k$).
  Once done, we return the value stored in the $\var{ReturnValue}$ variable.
  
Formally, the code of the $\fun{CstProc}'(\mbox{})$ procedure is the concatenation of the following calls using the procedures of Algorithm~4 below:
\begin{itemize}
\item $\fun{InitCstProc'}(\mbox{})$;
\item $\fun{InitCstProc}(\mbox{})$ where each instruction of the form $\var{x}_{\var{i}} \gets u$
(resp.\ $\arr{A}_{\var{i}}[t] \gets u$, $\mathtt{Return}\;t$  )
is replaced by the call $\fun{AssignVar}(\var{i},u)$ (resp.\ $\fun{AssignCell}(\var{i},t, u)$, 
$\fun{StoreReturn}(t)$);
\item $\fun{FinalCstProc'}(\mbox{})$.
\end{itemize}

Finally, note that the time overhead of the procedure $\fun{CstProc}'$
is proportional to the number of assignments performed by the original
constant-time procedure $\fun{CstProc}$ it simulates, therefore the
$\fun{CstProc}'$ procedure is also constant-time.
\end{proof}

 \begin{minipage}{\almosttextwidth} 
  \begin{algorithm}[H] 
  \caption{
  Procedures used to construct \fun{CstProc}' from the code of \fun{CstProc} whose list of variables
  is $\var{x}_0,\dots,\var{x}_k$ and the list of arrays is $\arr{A}_{k+1},\dots,\arr{A}_{k'}$
  }
  
      \begin{minipage}[t]{0.53\textwidth} 
      
      \vspace{0.5em}
      
       \begin{algorithmic}[1]
       \Procedure{InitCstProc'}{$\mbox{}$}
       \State $\var{NbWrite}\gets 0$
       \State $\fun{LinPreprocPred}{()}$
       \EndProcedure
       \end{algorithmic}
       
      \vspace{1em}
      
      \begin{algorithmic}[1] 
     \Procedure{AssignVar}{\var{i},\var{u}} \Comment{$\var{x}_{\var{i}} \gets \var{u}$}
      \State $\arr{OldVal}[\var{NbWrite}]\gets \var{x}_{\var{i}}$
      \State $\arr{Dest}[\var{NbWrite}]\gets \var{i}$
      \State $\var{x}_{\var{i}} \gets \var{u}$
      \State $\var{NbWrite}\gets \var{NbWrite}+1$
    \EndProcedure
      \end{algorithmic} 
  
      \vspace{1em}
      
      \begin{algorithmic}[1] 
     \Procedure{AssignCell}{\var{i},\var{t},\var{u}} \Comment{$\arr{A}_{\var{i}}[\var{t}] \gets \var{u}$}
      \State $\arr{OldVal}[\var{NbWrite}]\gets \arr{A}_{\var{i}}[\var{t}]$
      \State $\arr{Dest}[\var{NbWrite}]\gets \var{i}$
      \State $\arr{Index}[\var{NbWrite}]\gets \var{t}$
      \State $\arr{A}_{\var{i}}[\var{t}] \gets \var{u}$
      \State $\var{NbWrite}\gets \var{NbWrite}+1$
      \EndProcedure
      \end{algorithmic} 
      
      \vspace{1em}
      
      \end{minipage} 
\begin{minipage}[t]{0.45\textwidth} 
 
 \vspace{1.5em}
 
 \begin{algorithmic}[1] 
\Procedure{StoreReturn}{$\var{t}$} 
\State $\var{ReturnValue}\gets \var{t}$
\EndProcedure
\end{algorithmic} 

      \vspace{1em}
      
\begin{algorithmic}[1] 
\Procedure{FinalCstProc'}{$\mbox{}$}
\While {$\var{NbWrite}>0$}
  \State $\var{NbWrite}\gets \arr{PRED}[\var{NbWrite}]$
  \State $\var{i}\gets \arr{Dest}[\var{NbWrite}]$
  \If{$\var{i} > k$}
     \State $\var{ind}\gets \arr{Index}[\var{NbWrite}]$
     \State $\arr{A}_{\var{i}}[\var{ind}] \gets \arr{OldVal}[\var{NbWrite}]$
   \Else
      \State $\var{x}_{\var{i}} \gets \arr{OldVal}[\var{NbWrite}]$
  \EndIf
\EndWhile
    \State \Return $\var{ReturnValue}$
\EndProcedure
\end{algorithmic} 

\end{minipage} 
      
\end{algorithm} 
 \end{minipage} 
 
\bigskip
The following result is a direct consequence of Lemma~\ref{lemma:fund}:

\begin{corollary}\label{cor:transitivity_base}
Let $\mathtt{Op}$ be a set of operations including $+$, and let $\op'$
be an operation in $\RAM[\mathtt{Op}\cup\{\op\}]$ where $\op$ is an
operation in $\RAM{}$.  Then we have
$\op'\in\RAM{}$.
\end{corollary}

\begin{proof}
  If $\op'$ is in $\RAM[\mathtt{Op}\cup\{\op\}]$ and $\op$ is an operation in $\RAM{}$, 
  then we can emulate the linear-time preprocessing and the constant-time computation of
  $\op$ in $\RAM{}$.
\end{proof}

\begin{corollary}\label{cor:robustness}
Let $\mathtt{Op}$ be a set of operations including $+$, and let $\op$ be an operation in $\RAM{}$. 
Then we have the equality of sets of operations 
$\RAM[\mathtt{Op}\cup\{\op\}]=\RAM{}$.
\end{corollary}

Using iteratively Corollary~\ref{cor:robustness}, we deduce the
following result.

\begin{corollary}[Transitivity Corollary]\label{cor:transitivity}
Let $\op_1,\dots,\op_k$ be a list of operations. Let us define the set of operations 
$\mathtt{Op}_0\coloneqq \{+\}$ and 
$\mathtt{Op}_i\coloneqq \mathtt{Op}_{i-1}\cup \{\op_i\}$, for $i=1,\dots, k$.
Assume $\op_i \in \RAM[\mathtt{Op}_{i-1}]$, for $i=1,\dots, k$.

Then we have the sequence of equalities
$ \RAM[\{+\}]  = \RAM[\mathtt{Op}_1]=\cdots = \RAM[\mathtt{Op}_k]$.
This means that each set of operations $\mathtt{Op}_{i}$ thus defined recursively is equivalent to 
$\{+\}$.
\end{corollary}

\begin{proof}
The first equality is given by Lemma~\ref{lemma:+=+AndPred}. Each of the $k$ following equalities are deduced from Corollary~\ref{cor:robustness} by setting $\mathtt{Op}\coloneqq \mathtt{Op}_{i-1}$
and $\op \coloneqq \op_i$, for each  $i=1,\dots k$. 
\end{proof}

From now on, we will make repetitive and informal (intuitive) use of Lemma~\ref{lemma:fund} (Fundamental Lemma) and Corollary~\ref{cor:transitivity} (Transitivity Corollary) without mentioning them explicitly.

\begin{remark}[Cumulative property]
The result of Corollary~\ref{cor:transitivity} is ``cumulative''.  As
we will use for the successive operations studied in the following
sections, the preprocessing and computation of an $ \op_i$ operation
can use any operation $\op_j$ introduced before ($j<i$) but also any
element (array, constant integer, etc.) pre-computed previously,
i.e.\ constructed to compute a previous~$\op_j$.
\end{remark}

%% file: 4_BaseK.tex
\section{Computing sum, difference, product and base change in constant time}\label{sec:SumDiffProdBase}

In this section, we show that the first three usual operations acting
on ``polynomial'' integers can be performed in constant time on a RAM
of $\RAM[\{+\}]$ after a linear-time preprocessing. 

\subsection{Moving from register integers to a smaller base}

A fundamental property for establishing the robustness of complexity
classes in the RAM model is its ability to \emph{pre-compute} some
tables.  In particular, in this subsection, we show how a RAM with
addition can \emph{compute in linear time} several tables which allow
to transform a register integer $x$ into an integer in base $\largeB$,
with $\largeB$ large enough, so that $x$ fits within a constant number
of registers, but small enough to allow us, in the following
subsections, to pre-compute in \emph{constant time} the usual
arithmetic operations on integers less than $\largeB$: subtraction,
multiplication, base change, division, etc. This in turn will allow us
to compute in constant time operations over $O(N)$ or $O(N^d)$ integers, for a fixed integer $d$.

\paragraph{A general principle to implement binary arithmetic operations in constant time:}
To implement in constant time an arithmetic binary operation $
\mathtt{op} $, e.g.\ subtraction or multiplication, on operands less
than $ N $ (or less than $ N^d $, for a fixed integer $d$), it is
convenient to convert these operands into a base $\largeB$ such that
$N^{1/2}\leq \largeB = O(N^{1/2})$. We will see below how we can use this:
\begin{itemize}
\item on the one hand, to represent any natural number $ a <N$ (and
  therefore less than $\largeB^2$) by its two digits $ a_1, a_0 $ in
  the base $ \largeB $: $a=a_1\times B + a_0$ and thus $ a_1 = a \;
  \mathtt{div} \; \largeB $ and $ a_0 = a \; \mathtt{mod} \; \largeB
  $,
\item on the other hand, for any constant integer $c\geq 1$, to
  construct in time $ O(\largeB^2) = O(N) $ a ``binary" table $
  T_{\mathtt{op}}[0..\largeB-1][0..\largeB-1]$ of a binary operation $ \mathtt{op}
  $ restricted to operands less than $ \largeB $, i.e., such that, for $ x,
  y <\largeB $, $ T_{\mathtt{op}} [x] [y] \coloneqq x \; \mathtt{op} \; y $.
\end{itemize}

For this, we must pre-compute the constant $\largeB \coloneqq \lceil \sqrt{N}\rceil$ 
and then the tables of the functions 
$x\mapsto x \;\mathtt{div} \; \largeB$ and $x\mapsto x \; \mathtt{mod} \; \largeB$, for $x<cN$,
and of the function $x\mapsto \largeB x$, for $x<\largeB$. This will be done below.

\paragraph{Computing square roots:}
we will construct the array $\mathtt{CEIL}\_\mathtt{SQRT}[1..N]$ and the constant $\largeB$
defined by $\mathtt{CEIL}\_\mathtt{SQRT}[y]\coloneqq \lceil \sqrt{y}\rceil$, 
for $1\leq y \leq N$ and  $\largeB\coloneqq \lceil \sqrt{N}\rceil = \arr{CEIL\_SQRT}[N] $.  
They are computed in time $O(N)$ by the following code using only
equality tests and the addition operation:

      \begin{minipage}{0.05\textwidth}
        ~
      \end{minipage}
      \begin{minipage}{0.85\textwidth}
        \begin{algorithm}[H]
          \caption{Computation of \fun{Sqrt}, with $\fun{Sqrt}(x)\coloneqq \lc \sqrt{x} \rc$ for $1\leq x \leq N$, and $\var{B}\coloneqq \lc \sqrt{\var{N}} \rc$}
      \begin{minipage}[t]{0.50\textwidth}
      \begin{algorithmic}[1]
        \Procedure{LinPreprocSqrt}{\mbox{}}
  \State $\var{x} \gets 1$
  \State $\var{xSq} \gets 1$
  \For{$\var{y} \From  1 \To \var{N}$}
  \State $\arr{CEIL\_SQRT}[\var{y}] \gets \var{x}$
  \If{$\var{y} = \var{xSq}$}
  \State $\var{xSq} \gets \var{xSq}+\var{x}+\var{x}+1$
  \State $\var{x} \gets \var{x}+1$
  \EndIf
  \EndFor
  \EndProcedure  
      \end{algorithmic}
      \end{minipage}
      \begin{minipage}[t]{0.45\textwidth}
      \vspace{1em}
      \begin{algorithmic}[1]
\Procedure{Sqrt}{$\var{x}$}
\State \Return $\arr{CEIL\_SQRT}[\var{x}]$
\EndProcedure
      \end{algorithmic}
      
      \vspace{1em}
      \begin{algorithmic}[1]
      \Procedure{B}{$\mbox{}$}
 \State \Return $\arr{CEIL\_SQRT}[\var{N}]$     
\EndProcedure
      \end{algorithmic}      
      \end{minipage}
      
    \end{algorithm}
      \end{minipage}

\begin{proof}[Proof of correctness]
When initializing the variables $\mathtt{x},\mathtt{xSq}$ (at lines
2-3) and after each update (lines 7-8), we always have
$\mathtt{xSq}=\mathtt{x}^2$ due to the identity $(x+1)^2=x^2+x+x+1$.
Besides, the inequalities $(\var{x}-1)^2 < \var{y} \leq \var{x}^2$
(meaning $\var{x}=\lceil \sqrt{\var{y}} \rceil$) are maintained
throughout the execution of the algorithm: first, they hold at the
initialization, that is $\var{y}=1=\var{x}^2$; second, if \linebreak
$(\var{x}-1)^2 < \var{y} \leq \var{x}^2$, then either we have
$(\var{x}-1)^2 < \var{y}+1 \leq \var{x}^2$ (case $\var{y}<\var{x}^2$),
or we have $\var{x}^2<\var{y}+1\leq (\var{x}+1)^2$ (case $\var{y}=\var{x}^2$).  
This justifies the instruction
$\arr{CEIL\_SQRT}[\var{y}]\gets\var{x}$.
\end{proof}

\paragraph{Choosing a new reference base:}
We have defined the base $\largeB\coloneqq\lceil \sqrt{N} \rceil$. Now that we are
manipulating integers represented in base $\largeB$, an integer $O(N)$
cannot be stored in a single register. For the sake of simplicity, we
will represent integers as bounded arrays with a constant size $d$,
hence a representation able to represent all the integers up to
$N^{d/2}$.

\paragraph{Conversion between integers less than $cN$ and integers in base $\largeB$:}
An array $\arr{a}[0..d-1]$ of size $d$ \emph{represents in base $\largeB$} the
integer $\sum_{0\leq i < d} \arr{a}[i]\largeB^i$. In general we will assume
that the representation is \emph{normalized} which means that
$\arr{a}[i] < \largeB$ for all $0\leq i < d$. All the procedures that we are going to present
return normalized integers but they also often accept non
normalized integers as inputs. For functions that only accept
normalized integers, we can normalize them using the $\fun{Normalize}$
function. Let us now explain how to move from register integers less than $cN$ (recall that $c\geq 1$ is defined as the integer such that $cN$ is an upper bound of the register contents)
into base $\largeB$ integers (i.e.\ arrays of $d$ integers less than $\largeB$).

Let us define the arrays $\divtab \largeB[0..cN]$ and $\modtab
\largeB[0..cN]$ by $\divtab \largeB[x] \coloneqq x \divop \largeB$ and $\modtab \largeB[x]
\coloneqq x \modop \largeB$, for $x\leq cN$. Clearly, $\divtab \largeB$ and $\modtab \largeB$ are computed by the following algorithm running in time $O(N)$, and 
the procedure $\fun{ToBase\largeB}(\var{x},\arr{Out})$, which computes for an integer $x\leq cN$ its base $\largeB$ representation\footnote{Note that for $N>c^2$ we have $c<\lc \sqrt{N} \rc$, which implies $cN < \lc \sqrt{N} \rc^3$, so that any integer $x\leq cN$ can be represented with 3 digits in base $\largeB =\lc \sqrt{N} \rc$. }
 $\arr{Out}[0..d-1]$, with $d\geq 3$,
runs in constant time.

      \begin{minipage}{\almosttextwidth}
        \begin{algorithm}[H]
          \caption{Computation of the base conversion}
      \begin{minipage}[t]{0.50\textwidth}
      \begin{algorithmic}[1]
        \Procedure{LinPreprocBaseConv}{\mbox{}}
  \State $\arr{DIV\largeB}[0] \gets 0$
  \State $\arr{MOD\largeB}[0] \gets 0$
  \For{$\var{x} \From 1  \To cN$}
    \If{$\arr{MOD\largeB}[\var{x}-1] \neq \var{\largeB}-1$}
      \State $\arr{MOD\largeB}[\var{x}] \gets \arr{MOD\largeB}[\var{x}-1]+1$
      \State $\arr{DIV\largeB}[\var{x}] \gets \arr{DIV\largeB}[\var{x}-1]$
    \Else 
      \State $\arr{MOD\largeB}[\var{x}] \gets 0$
      \State $\arr{DIV\largeB}[\var{x}] \gets \arr{DIV\largeB}[\var{x}-1]+1$
      \State $\arr{MULT\largeB}[\arr{DIV\largeB}[\var{x}]] \gets \var{x}$
    \EndIf
  \EndFor 
  \EndProcedure  
      \end{algorithmic}
      \end{minipage}
      \begin{minipage}[t]{0.49\textwidth}
      \begin{algorithmic}[1]
        \Procedure{ToBase\largeB}{$\var{x}$}
        \State $\arr{Out} \gets $ array of size $d$
        \For{$\var{i} \From 0 \To d-1$}
          \State $\arr{Out}[\var{i}] \gets \arr{MOD\largeB}[\var{x}]$
          \State $\var{x}\gets\arr{DIV\largeB}[\var{x}]$
        \EndFor
        \State \Return $\arr{Out}$  
\EndProcedure
      \end{algorithmic}
\vspace{1em}
      \begin{algorithmic}[1]
          \Procedure{Normalize}{$\arr{In}$}
        \State{$\var{carry}\gets 0$}
        \State $\arr{Out} \gets $ array of size $d$
        \For{$\var{i} \From 0  \To \var{d}-1$}
          \State $\var{temp}\gets \arr{In}[\var{i}] + \var{carry}$
          \State $\var{carry}\gets \arr{DIV\largeB}[\var{temp}]$
          \State $\arr{Out}[\var{i}]\gets \arr{MOD\largeB}[\var{temp}]$
        \EndFor
        \State \Return $\arr{Out}$
      \EndProcedure
   \end{algorithmic}
      \end{minipage}
    \end{algorithm}
      \end{minipage}

\begin{proof}[Proof of correctness]
  The arrays $\arr{DIV\largeB}$ and $\arr{MOD\largeB}$ are initialized such that
  $\arr{DIV\largeB}[x] = \lfloor x/\largeB \rfloor $ and $\arr{MOD\largeB}[x]=x \modop \largeB$, for $x\leq cN$.
  The array $\arr{MULT\largeB}$ satisfies the equality $\arr{MULT\largeB}[y]=y\times \largeB$ 
  for all $y$ such that $y\times \largeB\leq cN$.
Obviously, the procedure $\fun{Normalize}$ correctly normalizes any array~$\mathtt{In}$ which represents (in base~$\largeB$) an integer less than~$\largeB^d$.
\end{proof}

\subsection{Computing sum, difference and product in constant time}
\label{sec:sumdiffproduct}

  In the previous subsection, we showed that our integers $O(N)$
          could be moved into a base $\largeB$ such that 
          $N^{1/2}\leq \largeB = O(N^{1/2})$.
This allows us to create $d$-dimensional arrays $\largeB\times \dots \times \largeB$, for a fixed $d\geq 3$,
          storing the results for each ``digit'' of our base and by combining
          them we will get the result for integers $O(N)$.
This technique will be first used for the \emph{comparison} ($\leq$-inequality test) of integers
          and then we will show how to apply it 
          to the \emph{sum}, \emph{difference} and \emph{product} of integers $O(N)$, or even
          ``polynomial'' integers, i.e.\ integers less than $N^{d/2}$ (since $N^{d/2}\leq \largeB^d$).

          The algorithms presented here are variants of the standard
          algorithms for large number manipulation but here they run
          in constant time and rely on the pre-computed tables (whereas
          standard algorithms generally suppose that each operation
          is available on registers).

\paragraph{Adding comparison of integers:}

In our RAM model we assumed that we cannot compare large integers by inequalities, $<$, $\leq$, etc., but we saw a way to move from large integers (up to $cN$) to relatively small integers (less than $\largeB$).
We now see how to pre-compute an array $\arr{LowerEQ}[0..\largeB-1][0..\largeB-1]$ 
such that $\arr{LowerEQ}[\var{x}][\var{y}] := 1$ when $x\leq y$ and $0$ otherwise. 
Using this array and a simple for loop we will be able to
compare integers lower than $\largeB^d$ by inequalities, all of that just using the addition 
and $\mathtt{pred}$.
Note that this comparison requires that both arguments are
\emph{normalized}! 

\begin{minipage}{\almosttextwidth}
  \begin{algorithm}[H]
    \caption{Computation of $\leq$-comparison of integers lower than $\var{\largeB}^d$}
    \begin{minipage}[t]{0.45\textwidth}
      \vspace{1em}
      \begin{algorithmic}[1]
        \Procedure{LinPreprocComp}{\mbox{}}
\For{$\var{x} \From  0 \To  \var{\largeB}-1$}
\For{$\var{y} \From  \var{x} \To \var{\largeB}-1$}
    \State $\arr{LowerEQ}[\var{x}][\var{y}] \gets 1$
  \EndFor
  \For{$\var{y} \From  \var{x}-1 \Downto 0$}
    \State $\arr{LowerEQ}[\var{x}][\var{y}] \gets 0$
  \EndFor
  \EndFor
  \EndProcedure
      \end{algorithmic}
    \end{minipage}
    \begin{minipage}[t]{0.540\textwidth}
      \begin{algorithmic}[1]
\Procedure{LowerEqualBase\largeB}{$\arr{a},\arr{b}$}
\For{$\var{i} \From \var{d}-1 \Downto 0$}
  \If{$\var{a}[\var{i}] \neq \var{b}[\var{i}]$}
     \State \Return{$\arr{LowerEQ}[\var{a}[\var{i}]][\var{b}[\var{i}]]$}
  \EndIf
  \EndFor
  \State \Return{1}\Comment{Case of equality }
  \EndProcedure
      \end{algorithmic}
      \vspace{1em}
      \begin{algorithmic}[1]
\Procedure{LowerEqual}{$\arr{a},\arr{b}$}
\State \Return{\Call{LowerEqualBase\largeB}{\Call{ \\
\mbox{~}\hspace{8em} ToBase\largeB}{$\var{a}$}, \\
\mbox{~}\hspace{9em}\Call{ToBase\largeB}{$\var{b}$}}}
  \EndProcedure
      \end{algorithmic}
    \end{minipage}
  \end{algorithm}
\end{minipage}

\begin{proof}[Proof of correctness]
Once the array $\arr{LowerEQ}$ is initialized as expected, if we take
two arrays of size $d$ representing normalized integers lower than
$\largeB^d$ then it is clear that we just have to look at the biggest index
where the arrays differ and compare the values contained at this
position. If the arrays do not differ, then the numbers are equal.
\end{proof}

\paragraph{Sum and difference:}
The sum and difference (of integers less than $\largeB^d$) are pretty straightforward applications of the
textbook algorithms, we just have to make sure to propagate the
carries in the right way. For the difference, we pre-compute an array
$\arr{DIFF}[0..2\largeB-1][0..2\largeB-1]$ such that $\arr{DIFF}[x][y]\coloneqq x-y$ for $0\leq y \leq x < 2\largeB$.

Note that the sum is not checking for a potential overflow (i.e.\ when
the result is greater than or equal to $\largeB^d$) and the difference is not
verifying that the number represented by $a$ is actually greater than the
number represented by $b$ but in both cases we could read it from the
$\var{carry}$ variable. Those algorithms work as if the numbers are
treated modulo $\largeB^d$ (just like many modern computers work modulo
$2^{64}$). Note also that the procedure $\fun{Difference}$ expects 
normalized arguments but the $\fun{Sum}$ procedure works even if the
given arguments are not normalized and returns a
normalized number.

\begin{minipage}{\almosttextwidth}
  \begin{algorithm}[H]
    \caption{Computation of sum and difference}
    \begin{minipage}[t]{0.50\textwidth}
      \begin{algorithmic}[1]
        \Procedure{LinPreprocSumDiff}{\mbox{}}
        
\For{$\var{x} \From  0 \To  2\times \var{\largeB}-1$}
\State $\arr{DIFF}[\var{x}][\var{x}] \gets 0$
\For{$\var{y} \From  \var{x}-1 \Downto 0$}
\State $\arr{DIFF}[\var{x}][\var{y}] \gets \arr{DIFF}[\var{x}][\var{y}+1]+1 $
\EndFor
\EndFor
\EndProcedure
      \end{algorithmic}
      \vspace{1em}
      \begin{algorithmic}[1]
\Procedure{Sum}{$\arr{a},\arr{b}$}
\State $\arr{res} \gets $ array of size $d$
\For{$\var{i} \From 0 \To \var{d}-1$}
 \State $\var{res}[\var{i}] \gets \var{a}[\var{i}] +  \var{b}[\var{i}] $
 \EndFor
\State \Return \Call{$\fun{Normalize}$}{$\var{res}$}
  \EndProcedure
      \end{algorithmic}
    \end{minipage}
    \begin{minipage}[t]{0.480\textwidth}
      \begin{algorithmic}[1]
\Procedure{Difference}{$\arr{a},\arr{b}$}
\State{$\var{carry} \gets 0$}
\State $\arr{res} \gets$ array of size $d$
\For{$\var{i} \From 0 \To \var{d}-1$}
\State $\var{c} \gets \var{b}[\var{i}] + \var{carry}$
\If{$\arr{LowerEQ}[\var{c}][\var{a}[\var{i}]]$}
\State $\var{res}[\var{i}] \gets \arr{DIFF}[\var{a}[\var{i}]][\var{c}]$
\State $\var{carry} \gets 0$
\Else
\State $\var{res}[\var{i}] \gets \arr{DIFF}[\var{a}[\var{i}]+\largeB][\var{c}]$
\State $\var{carry} \gets 1$
\EndIf
\EndFor
\State \Return $\arr{res}$
  \EndProcedure
      \end{algorithmic}
    \end{minipage}
  \end{algorithm}
\end{minipage}

\begin{proof}[Proof of correctness]
  The sum is done component by component. The result is not
  necessarily a normalized integer (even if the arguments $\var{a},\var{b}$ are normalized)
  but we normalize it afterward: finally, the integer represented by $\var{res}$ is the sum of the integers 
  $\mathtt{val}(\var{a})$ and $\mathtt{val}(\var{b})$ represented by $\var{a}$ and $\var{b}$ if it is less than $B^d$ (more generally, it is $\mathtt{val}(\var{a})+\mathtt{val}(\var{b}) \mod B^d$). 
  
  For the difference, the idea is to
  maintain a carry along the algorithm. When the carry is set to one,
  it means that we need to subtract it from the next ``digit''. Notice
  that storing a carry of $1$ is enough as $\arr{b}[i]<\largeB$ and
  $\var{carry}\leq 1$ and therefore $\arr{b}[i]+\var{carry} \leq \largeB \leq
  \largeB+\arr{a}[\var{i}]$. 
  Notice also that $\arr{a}[\var{i}]<\largeB$ and thus $\arr{a}[\var{i}]+\largeB < 2\largeB$, which justifies that the array
  $\arr{DIFF}$ only needs to be initialized up to $2\largeB-1$.
Finally, note that the $\fun{Difference}$ procedure uses not only the pre-computed array $\arr{DIFF}$ but also the array $\arr{LowerEQ}$ pre-computed in Algorithm~7 above.
\end{proof}

\paragraph{Multiplication:}
The recipe for getting multiplication is, unsurprisingly, the same as
before: we pre-compute an array $\arr{MULT}[0..\largeB-1][0..\largeB-1]$ 
and use it with the naive multiplication algorithm.

\begin{minipage}{\almosttextwidth}
  \begin{algorithm}[H]
    \caption{Computation of multiplication}
    \begin{minipage}[t]{0.50\textwidth}
      \vspace{1.5em}
      \begin{algorithmic}[1]
\Procedure{LinPreprocMult}{\mbox{}}
\For{$\var{x} \From  0 \To \var{\largeB}-1$}
\State $\arr{MULT}[\var{x}][0] \gets 0$
\For{$\var{y} \From  0 \To \var{\largeB}-2$}
\State $\arr{MULT}[\var{x}][\var{y}+1] \gets \arr{MULT}[\var{x}][\var{y}]+\var{x}$
\EndFor
\EndFor
\EndProcedure
      \end{algorithmic}
    \end{minipage}
    \begin{minipage}[t]{0.490\textwidth}
      \begin{algorithmic}[1]
\Procedure{Multiply}{$\arr{a},\arr{b}$}
\State $\arr{res} \gets $ array of size $d$
\For{$\var{i} \From 0 \To \var{d}-1$}
\State $\arr{res}[\var{i}] \gets 0$
\EndFor
\For{$\var{i} \From 0 \To \var{d}-1$}
\For{$\var{j} \From 0 \To \var{d}-1-\var{i}$}
\State $\arr{res}[\var{i}+\var{j}] \gets \arr{res}[\var{i}+\var{j}] +$
\State \hspace{6em} 
$\arr{MULT}[\arr{a}[\var{i}]][\arr{b}[\var{j}]]$
\EndFor
\EndFor
\State \Return \Call{$\fun{Normalize}$}{$\var{res}$}
\EndProcedure
      \end{algorithmic}
    \end{minipage}
  \end{algorithm}
\end{minipage}

\begin{proof}[Proof of correctness]
Here $\arr{MULT}$ stores a result that might be larger than $\largeB$, which
means that in the multiplication $\arr{res}$ will not be normalized
but it will get normalized later. (Each cell of~$\arr{res}$ might contain a value less than or equal to
$d\times(\largeB-1)^2 < dN$. If that is considered too much we could also
normalize after each loop but we don't do it here for the sake of
simplicity.)
Finally, the integer represented by $\var{res}$ is the product of the integers 
$\mathtt{val}(\var{a})$ and $\mathtt{val}(\var{b})$ represented by $\var{a}$ and $\var{b}$ if it is less than $B^d$ and, more generally, it is $(\mathtt{val}(\var{a})\times \mathtt{val}(\var{b})) \modop B^d$.
\end{proof}

\paragraph{How does the complexity of the sum, the difference, and the product depend on the size of the operands?}
Clearly, all the algorithms presented above for addition, subtraction
and product run in constant time. This is due to the fact that each
instruction involves a fixed number of registers and that it is
executed a number of times bounded by the constant $d$. Of course,
more precisely, the execution time of our algorithms depends on the
number of digits (= number of registers) in the base $\largeB=\lceil
N^{1/2}\rceil$ of their operands, which is at most twice their number
of digits in the base $N$.  It is a simple matter to observe that the
above addition and subtraction algorithms are executed in time $O(d)$
and that the multiplication algorithm of two numbers of $k$ and $m$
digits, respectively, runs in time $O(k(k+m))$, which is $O(d^2)$, for
$k+m=d$. Note that if $d$ was large, one could fall back to any fast
multiplication algorithm such as the Fast Fourier Transform.

\begin{remark}\label{rem:complexitySumDiff}
Note that our $\fun{Sum}$ (resp.\ $\fun{Diff}$) algorithm on ``polynomial'  integers $a,b< N^d$ (or $a,b< \largeB^d$, for $\largeB=\lceil N^{1/2}\rceil$) working in time $O(d)$, can be seen as an addition (resp. subtraction) algorithm on integers $a,b<N^\ell$ working in time $O(\ell)$, for any non-constant integer~$\ell$. 

In particular, suppose that $a,b$ are arbitrary integers written in binary with at most $n$ bits, i.e.\ $a,b<2^n$; 
then, splitting their binary representations into subwords of length $\lc (\log_2 n)/2 \rc$ (the last subword can be shorter), $a,b$ can be represented in base $\beta \coloneqq 2^{\lc (\log_2 n)/2 \rc}$ by the arrays $\arr{a}[0..N-1]$ and $\arr{b}[0..N-1]$ of $N$ digits $\arr{a}[i],\arr{b}[i]$, where 
$N \coloneqq \lc n/ \lc (\log_2 n)/2 \rc \rc$, with $\arr{a}[i],\arr{b}[i]<\beta<N$ (for~$n$ large enough).
Obviously, our $\fun{Sum}$ (resp.\ $\fun{Diff}$) algorithm applied to the arrays $\arr{a}[0..N-1]$ and $\arr{b}[0..N-1]$, with $N$ instead of $d$, computes their sum array $\arr{s}[0..N]$ (resp. difference array $\arr{s}[0..N-1]$) in time $O(N)=O(n/ \log n)$ under \emph{unit cost criterion}. Note that this is $O(n)$ time (linear time) under \emph{logarithmic cost criterion}.
\end{remark}

%% file: 5_Division.tex
\section{Executing division in constant time}\label{section:division}

In this section, we present one of the most significant and surprising
result of this paper. Although division is obviously the most
difficult of the four standard operations, we will show in this
section that it can also be computed in constant time -- again after
linear-time preprocessing.

The importance of this result is underlined by its consequences:
\begin{itemize}
\item the following sections will show that a number of other
  arithmetic operations on ``polynomial" integers -- square root,
  general $c$th root, exponential, logarithm -- can also be computed
  in constant time, by using the four operations, including,
  essentially, division;
\item like the other arithmetic operations, the division can be used
  by some ``dynamic'' problems: as an example, the
  paper~\cite{BaganDGO08} uses a particular case of division to
  compute in constant time (with linear-time preprocessing) the $j$th
  solution of a first-order query on a bounded-degree structure.
\end{itemize}

\bigskip 
The novelty of our division algorithm is a very subtle case separation
which uses a form of approximation and recursive procedures whose
number of calls is bounded by a constant.

It will be convenient to represent the operands of a division sometimes in the base \linebreak
$\smallB{}\coloneqq \lc N^{1/6} \rc$, 
sometimes in the base $\largeB \coloneqq \smallB{}^3$. (Note that the conversion of an integer from base $\smallB{}$ to base $\largeB$, and vice versa, is immediate.) The bounds $N^{1/2}\leq
\largeB= O(N^{1/2})$ will be useful.

We say that an integer $x$ is ``small" if $x<\largeB$, i.e., if it has
only one digit when represented in the base $\largeB$. 
First, let us check that this division by a ``small'' integer can be done in
constant time by a variant of the standard division algorithm.

\subsection{Dividing a ``polynomial" integer by a ``small'' integer}\label{subsec:divbysmallint}
Here, we want to divide a ``polynomial" integer $a\coloneqq
a_{d-1}\largeB^{d-1}+\cdots+a_1\largeB+a_0$, represented in the base
$\largeB$, by an integer $b<\largeB$. For that we will use Horner's
method. Let us define the sequence $v_{d},\ldots,v_1,v_0$ defined
inductively by $v_{d}\coloneqq 0$ and $v_i\coloneqq \largeB \times v_{i+1}+a_i$,
for $0\leq i < d$. Note that $v_0=a$.

For the $v_i$, we can define $q_i$ and $r_i$, the quotient and the
remainder of dividing $v_i$ by $b$, for $0\leq i \leq d$. Note that
$v_0=a$, therefore $q_0$ is the number we are looking for.  Since
$v_i=\largeB \times v_{i+1}+~a_i$ and $v_{i+1}=q_{i+1} \times
b+r_{i+1}$ we deduce 
\[
v_i=\largeB \times q_{i+1} \times b+\largeB
\times r_{i+1}+a_i
\]
Let us introduce the respective quotients $q',q''$ and remainders $r',r''<b$ of the division of $a_i$ and 
$\largeB \times r_{i+1}$ by $b$: we have
$a_i = q' \times b + r'$ and $\largeB \times r_{i+1} = q'' \times b + r''$ and therefore 
\[
v_i = (q_{i+1}\times \largeB + q' + q'') \times b + r' + r''
\]
 and thus, depending on whether $r'+r'' < b $, we have 
\begin{itemize}
\item $q_i=q_{i+1}\times B + q'+q''$ and $r_i = r'+r''$ in case $r'+r''<b$, 
which gives the equalities (1,2)
\begin{enumerate}
\item $q_i=q_{i+1}\times B + (a_i \divop b) + (r_{i+1}\times B) \divop b$
\item $r_i=(a_i \modop b) + (r_{i+1}\times B) \modop b$, 
\end{enumerate}
and, in case $b \leq r'+r'' < 2b$,
\item $q_i= q_{i+1}\times B + q'' + q' + 1$ and $r_i=r'+r''-b$.
\end{itemize}


\noindent
The following algorithm uses four arrays $\arr{D}[0..B][1..B]$, $\arr{R}[0..B][1..B]$, 
$\arr{DM}[0..B][1..B]$, and $\arr{RM}[0..B][1..B]$, defined by 
$\arr{D}[a][b] \coloneqq \lfloor a/b \rfloor$, $\arr{R}[a][b] \coloneqq a \modop b$,
$\arr{DM}[a][b] \coloneqq \lfloor (a\times \largeB)/b \rfloor$,
and $\arr{RM}[a][b] \coloneqq (a\times \largeB) \modop b$.

\begin{minipage}{\almosttextwidth}
  \begin{algorithm}[H]
    \caption{Computation of division $\lf \var{a}/\var{b} \rf$ by small integer 
    ($\var{a}< \var{\largeB{}}^d$ and $0<\var{b}< \var{\largeB{}}$)}
    \begin{minipage}[t]{0.55\textwidth}
      \begin{algorithmic}[1]
\Procedure{LinPreprocDivBySmall}{\mbox{}}
\For{$\var{b} \From  1 \To \var{\largeB{}}$}
\State $\arr{D}[0][\var{b}] \gets 0$
\State  $\arr{R}[0][\var{b}] \gets 0$
\For{$\var{a} \From  1 \To \var{\largeB}$}
   \State $\arr{D}[\var{a}][\var{b}] \gets \arr{D}[\var{a}-1][\var{b}]$
   \State $\arr{R}[\var{a}][\var{b}] \gets \arr{R}[\var{a}-1][\var{b}]+1$
   \If{$\arr{R}[\var{a}][\var{b}]=\var{b}$}
   \State $\arr{D}[\var{a}][\var{b}] \gets \arr{D}[\var{a}][\var{b}]+1$
    \State $\arr{R}[\var{a}][\var{b}] \gets 0$ 
   \EndIf
\EndFor
\State $\arr{DM}[0][\var{b}] \gets 0$
\State $\arr{RM}[0][\var{b}] \gets 0$
\For{$\var{a} \From 1 \To \var{\largeB}$}
   \State $\arr{DM}[\var{a}][\var{b}] \gets \arr{DM}[\var{a}-1][\var{b}]+\arr{D}[\var{\largeB}][\var{b}]$
   \State $\arr{RM}[\var{a}][\var{b}] \gets \arr{RM}[\var{a}-1][\var{b}]+\arr{R}[\var{\largeB}][\var{b}]$
   \If {$\arr{RM}[\var{a}][\var{b}] \geq \var{b}$}
       \State $\arr{DM}[\var{a}][\var{b}] \gets \arr{DM}[\var{a}][\var{b}]+1$
       \State $\arr{RM}[\var{a}][\var{b}] \gets \arr{RM}[\var{a}][\var{b}]-\var{b}$
\EndIf
\EndFor
\EndFor
\EndProcedure
      \end{algorithmic}
    \end{minipage}
    \begin{minipage}[t]{0.4400\textwidth}
      \vspace{4em}
      \begin{algorithmic}[1]
\Procedure{DivBySmall}{$\arr{a},\var{b}$}
\State{$\var{q} \gets 0$}
\State{$\var{r}\gets 0$}
\For{$\var{i} \From \var{d-1} \Downto 0$}
\State $\var{q} \gets \var{q}\times \var{\largeB}+ \arr{D}[\arr{a}[\var{i}]][\var{b}] + \arr{DM}[\var{r}][\var{b}]$
\State $\var{r} \gets \arr{R}[\arr{a}[\var{i}]][\var{b}] + \arr{RM}[\var{r}][\var{b}]$
\If{$\var{r} \geq \var{b}$}
\State $\var{q} \gets \var{q} + 1$
\State $\var{r} \gets \var{r} - \var{b}$
\EndIf
\EndFor
\State \Return $\var{q}$
\EndProcedure
      \end{algorithmic}
 \end{minipage}
 \end{algorithm}
\end{minipage}

\bigskip \noindent
\emph{Justification:} Clearly, the $\fun{LinPreproc}$ procedure correctly computes the arrays 
$\arr{D}, \arr{R}$, $\arr{DM}$ and $\arr{RM}$ in time $O(B^2)=O(N)$.
The  $\fun{DivBySmall}$ procedure initializes the pair $(q,r)$ to $(0,0)=(q_d,r_d)$. 
Then, for each $i=d-1,\ldots, 1,0$ of the for loop, it transforms $(q,r)=(q_{i+1},r_{i+1})$ into 
$(q,r)=(q_{i},r_i)$: indeed, 
lines 5 and 6 implement exactly the equalities of the above items (1) and (2), respectively, and lines 7-9 implement the case distinction according to whether
$r$ ($=r'+r'')$ is less than $b$ or not.
Therefore, the return value of $q$ on line 10 is $q_0$, as expected.
Obviously, $\fun{DivBySmall}$ runs in constant time.

\subsection{Two fundamental ``recursive" lemmas}
Our division algorithm for the general case is based on the following two lemmas, which justify recursive processing.

Item 1 of Lemma~\ref{lemma:approx} below expresses that if the divisor $b$ is large enough then the quotient \linebreak
$\lf \lf a/\smallB{} \rf / \lc b/\smallB{} \rc \rf$ is an approximation within 1 of the desired quotient $\lf a/b \rf$ while item 2 justifies that this computation can be applied iteratively until the divisor is smaller than $\smallB{}^3$, i.e.\ becomes a ``small'' integer.

\begin{lemma}[first ``recursive" lemma]\label{lemma:approx}
Assume  $N>1$ so that $\smallB{}\coloneqq \lc N^{1/6}\rc$ is at least 2.
Let $a,b$ be integers such that $b \geq \smallB{}^3$ and $b>a/\smallB{}$. 
Let $q$ be the quotient $\lf a'/b' \rf$ where $a' \coloneqq \lf a/\smallB{} \rf$ and $b' \coloneqq \lc b/\smallB{} \rc$. 
\begin{enumerate}
\item We have the following lower and upper bounds: 
$q\leq a/b <q+2$. In other words, we have either $\lf a/b \rf=q$ or $\lf a/b \rf=q+1$.
\item Moreover, we still have $b'>a'/\smallB{}$. 
\end{enumerate}

\end{lemma}

\begin{proof}
We have
$q=\lfloor\lfloor a/\smallB{} \rfloor / \lceil b/\smallB{} \rceil \rfloor\le \lfloor a/\smallB{} \rfloor / \lceil b/\smallB{} \rceil \le (a/\smallB{})/(b/\smallB{})=a/b$, hence $q\le a/b $. 
There remains to show $a/b < q+2$.
We give a name to the difference between each rounded expression that appears in the expression of $q$ and the value it approximates:
\begin{itemize}
\item $\epsilon_1\coloneqq a/\smallB{} - \lfloor a/\smallB{} \rfloor$; $\epsilon_2 \coloneqq \lceil b/\smallB{} \rceil - b/\smallB{}$;
$\epsilon_3 \coloneqq \lf a/\smallB{} \rf / \lc b/\smallB{} \rc  - \lf \lf a/\smallB{} \rf / \lc b/\smallB{} \rc \rf$.
\end{itemize}
Of course, we have $0\leq \epsilon_i < 1$ for each $\epsilon_i$ and, by definition of $q$ and the $\epsilon_i$'s, 
\[
q= (a/\smallB{} -\epsilon_1)/(b/\smallB{}+\epsilon_2)-\epsilon_3.
\]
It comes $(q+\epsilon_3)(b/\smallB{}+\epsilon_2)=a/\smallB{}-\epsilon_1$, and then, by multiplying by $\smallB{}$,\linebreak
$qb+\epsilon_3 b+q\epsilon_2\smallB{}+\epsilon_3\epsilon_2\smallB{}=a-\epsilon_1\smallB{}$. Hence, we deduce
\[
a-qb=\epsilon_3 b+q\epsilon_2\smallB{}+\epsilon_3\epsilon_2\smallB{}+\epsilon_1\smallB{}.
\]
Since each $\epsilon_i$ is smaller than 1, we get the inequality
\begin{eqnarray}\label{eqn: a-qb<...}
a-qb<b +(q+2)\smallB{}
\end{eqnarray}
From the hypothesis $b> a/\smallB{}$ which can be written $a/b < \smallB{}$ and from the inequality
$q\le a/b$ we deduce $q < \smallB{}$, and then $q\le \smallB{}-1$, 
hence $(q+2)\smallB{}\le \smallB{}^2+\smallB{}$ which is less than $\smallB{}^3$ by the hypothesis $\smallB{}\geq 2$. 
From the hypothesis $b\geq \smallB{}^3$, one deduces $(q+2)\smallB{}<b$ 
and then by~(\ref{eqn: a-qb<...}) $a-qb<2b$, which can be rewritten $a/b<q+2$.
We have therefore proved item 1 of the lemma. 

Let us now prove item 2. From the hypothesis $b>a/\smallB{}$ we deduce 
$b'=\lc b/\smallB{} \rc \geq b/\smallB{} > (a/\smallB{})/\smallB{} \geq \lf a/\smallB{} \rf /\smallB{}=a'/\smallB{}$ and finally $b' > a'/\smallB{}$.
\end{proof}

It remains to deal with the case where the condition $b> a/\smallB{}$ of Lemma~\ref{lemma:approx} is false. The idea is to return to the conditions of application of Lemma~\ref{lemma:approx} and, more precisely, to the computation of a quotient $\lfloor a'/b \rfloor$ for an integer $a'$ such that $b>a'/\smallB{}$. This is expressed by item 2 of the following lemma.

\begin{lemma}[second ``recursive" lemma]\label{lemma:a/Kb}
Assume $0<b \le a/\smallB{}$.
Let $q\coloneqq \smallB{}\lfloor a/(\smallB{}b) \rfloor$.
\begin{enumerate}
\item Obviously, we have $q\le a/b$ and $\lfloor a/b \rfloor= q +\lfloor (a-qb)/b \rfloor$.
\item We have $a-qb<\smallB{}b$ and therefore the computation of $\lfloor a/b\rfloor$ amounts to the computation of  
$\lfloor a'/b \rfloor$ for $a'=a-qb<\smallB{}b$, that means $b>a'/\smallB{}$.
\end{enumerate}
\end{lemma}

\begin{proof}
It suffices to prove the inequality $a-qb<\smallB{}b$ of item~2. We have
$\lfloor a/(\smallB{}b) \rfloor>  a/(\smallB{}b)-1$, then $q=\smallB{}\lfloor a/(\smallB{}b) \rfloor>a/b-\smallB{}$, 
and finally $qb>a-\smallB{}b$, which gives the expected inequality.
\end{proof}

\subsection{Recursive computation of division in constant time}

The ``recursive'' lemmas 
justify recursive procedures for computing the quotient $\lf a/b \rf$, for ``polynomial" integers $a,b<\smallB{}^d$, for a constant integer $d$.

Here, we assume that the operands $a,b$ are represented in base $\smallB{}$. 
(Recall: $\smallB{}\coloneqq \lc N^{1/6} \rc$ and $\largeB \coloneqq \smallB{}^3$.)
For example, $a$ is represented by the array
$\arr{a}[0..\var{d}-1]$ of its digits \linebreak
$\arr{a}[\var{i}]\coloneqq a_i < \smallB{}$ with
$a=a_{d-1}\smallB{}^{d-1}+\cdots+a_1\smallB{}+a_0$, which is also
written $\overline{a_{d-1}\ldots a_1 a_0}$. Note that the quotients
$\lf a/\smallB{} \rf$ and $\lc b/\smallB{} \rc$ involved in
Lemma~\ref{lemma:approx} are computed in a straightforward way: we
have $\lf a/\smallB{}
\rf=a_{d-1}\smallB{}^{d-2}+\cdots+a_2\smallB{}+a_1$, that means the
$d$-digit representation of $a/\smallB{}$ is the ``right shift"
$\overline{0a_{d-1}\ldots a_2 a_1}$; moreover, the following identity
can be easily verified:
\begin{eqnarray}\label{eqn: upRoundDiv}
\lc b/\smallB{} \rc = \lf (b-1)/\smallB{} \rf +1
\end{eqnarray}

\paragraph{Dividing by an integer $b > a/\smallB{}$.}
Lemma~\ref{lemma:approx} completed by identity~(\ref{eqn: upRoundDiv})
proves that the following recursive procedure correctly computes
$\lf a/b \rf$ when 
$a$ and $b$ are ``close'', which means $b> a/\smallB{}$.

\begin{minipage}{\almosttextwidth}
  \begin{algorithm}[H]
    \caption{Computation of division for 
    ``close'' operands}
      \begin{algorithmic}[1]
\Procedure{divClose}{$\var{a},\var{b}$}
\Comment{$a/\smallB{}<b<K^d$ and $a<K^d$}
  \If{$\var{b} < \smallB{}^3$} \Comment{recall $\largeB=\smallB{}^3$}
    \State \Return \Call{DivBySmall}{$\var{a},\var{b}$}
  \Else
    \Comment{here, $b\geq \smallB{}^3$ and $b> a/\smallB{}$ : apply Lemma \ref{lemma:approx}}
 \State $\var{q} \gets \Call{divClose}{\Call{DivBySmall}{\var{a},\smallB},\Call{DivBySmall}{b-1,\smallB}+1}$
      \If{$\var{a} < (\var{q}+1)\times \var{b}$} \label{algDiv:line1}
      \State \Return $\var{q}$
      \Else
      \State \Return $\var{q}+1$
      \EndIf
 \EndIf
\EndProcedure
      \end{algorithmic}
  \end{algorithm}
\end{minipage}

\paragraph{The general division algorithm.}
Lemma~\ref{lemma:a/Kb} justifies that the following recursive
procedure which uses the previous procedure $\fun{divClose}$
correctly computes $\lf a/b \rf$ for $b>0$ in the general case.

\begin{minipage}{\almosttextwidth}
  \begin{algorithm}[H]
    \caption{Computation of division}
      \begin{algorithmic}[1]
\Procedure{divide}{$\var{a},\var{b}$} \Comment{$a,b<K^d$ and $b>0$}
  \If{$\var{b} \times \smallB > \var{a}$}
    \State \Return \Call{DivClose}{$\var{a},\var{b}$}
  \Else
    \Comment{here, $0 < b \leq a / \smallB{}$ : apply Lemma \ref{lemma:a/Kb}}
    \State $\var{q} \gets \smallB \times \Call{divide}{\var{a},\smallB \times \var{b}}$
     \label{algDiv:line2}
    \State \Return $\var{q}+\Call{DivClose}{\var{a}-\var{q}\times \var{b},\var{b}}$
     \label{algDiv:line3}
 \EndIf
\EndProcedure
      \end{algorithmic}
  \end{algorithm}
\end{minipage}

\begin{lemma}\label{lemma: division constant time} 
The time complexities of the procedures $\fun{divClose}$ and
$\fun{divide}$ are respectively $O(d^3)$ and $O(d^4)$.
\end{lemma}

\begin{proof}

The complexity of those procedures mainly depends on the number of recursive calls they make.

\medskip \noindent \emph{Complexity of} $\fun{divClose}$: If
$\smallB{}^3 \leq b < \smallB{}^d$ and $b>a/\smallB{}$ then the
execution of $\fun{divClose}(a,b)$ calls the same procedure
$\fun{divClose}$ with the second argument $b$ replaced by 
$\lc b/\smallB{} \rc \leq \smallB{}^{d-1}$, etc.  
The number of nested recursive calls until the second argument of $\fun{divClose}$ becomes less than $\smallB{}^3$ is therefore at most $d-2$. 

The part of the procedure $\fun{divClose}$ that takes the most time is
the line~6 
of its code which computes the product
${(\var{q}+1)\times\var{b}}$ in time $O(d^2)$. Since
$\fun{divClose}(a,b)$ executes line~6 
once for each recursive call and therefore less than $d$ times, its total time
complexity is $O(d^3)$.

\medskip \noindent \emph{Complexity of} $\fun{divide}$: The execution
of $\fun{divide}(a,b)$, for $a,b<\smallB{}^d$ and 
$b\leq a/\smallB{}<\smallB{}^{d-1}$, calls
$\fun{divide}(a,b\times\smallB{})$, which calls
$\fun{divide}(a,b\times\smallB{}^2)$, etc.  Due to the obvious inequality $b\times\smallB{}^d>a$,
the execution of $\fun{divide}(a,b)$ stops with the stack of recursive calls to the
procedure $\fun{divide}$ of height at most $d$.

The most time consuming part of the procedure $\fun{divide}$ 
is the line~\ref{algDiv:line2} of its code which is executed as many times as the number
of the nested recursive calls of the form \linebreak
\Call{divide}{$\var{a},\smallB{}\times \var{b}$} that it makes. This number is at most 
$d$. It is also the number of calls of the form 
\Call{divClose}{$a-q \times b,b$}, line~\ref{algDiv:line3}, each of which is executed in time $O(d^3)$.
Thus, the time of the procedure $\fun{divide}$ is $O(d^4)$.
\end{proof}

As a consequence of the previous results including Lemma~\ref{lemma:
  division constant time}, the following perhaps surprising theorem is
established.

\begin{theorem}\label{theorem: division constant time} 
Let $d\ge 1$ be a fixed integer.
There is a RAM with addition, such that, for any input integer $N$:
\begin{enumerate}
\item \emph{Pre-computation:} the RAM computes some tables in time
  $O(N)$;
\item \emph{Division:} using these pre-computed tables and reading any
  integers $a,b<N^d$ with $b>0$, the RAM computes in constant time
  $O(d^4)$ the quotient $\lfloor a/b \rfloor$ and the remainder $a
  \modop b$.
\end{enumerate}
\end{theorem}

Using the fact that the usual four operations, addition, subtraction, product and, most importantly, division, are computable in constant time on a RAM with addition, we can now also design constant time algorithms for several other important operations on integers:  square or cubic root  (rounded), exponential function and rounded logarithm, etc.

\begin{convention}
For simplicity, we will henceforth write $x/y$ to denote the Euclidean quotient 
$\fun{divide}(x,y)=\lf x/y \rf$ in the algorithms.
\end{convention}

%% file: 7_ExpLog.tex
\section{Computing exponential and logarithm in constant time}\label{sec:expLog}

The next examples of arithmetic operations (on ``polynomial" operands)
that we will prove to be computable in constant time are the logarithm
function $(x,y)\mapsto \lf\log_x y\rf$ and the exponential function
$(x,y)\mapsto x^y$ provided the result $x^y$ is also ``polynomial".

\subsection{Computing exponential}

We want to compute $z=x^y$, for all integers $x,y<N^d$ such that the result $z$ is less than $N^d$, for a constant integer $d\ge 1$. In other words, we are going to compute the function 
$\mathtt{exp}_{N,d}$ from $ [0,N^d[^2$ to $[0,N^d[\cup \{\mathtt{OVERFLOW}\}$ defined by:
\begin{equation*}
\mathtt{exp}_{N,d}: (x,y)\mapsto
\begin{cases}
x^y & \mathtt{if}\; x^y<N^d\\
\mathtt{OVERFLOW} & \mathtt{otherwise}
\end{cases}
\end{equation*}

\medskip
Here, we will use the two integers  $\largeB\coloneqq\left\lceil N^{1/2}\right\rceil$ and $L_d\coloneqq \lceil \log_2(N^d) \rceil=\lceil d\log_2 N \rceil$ computable in time $O(N)$.

\paragraph{Pre-computation.} We are going to compute two arrays:
\begin{itemize}
\item the array $\arr{BOUND}[2..\largeB-1]$ such that
$\arr{BOUND}[x]\coloneqq\mathtt{max}\{y \mid x^y<N^d\}$ or, equivalently,
$\arr{BOUND}[x]\coloneqq\lceil \log_x(N^d)\rceil-1$ (note that $\arr{BOUND}[x]<L_d$);
\item the two-dimensional array $\arr{EXP}[2..\largeB-1][1..L_d-1]$ such that
$\arr{EXP}[x][y] \coloneqq x^y$ for all integers $x,y$ such that 
$2\leq x <\largeB=\left\lceil N^{1/2}\right\rceil$ and $1\leq y \leq \arr{BOUND}[x]$.
\end{itemize}
In other words, the arrays $\arr{BOUND}$ and $\arr{EXP}$ are defined so that the following equality is true:
\[
\{(x,y,z)\in [2,\largeB[\times [0,N^d[^2 \;\mid\; x^y=z\}=
\{(x,y,z) \mid x\in [2,\largeB[ \;\land\;1\leq y \leq \arr{BOUND}[x]\; \land\; \arr{EXP}[x][y]=z\}.
\]
The following code computes the arrays $\arr{BOUND}$ and $\arr{EXP}$.
\\
\begin{minipage}{0.2\textwidth}
  ~
\end{minipage}
\begin{minipage}{0.60\textwidth}
\begin{algorithm}[H]
  \caption{Pre-computation for \fun{exponential}}
  \begin{algorithmic}[1]
    \For{$\var{x} \From 2 \To \largeB-1$}
    \State $\var{y}\gets  0$
    \State $\var{z} \gets x$
    \While{$z<N^d$}
      \State $y\gets y+1$
      \State $\arr{EXP}[\var{x}][\var{y}]\gets z$
      \State $\var{z} \gets \var{z}\times \var{x}$    
    \EndWhile
    \State $\arr{BOUND}[\var{x}] \gets \var{y}$
    \EndFor
  \end{algorithmic}
\end{algorithm}
\end{minipage}
\\
\noindent
\emph{Justification:} One easily verifies inductively that each time the algorithm arrives at line~4, we have $x^y\times x = z$, so that the assignment at line~6 means 
$\arr{EXP}[\var{x}][\var{y}]\gets \var{x}^{\var{y}}$. 
This also explains why we have $x^y<N^d\leq x^{y+1}$ at line~8, which justifies the assignment $\arr{BOUND}[\var{x}] \gets \var{y}$.

\paragraph{Linear-time complexity of the pre-computation:}
Clearly, for each $x\in [2,\largeB[$, the body of the while loop is repeated $i$ times, where $i$ is the greatest integer $i\leq \log_x (N^d)$, or, equivalently, $i=\lfloor  d \log_x N \rfloor$, which is at most $\lfloor  d \log_2 N\rfloor \leq L_d$. It follows that the pre-computation runs in time and space $O(\largeB \times L_d)=O(N^{1/2}\times d  \log_2 N)=O(N)$.

\paragraph{Main procedure:}
The following remark justifies that the $\fun{exponential}$ procedure given here below correctly computes the function $\arr{exp}_{N,d}$.

\begin{remark} 
If we have $x^y<N^d$, $x\geq 2$ and $y\geq 2d$ then we also have $2\leq x < N^{1/2}$ 
and $2d\leq y <d \log_2 N \leq L_d$.
\end{remark} 

\begin{proof}
Note that $x^y<N^d$ and $y\geq 2d$ imply $x^{2d}\leq x^y<N^d$ and consequently $x<N^{1/2}$.
Also, note that $x^y<N^d$ and $x \geq 2$ imply $2^y\leq x^y<N^d$ and therefore $y<\log_2(N^d)$.
\end{proof}

  \begin{algorithm}[H]
 \caption{Computation of \fun{exponential}}
      \begin{algorithmic}[1]
\Procedure{exponential}{$\var{x},\var{y}$}
\If{$\var{x} = 1$ or $\var{y} = 0$}
  \State \Return 1
\EndIf
\If{$\var{x} = 0$}
\State \Return 0
\EndIf
\If{$\var{y} < 2d$} \Comment{$x \geq 2$}
  \State $\var{z} \gets 1$  
  \For{$\var{i} \From{1} \To{\var{y}}$}\label{line:expLoop}
    \State $\var{z} \gets \var{x} \times \var{z}$ \Comment{$z=x^i$}
    \If{$\var{z} \geq \var{N}^d$}
      \State \Return $\mathtt{OVERFLOW}$  \Comment{$x^y\geq x^i \geq \var{N}^d$}
    \EndIf
    \EndFor
    \State \Return \var{z} \Comment{$z=x^y$}
  \EndIf
  \If{$\var{x} < \var{\largeB}$ and $\var{y} \leq \arr{BOUND}[\var{x}]$}
    \State \Return $\arr{EXP}[\var{x}][\var{y}]$ 
    \Comment{the correct value by the definition of the array $\arr{EXP}$} 
  \Else \Comment{either $y\geq 2d$ and $x\geq \largeB$, or $y>\arr{BOUND}[x]$ for $x\geq 2$; therefore $x^y \geq N^d$}
    \State \Return $\mathtt{OVERFLOW}$
  \EndIf
\EndProcedure
      \end{algorithmic}
  \end{algorithm}

\paragraph{Constant complexity of the $\fun{exponential}$ procedure:}
Clearly, the most time consuming part is the for loop, lines 8-11. It makes at most $2d$ multiplications of integers $x,z<N^d$, each in time $O(d^2)$. Therefore, the procedure runs in time $O(d^3)$.

\subsection{Computing logarithm}

We want to compute $\lf \log_x y \rf$, for all integers $x\in[2,N^d[$ and $y\in[1,N^d[$, where $d\geq 1$ is a constant integer. 

\medskip
Here again, we use the base $\largeB\coloneqq\left\lceil N^{1/2}\right\rceil$.

\subsubsection*{Preliminary lemmas}

\begin{lemma}\label{lemma ineqa}
For any integer $x\ge 2$ and any real $a\ge 2$, we have  
$\log_x (a) -1 < \log_x \lfloor a \rfloor \leq \log_x (a)$.
\end{lemma}

\begin{proof}
We have $a/x \leq a/2 \leq a-1$ and therefore
$\log_x(a)-1= \log_x(a/x)\leq \log_x(a-1)<\log_x \lfloor a \rfloor \leq \log_x (a)$.
\end{proof}

Here again, our main procedure will be recursive. It is based on the following Lemma.

\begin{lemma}\label{lemma:rec}
For all integers $x\geq 2$, $\ell\geq 1$ and $y\geq 2x^{\ell}$, we have 
\begin{equation*}
\lfloor \log_x y\rfloor =
\begin{cases}
\lfloor \log_x \lfloor y/x^{\ell}\rfloor\rfloor +\ell & (1)\\
\mathtt{or} &\\
\lfloor \log_x \lfloor y/x^{\ell}\rfloor\rfloor +\ell +1& (2)
\end{cases}
\end{equation*}
\end{lemma}

\begin{proof}
Lemma~\ref{lemma ineqa} applied to $a=y/x^{\ell}\geq 2$ 
gives $\log_x (y/x^{\ell}) -1 < \log_x \lfloor y/x^{\ell} \rfloor \leq \log_x (y/x^{\ell})$.
 Adding~$\ell$ to each member of those inequalities yields 
 \[
 \log_x (y) -1 < \log_x \lfloor y/x^{\ell} \rfloor + \ell \leq \log_x (y).
 \]
 This can be rewritten:
  $\log_x \lfloor y/x^{\ell} \rfloor + \ell \leq \log_x (y) < \log_x \lfloor y/x^{\ell} \rfloor + \ell +1$.
  Taking the integer parts, we obtain
  \[
 \lfloor\log_x \lfloor y/x^{\ell} \rfloor \rfloor+ \ell \leq \lfloor\log_x (y) \rfloor \leq 
 \lfloor\log_x \lfloor y/x^{\ell} \rfloor \rfloor+ \ell +1.
 \]
\end{proof}

How to determine which case (1) or (2) of Lemma~\ref{lemma:rec} is true? 
The following lemma completes Lemma~\ref{lemma:rec} by giving a simple criterion.

\begin{lemma}\label{lemma:choice}
Let $s\coloneqq \lfloor \log_x \lfloor y/x^{\ell}\rfloor\rfloor +\ell$ for integers $x\geq 2$, $\ell\geq 1$ and $y\geq 2x^{\ell}$. 
Lemma~\ref{lemma:rec} says \linebreak
(1) $\lfloor \log_x y\rfloor=s$ or (2) $\lfloor \log_x y\rfloor=s+1$.
If $x^{s+1}>y$ then $\lfloor \log_x y\rfloor=s$. Otherwise,  $\lfloor \log_x y\rfloor=s+1$.
\end{lemma}

\begin{proof}
(1) means $x^s\leq y < x^{s+1}$ and (2) means $x^{s+1}\leq y < x^{s+2}$. 
So, $x^{s+1}>y$ implies (1), and $x^{s+1}\leq y$ implies (2).
\end{proof}

\paragraph{Principle of the algorithm and pre-computation in linear time.}

For each $x\in [2,\largeB]$, let us define 
$\mathtt{L}[x]\coloneqq \lceil \log_x \largeB \rceil=\mathtt{min}\{z\in\mathbb{N}\mid x^z\geq \largeB\}$. 
$x \leq \largeB$ implies $\arr{L}[x]\geq 1$.
Clearly, the array $\arr{L}[2..\largeB]$ can be pre-computed in time $O(\largeB\log_2 \largeB)=O(N)$ 
by the following algorithm where an invariant of the while loop is $y=x^{z}$.

\begin{minipage}{0.2\textwidth}
  ~
\end{minipage}
\begin{minipage}{0.6\textwidth}
\begin{algorithm}[H]
  \caption{Computation of the array $\arr{L}$}
  \begin{algorithmic}[1]
\For{$\var{x} \From 2 \To \largeB $}
  \State $\var{y} \gets \var{x}$
  \State $\var{z} \gets 1$
  \While{$\var{y} < \largeB$}                 \Comment{$y=x^{z}< \largeB$}
    \State $\var{y}\gets \var{x} \times \var{y}$
    \State $\var{z} \gets \var{z}+1$
  \EndWhile                                           \Comment{$x^{z-1} < \largeB \leq y=x^{z}$}
  \State $\arr{L}[\var{x}] \gets \var{z}$
 \EndFor
  \end{algorithmic}
\end{algorithm}
\end{minipage}
\\


The algorithm that computes $\lfloor\log_x y\rfloor$, for integers $x\in[2,N^d[$ and $y\in[1,N^d[$, divides into the following four cases:
\begin{enumerate}
\item $x>y$: then $\lfloor\log_x y\rfloor=0$;
\item $x\leq y < 2\largeB$: then $\lfloor\log_x y\rfloor$ will be computed by a simple look up in a pre-computed array $\mathtt{LOGAR}[2..2\largeB-1][2..2\largeB-1]$ defined by 
$\mathtt{LOGAR}[x][y] \coloneqq \lfloor\log_x y\rfloor$, for $2\leq x \leq  y<2\largeB$;
\item $x\leq y$ and $2B\leq y < 2 x^{L[x]}$: 
then we have $x\le y < 2 x^{L[x]} <2 x^{1+\log_x B} = 2xB$; 
this implies $1\leq \lf y/x \rf < 2B$, and because of the identity\footnote{
Indeed, $z=\lf \log_x y \rf$ means $x^z\leq y < x^{z+1}$, which implies 
$x^{z-1}\leq \lf y/x \rf < x^{z}$ and therefore \linebreak
$z-1 \leq \lf \log_x \lf y/x \rf \rf< z$. That means $z=\lf \log_x \lf y/x \rf \rf +1$.} 
$\lf \log_x y \rf = \lf \log_x \lf y/x \rf \rf+1$, the value of $\lf \log_x y \rf$ is obtained by computing $\lf \log_x \lf y/x \rf \rf$ according to cases~1 or~2;
\item $y \geq 2 x^{L[x]}$: then apply Lemma~\ref{lemma:choice} recursively with $\ell\coloneqq L[x]$ (noting that the required assumptions $x\geq 2$, $\ell\geq 1$ and $y\geq 2x^{\ell}$ are satisfied).
   
\end{enumerate}

We claim that the following code computes the array $\arr{LOGAR}$:

\begin{minipage}{0.2\textwidth}
  ~
\end{minipage}
\begin{minipage}{0.6\textwidth}
\begin{algorithm}[H]
  \caption{Computation of the array $\arr{LOGAR}$}
  \begin{algorithmic}[1]
\For{$\var{x} \From 2 \To 2\times \largeB-1$}
\State $\var{y} \gets \var{x}$ 
\State $\var{z} \gets 1$ 
\State $\var{t} \gets \var{x} \times \var{x}$ 
\While{$\var{y} < 2\times\largeB$}
 \While{$\var{y} < \var{t}$}\label{line:logarwhile}
   \State $\arr{LOGAR}[\var{x}][\var{y}] \gets \var{z}$
   \State $\var{y} \gets \var{y}+1$
 \EndWhile
 \State $\var{z} \gets \var{z}+1$
 \State $\var{t} \gets \var{x}\times \var{t}$
 \EndWhile
\EndFor
  \end{algorithmic}
\end{algorithm}
\end{minipage}
\\


\noindent
Indeed, it can be checked inductively that each time the execution of the algorithm reaches line~7 then we have $x^z \leq y < x^{z+1}=t$, which justifies the assignment 
$\arr{LOGAR}[\var{x}][\var{y}] \gets \var{z}$.
It is also easy to verify that the algorithm runs in time $O(\largeB^2)=O(N)$.

\paragraph{The main algorithm.}

We assert that the following procedure $\fun{logarithm}$ 
correctly computes $\lfloor \log_x y\rfloor$, for all the integers $x\in[2,N^d[$ and $y\in[1,N^d[$.

\begin{algorithm}[H]
  \caption{Computation of $\fun{logarithm}$}
  \begin{algorithmic}[1]
\Procedure{logarithm}{\var{x},\var{y}}
 \If{$\var{y} < \var{x}$} \Comment{case 1}
     \State \Return 0
 \EndIf
 \If{$\var{y} < 2\times\largeB$} \Comment{case 2}
     \State \Return $\arr{LOGAR}[\var{x}][\var{y}]$
 \EndIf
   \If {$\var{y} < 2\times \Call{exponential}{\var{x},\arr{L}[x]}$}
   \Comment{
   case 3: $1\leq  y/x < 2B$}
        \State \Return $\Call{logarithm}{\var{x},\var{y} / \var{x}}+1$ 
   \Else  \Comment{case 4 ($y \geq 2 x^{L[x]}$): apply Lemma~\ref{lemma:choice} with $\ell=L[x]$}
 \State $\var{s} \gets \fun{logarithm}(\var{x}, \var{y} / \Call{exponential}{\var{x},\arr{L}[\var{x}]}) + \arr{L}[\var{x}]$     
   \If{$\Call{exponential}{\var{x},\var{s}+1} > \var{y}$}
     \State \Return $\var{s}$
   \Else
     \State \Return $\var{s}+1$
   \EndIf
 \EndIf 
\EndProcedure       
  \end{algorithmic}
\end{algorithm}

\noindent
\emph{Justification:} Obviously, the  $\fun{logarithm}$ procedure is correct for $y<x$ (case 1: lines 2-3) and for $x\leq y<2\largeB$ (case 2: lines 4-5). 
Lines 6-7 treats the case 3, $x\leq y < 2 x^{L[x]}$ (recall $L[x]=\lc \log_x B \rc$), which implies $1\leq y/x < 2B$ : we therefore return to cases 1 or 2 by taking $y/x$ instead of $y$. This justifies line 7.
The case $y \geq 2 x^{L[x]}$ (case 4: lines 8-13) is justified by Lemma~\ref{lemma:choice} by taking $\ell\coloneqq L[x]$ and noting that the required assumptions $x\geq 2$, $\ell\geq 1$ and $y\geq 2x^{\ell}$ are satisfied.

\paragraph{Constant time and space complexity of the $\fun{logarithm}$ procedure.}
Let us first note that all the numbers (operands, intermediate values, results) manipulated by the procedure are less than $N^{2d}$ and that all the arithmetic operations it uses can be computed in constant time and space. In particular, this holds for the ``exponential" expressions $2x^{\mathtt{L}[x]} < 2x \largeB\leq N^{d+1}$ and $x^{s+1}\leq x y < N^{2d}$, 
which involve the function $\mathtt{exp}_{N,2d}$ of the previous subsection, computed by the $\fun{exponential}$ procedure in constant time and space. Therefore, it suffices to prove that the $\fun{logarithm}$ procedure performs at most a constant number of recursive calls.
We have $x^{\arr{L}[x]}\geq \largeB$, by definition of $\arr{L}[x]$. 
Therefore, for $y<N^d\leq \largeB^{2d}$, the algorithm cannot divide $y$ by 
$x^{\arr{L}[x]}\geq \largeB$ more than $2d-1$ times. This implies that it makes at most $2d-1$ recursive calls (by case 4) before reaching a value of $y$ less than $x$ (basic case~1), or less than $2\largeB$ (basic case~2), or less than $2x^{L[x]}$ (case~3, which performs one recursive call leading to case 1 or case 2). 
Thus, for a fixed~$d$, the $\fun{logarithm}$ procedure runs in constant time and space.


\paragraph{Open problem: computing the modular exponentiation.}
Obviously, the functions \linebreak $(x,y,z)\mapsto x+y \mod z$, and
$(x,y,z)\mapsto x\times y \mod z$, for $x,y<N^d$ and $1\leq z \leq
N^d$, are computable in constant time.  Thus, a natural candidate for
being computable (or not) in constant time is the following version of
the exponential function, widely used in cryptography: the function
$\fun{ExpMod}$ from $[0,N^d[^2\times [1,N^d]$ to $ [0,N^d[$ defined,
        for $x,y<N^d$ and $1\leq z \leq N^d$, by
\[
\Call{ExpMod}{x,y,z}\coloneqq x^y \mod z.
\] 
Now, we cannot prove that $\fun{ExpMod}(x,y,z)$ is computable in
constant time, even when fixing $x=c$ or $z=N^d$, for constant integers $c>1$, $d\geq1$. 
Indeed, it seems likely that even the following functions are
not computable in constant time after a linear preprocessing, for $x,y,z<N$:
\[
(y,z)\mapsto 2^y \mod z
\hspace{3em}
(x,y)\mapsto x^y \mod N^2 
\hspace{3em}
\mathrm{or}\; \mathrm{even}\;\; \;(x,y)\mapsto x^y \mod N.
\]

Intuitively, the irregularity of the modulo function applied to a
large number -~an exponential~- prevents us from designing an
induction by approximation as we have done for division and logarithm
and as we will for the square root and the other roots in the next
section.

%% file: 6_Roots.tex
\bigskip
\section{Computing the square root and other roots in constant time}\label{sec:roots}

In this section, we will show that, for a fixed integer $c\ge 2$, the
$c$-th root of a ``polynomial'' integer can be computed in constant
time.
For ``small'' integers (i.e.\ less than $N$), it is easy to compute
during a linear time preprocessing a table giving the square roots or
the $c$-th roots $\lfloor x^{1/c} \rfloor$, for any fixed integer
$c\ge 2$ but here we will show that we can generalize to polynomial
integers using a discrete adaptation of Newton's approximation method.

\subsection{The $c$-th root algorithm}\label{subsec:croot algo}
Let $c\ge 2$ be a fixed integer. Let us explain our strategy to compute the
$c$-th root of an integer $x<N^d$, for a constant integer $d$, i.e.
the function $x\mapsto \lfloor x^{1/c} \rfloor$ on the domain
$[0,N^d[$.

Similarly to the division method, our algorithm will distinguish
``small values'' for which the $c$-th root will be pre-computed and
``large values'' for which we will compute an approximation (using
recursively the $\fun{cthRoot}$ function with smaller values) before
improving this approximation.

The preprocessing here is similar to that of the division: we have a bound
$\var{M}'$ and we populate an array $\arr{cthRoot}$ such that
$\arr{cthRoot}[\var{x}] = \lf \var{x}^{1/c}\rf $ for all
$0\leq \var{x}<\var{M}'$. To compute the $c$-th root for an $\var{x}$ below
this $\var{M}'$ limit, we return its pre-computed value. For an $\var{x}$ over
this $\var{M}'$ limit, we first compute $\var{t}$ which is the floored
$c$-th root of the Euclidean quotient $\lf \var{x}/ \var{K}^c \rf$.
Note that the flooring happens twice: we floor the quotient
$\var{x}/ \var{K}^c$ and we floor the $c$-th root of $\var{x}/ \var{K}^c$. 
Overall, we only know that the value $\var{x}^{1/c}$ is in the range 
$[\var{t}\times \var{K}, (\var{t}+1)\times \var{K}[$.
As we will show below, the following two points are essential:
\begin{itemize}
\item by a careful choice of $\var{K}$ and $\var{M}'$, this is a sufficiently good approximation and it suffices to update this approximation by using two steps of Newton's method and one step of checking whether the answer is $\var{g}$ or $\var{g}-1$ where 
$\var{g}$ is our current approximation; 
\item the values for $\var{K}$ and $\var{M}'$ can also be chosen such that 
$\var{M}'=O(N)$, which allows a linear preprocessing, and $\var{K}=\Theta(N^{1/(2c)})$, which ensures that the procedure executes in constant time.
\end{itemize}

We know present the algorithm whose
correctness is justified in the next paragraphs.

\begin{minipage}{\almosttextwidth}
  \begin{algorithm}[H]\label{algo croot}
    \caption{Computation of $x\mapsto \lf x^{1/c}\rf$ using the pre-computed constants $\var{M}'$
    and $\var{K}$
    }
    \begin{minipage}[t]{0.49\textwidth}

    \begin{algorithmic}[1]
    \Procedure{LinPreprocCthRoot}{$\mbox{}$} 
       \State $\arr{cthRoot} \gets $ an array of size $\var{M}'$
       \State $\arr{cthRoot}[0] \gets 0$
       \For {$\var{x} \From 1 \To \var{M}'-1 $}
         \State $\var{s} \gets \arr{cthRoot}[\var{x}-1]$
         \If {$\var{x} < (\var{s}+1)^c$} 
           \State $\arr{cthRoot}[\var{x}] \gets \var{s}$
         \Else 
           \State $\arr{cthRoot}[\var{x}] \gets \var{s}+1$
         \EndIf
       \EndFor
     \EndProcedure
\end{algorithmic}
\vspace{1em}
      \begin{algorithmic}[1]
\Procedure{improve}{$\var{g},\var{x}$}
  \If {$\var{g}^c \leq \var{x} < (\var{g}+1)^c $}
  \State \Return $\var{g}$ 
  \Else
  \State \Return $(\var{x}+(c-1) \times \var{g}^c)/(c \times \var{g}^{c-1})$
  \EndIf
\EndProcedure
      \end{algorithmic}
    \end{minipage}
    \begin{minipage}[t]{0.500\textwidth}
    \vspace{0.5em}
      \begin{algorithmic}[1]
\Procedure{cthRoot}{$\var{x}$}
  \If {$ \var{x} < \var{M}' $}
   \State \Return $\arr{cthROOT}[\var{x}]$
   \Else
  \State  $\var{s} \gets \var{x} / \var{K}^c$
  \State  $\var{t} \gets \fun{cthRoot}(\var{s})$ 
  \State  $\var{g}_0 \gets \var{t} \times \var{K}$
  \State  $\var{g}_1 \gets  \fun{improve}(\var{g}_0,\var{x})$
  \State  $\var{g}_2 \gets  \fun{improve}(\var{g}_1,\var{x})$
  \If {$ \var{g}_2^c \leq \var{x} < (\var{g}_2+1)^c$} 
  \State  \Return $\var{g}_2$ 
  \Else
  \State  \Return $\var{g}_2-1$
  \EndIf
  \EndIf
\EndProcedure
      \end{algorithmic}
    \end{minipage}
  \end{algorithm}
\end{minipage}

\paragraph{Setting the integers $\var{K}$ and $\var{M}'$:}
Here, we will simply use $\var{K}\coloneqq\left\lceil N^{1/(2c)}\right\rceil$, which is the integer~$\var{K}$ such that
$(\var{K}-1)^{2c}< N \leq \var{K}^{2c}$, and let $\var{M}\coloneqq \var{K}^{2c}$. 
(Note that $\var{K}$ and $\var{M}$ can be easily computed in time $O(N)$.) 
The value of $\var{M}'$ will be set as 
$\var{M}'\coloneqq\mathtt{max}(\var{M},c_0^{2c}+c_0^{c})$, for $c_0$ a constant specified in the next paragraph, which only depends on $c$. 
Note that, for any given $c$, we have
$\var{M}=\var{M}'$ for any~$N$ large enough, which
justifies the linearity of the preprocessing, and we have
$\var{K}^c \geq N^{1/2}$, which, considering lines 5-6 of the $\fun{cthRoot}$ procedure, will justify that the recursion depth is at most
$2d$ for an integer $x<N^d$ (and thus the $\fun{cthRoot}$ procedure is constant
time for any fixed~$d$).


\paragraph{The mysterious constant $c_0$:} We define the constant integers 
$c_0\coloneqq \left\lceil \left\lceil\sqrt{6c-12}\;\right\rceil\times
(c^2-1)/(6c) \right\rceil$ and $c_0^{2c}+c_0^{c}$. This allows us to
define $M'\coloneqq\mathtt{max}(M,c_0^{2c}+c_0^{c})$. The role of the
strange constants~$c_0$ and $c_0^{2c}+c_0^{c}$ will be elucidated at
the end of the proof of the algorithm.  Note that for $c=2$, we have
$c_0=0$ and therefore $M'=M$, and for $c=3$, we have $c_0=2$ and
$c_0^{2c}+c_0^{c}=72$.  More generally, for any fixed $c$, the values
$c_0$ and $c_0^{2c}+c_0^{c}$ are explicit integers.

\subsection{Correctness and complexity of the $c$-th root algorithm}\label{subsec:croot justif}
First, note the inequality, for all $a\ge 0$ and $c\geq 1$:
\begin{equation}
(a+1)^c+1 \leq \left(a+\dfrac{c+1}{c}\right)^c 
\label{inequal}
\end{equation}
It follows from the (in)equalities
$\left(a+\frac{c+1}{c}\right)^c=\left((a+1)+\frac{1}{c}\right)^c \ge (a+1)^c+c(a+1)^{c-1}\times \frac{1}{c} = (a+1)^c+(a+1)^{c-1}\ge (a+1)^c+1$.

When $x < M'$, the $\fun{cthRoot}$ procedure returns $\arr{cthROOT}[x]$, which is the correct value. Otherwise, we have $x\ge K^{2c}$ and $x\ge c_0^{2c}$.
The integers $s$ and $t$ will be such that \linebreak
$t^c \leq s < (t+1)^c$. 
Because of $t^c\le s\le x/K^c <s+1$ and $g_0=tK$ and inequality~\eqref{inequal} for $a=t$,
we obtain: $g_0^c=t^c K^c \leq x <(s+1)K^c <
(\left(t+1 \right)^c+1)K^c \leq \left(t+\dfrac{c+1}{c}\right)^c K^c = \left(tK+\dfrac{c+1}{c}K\right)^c$ 
and then $g_0^c \leq x < \left(g_0+\dfrac{c+1}{c}K\right)^c$.
Hence, $g_0 \leq x^{1/c} < g_0+\dfrac{c+1}{c} K$.
Let $\epsilon_0$ denote the real number such that 
\begin{equation}\label{epsilon0}
x^{1/c} = g_0 \times (1+\epsilon_0). 
\end{equation}
We have $0 \leq \epsilon_0 < \dfrac{K(c+1)}{cg_0} = \dfrac{c+1}{ct}$.

Starting from the approximate value $g_0$ of $x^{1/c}$, let us apply Newton’s method: let $g_1$ be the abscissa of the point of ordinate 0 of the tangent to the point $(g_0,f(g_0))$ of the function \linebreak
$X\mapsto f(X)=X^c-x$. This means that $g_1$ is the number satisfying the equation:
\begin{equation}\label{deriv}
f'(g_0)=\dfrac{0-f(g_0)}{g_1-g_0}
\end{equation}
This gives $c g_0^{c-1}=\dfrac{x-g_0^c}{g_1-g_0}$, and (multiplying by $g_1-g_0$),
$c g_1 g_0^{c-1}-cg_0^c=x-g_0^c$, and therefore $c g_1 g_0^{c-1} =(c-1)g_0^c + x$, and finally
\begin{equation}\label{Newton g1}
g_1 = \dfrac{1}{c}\left((c-1)g_0+\dfrac{x}{g_0^{c-1}}\right) 
\end{equation}
From~(\ref{Newton g1}) and~(\ref{epsilon0}) we get  
$g_1=\dfrac{(c-1)g_0+(1+\epsilon_0)^c\times g_0}{c} = 
g_0 (1+\epsilon_0) \times \dfrac{c+ c\epsilon_0 +  (1+\epsilon_0)^c -1 - c\epsilon_0}{c+c\epsilon_0} \leq \linebreak
x^{1/c} \times  \left(1+\dfrac{\sum_{i=2}^{c}\binom{c}{i}\epsilon_0^i}{c}\right)=
 x^{1/c} \times  \left(1+\epsilon_0^2 \times\dfrac{\sum_{i=2}^{c}\binom{c}{i}\epsilon_0^{i-2}}{c}\right)$. We deduce
\begin{equation}\label{equ: cth root bounded} 
 x^{1/c}\leq g_1 \leq x^{1/c} \times  \left(1+\epsilon_0^2/c \times \left(\sum_{i=2}^{c}\binom{c}{i}\epsilon_0^{i-2}\right)\right).
\end{equation}
 Let us compare two consecutive summands $S_i,S_{i+1}$ of the sum 
 $\sum_{i=2}^{c}S_i$ where $S_i \coloneqq \binom{c}{i}\epsilon_0^{i-2}$. \linebreak
 For $2\leq i < c$, we have 
 $S_{i+1}=S_i\times \dfrac{(c-i)\epsilon_0}{i+1}\leq S_i\times \dfrac{(c-2)\epsilon_0}{3}$,
 from which we deduce \linebreak
 $S_{i+1}/S_i\leq\dfrac{(c-2)\epsilon_0}{3} \leq \dfrac{(c-2)(c+1)}{3ct}$
 because of $\epsilon_0 \leq \dfrac{c+1}{ct}$. The quotient $S_{i+1}/S_i$
 is at most 1 if we assume 
 \begin{center}
 $t\geq c_1$ for the constant $c_1\coloneqq \dfrac{(c-2)(c+1)}{3c}$.
 \end{center}
 Under this assumption, each of the $c-1$ summands $S_i=\binom{c}{i}\epsilon_0^{i-2}$, $2\leq i \leq c$, is less than or equal to $S_2=\binom{c}{2}=\frac{c(c-1)}{2}$, so that we obtain
 $\dfrac{\sum_{i=2}^{c}\binom{c}{i}\epsilon_0^{i-2}}{c}\leq \dfrac{1}{c}\times (c-1) \times\dfrac{c(c-1)}{2}=\dfrac{(c-1)^2}{2}$ and, as another consequence of $\epsilon_0\leq\dfrac{c+1}{ct}$,
 \begin{equation}
 \epsilon_0^2 \times \dfrac{\sum_{i=2}^{c}\binom{c}{i}\epsilon_0^{i-2}}{c} \leq 
 \dfrac{(c^2-1)^2}{2c^2t^2}.
 \label{bound}
\end{equation}
Note that in the $\fun{cthRoot}$ procedure, the variable $\var{g}_1$ represents the integer part $\lfloor g_1\rfloor$ of $g_1$. 
We know that $x^{1/c} \leq g_1$ but it is possible that
$\lfloor g_1 \rfloor \le x^{1/c}$. In this case, we have $\lfloor g_1 \rfloor \le x^{1/c}\leq g_1$, 
which  implies $\var{g}_1=\lfloor g_1 \rfloor= \lfloor x^{1/c} \rfloor$, that means $\var{g}_1$  is the value we are looking for. 
Otherwise, we have $\lfloor g_1 \rfloor > x^{1/c}$ and we denote $\epsilon_1$ the positive number such that 
$\lfloor g_1 \rfloor=x^{1/c} \times(1+\epsilon_1)$.
From the inequalities $x^{1/c}\leq x^{1/c}(1+\epsilon_1)\leq g_1\leq x^{1/c} \times  \left(1+\epsilon_0^2/c \times \left(\sum_{i=2}^{c} \binom{c}{i}\epsilon_0^{i-2} \right) \right)$ and~\eqref{bound}, we deduce 
\begin{equation}
0\leq \epsilon_1 \leq  \dfrac{(c^2-1)^2}{2c^2t^2}
\label{epsilon_1}
\end{equation}

\noindent
When  $\var{g}_1=\lfloor g_1 \rfloor > \lfloor x^{1/c} \rfloor$ 
we iterate Newton's method by using the approximate value $\lfloor g_1\rfloor$
and, by replacing $g_0$ by $\lfloor g_1\rfloor$ in expression~(\ref{Newton g1}), we obtain the integer part $\var{g}_2= \lfloor g_2 \rfloor$ of the number 
\[
g_2 \coloneqq 
\dfrac{1}{c} \times \left((c-1) \lfloor g_1\rfloor + \dfrac{x}{\lfloor g_1\rfloor^{c-1}}\right).
\]
Then, we have
$g_2 =\dfrac{1}{c} \times \left((c-1)x^{1/c}(1+\epsilon_1) + 
\dfrac{x}{x^{(c-1)/c}\times (1+\epsilon_1)^{c-1}}\right) =
\dfrac{x^{1/c}}{c} \times\;\dfrac{(c-1)(1+\epsilon_1)^c+1}{(1+\epsilon_1)^{c-1}}=
x^{1/c}\times\;\dfrac{(c-1)(1+ c\epsilon_1+\sum_{i=2}^{c}\binom{c}{i}\epsilon_1^i)+1}{c(1+\epsilon_1)^{c-1}}= 
x^{1/c} \times\;\dfrac{c+ c(c-1)\epsilon_1+(c-1)(\sum_{i=2}^{c}\binom{c}{i}\epsilon_1^i)}{c+ c(c-1)\epsilon_1+c(\sum_{i=2}^{c-1}\binom{c-1}{i}\epsilon_1^i)}\leq \linebreak
x^{1/c} \times \left(1+\dfrac{c-1}{c}\times \left(\sum_{i=2}^{c}\binom{c}{i}\epsilon_1^{i-2}\right)\times \epsilon_1^2 \right)$.
Finally, we obtain
\begin{equation}
x^{1/c} \leq g_2 \leq 
x^{1/c} \times \left(1+\dfrac{c-1}{c}\times \left(\sum_{i=2}^{c}\binom{c}{i}\epsilon_1^{i-2}\right)\times \epsilon_1^2 \right)
\label{g2 inf}
\end{equation}
With the same argument as above, the quotient between two consecutive summands of the sum 
$\sum_{i=2}^{c}\binom{c}{i}\epsilon_1^{i-2}$ is not greater than $\dfrac{(c-2)\epsilon_1}{3}$, which is at most $\dfrac{(c-2)(c^2-1)^2}{6c^2t^2}$ because of~\eqref{epsilon_1}.
This quotient is not greater than 1 if we assume $t^2\geq \dfrac{(c-2)(c^2-1)^2}{6c^2}$,
that means 
\begin{center}
$t\geq  \dfrac{\sqrt{6c-12}\times (c^2-1)}{6c}$ or, if we assume, more strongly, $t \geq c_0$, by definition of $c_0$.
\end{center}
Therefore, if we assume $t \geq c_0$, then each of the $c-1$ summands $\binom{c}{i}\epsilon_1^{i-2}$, $2\leq i \leq c$, is less than or equal to the first summand $\binom{c}{2}=\dfrac{c(c-1)}{2}$ so that one deduces the inequality
\[
 \dfrac{c-1}{c} \times \left(\sum_{i=2}^{c}\binom{c}{i}\epsilon_1^{i-2}\right)\leq 
 \dfrac{c-1}{c}\times (c-1) \times\dfrac{c(c-1)}{2}=\dfrac{(c-1)^3}{2}
 \] 
 and, as another consequence of~\eqref{epsilon_1}, 
\begin{equation}
\dfrac{c-1}{c} \times \left(\sum_{i=2}^{c}\binom{c}{i}\epsilon_1^{i-2}\right)\times \epsilon_1^2 \leq 
\dfrac{(c-1)^3}{2} \times \dfrac{(c^2-1)^4}{4c^4t^4}=
\dfrac{(c-1)^7(c+1)^4}{8c^4t^4}.
\label{bound2}
\end{equation}
From the inequalities~\eqref{g2 inf} and~\eqref{bound2}, we deduce
\begin{center}
$x^{1/c} \leq g_2\leq x^{1/c} \times \left(1+ \dfrac{c_2}{t^4} \right)$ 
where 
$c_2 \coloneqq \dfrac{(c-1)^7(c+1)^4}{8c^4}$.
\end{center}

\begin{remark}
Recall that those inequalities 
hold \emph{under the hypotheses} $t\ge c_1=(c-2)(c+1)/(3c)$ and $t\ge c_0$. Also, note that the hypothesis~$t\ge c_0$ is sufficient since we have $c_0 > c_1$ for $c\ge 2$, which is easy but tedious to prove.  
\end{remark}

\paragraph{Proof of $x^{1/c} \leq g_2< x^{1/c}+1$ (under the hypothesis~$t\geq c_0$).}
We have $x^{1/c} \leq g_2\leq x^{1/c}+\epsilon_3$ 
for $\epsilon_3 \coloneqq \dfrac{c_2 x^{1/c}}{t^4}$. 
Besides, we have 
$t=\left\lfloor \left\lfloor \dfrac{x}{K^c}\right\rfloor^{1/c}\right\rfloor 
> \left(\dfrac{x}{K^c}\right)^{1/c}-1 = \dfrac{x^{1/c}-K}{K}$
(because we have 
$\left\lfloor \left\lfloor a \right\rfloor^{1/c}\right\rfloor= \left\lfloor a^{1/c} \right\rfloor >a^{1/c}-1$, for each positive real number $a$).
We deduce $\epsilon_3= \dfrac{c_2 x^{1/c}}{t^4}\leq 
\dfrac{c_2 K^4 x^{1/c}}{(x^{1/c}-K)^4}=
c_2\times \dfrac{K^4((x^{1/c}-K)+K)}{(x^{1/c}-K)^4}= 
c_2\times\left(\dfrac{K^4}{(x^{1/c}-K)^3}+ \dfrac{K^5}{(x^{1/c}-K)^4}\right)$.
Recall the hypothesis $x\geq K^{2c}$ for $K=\left\lceil N^{1/(2c)} \right\rceil$. 
This implies $x^{1/c}-K \ge K^2-K$ and \linebreak
$\epsilon_3 \leq c_2\times  \left(\dfrac{K^4}{(K^2-K)^3}+ \dfrac{K^5}{(K^2-K)^4}\right)=
c_2\times \left(\dfrac{K}{(K-1)^3}+ \dfrac{K}{(K-1)^4}\right)= 
c_2\times \dfrac{K^2}{(K-1)^4}$. \linebreak
We obtain $\epsilon_3 \leq c_2\times \dfrac{K^2}{(K-1)^4}$.\\
\noindent
For $K\geq 1+2 \sqrt{c_2}$~\footnote{That means
$M=K^{2c}\geq \left(1+(c-1)^{7/2}\times(c+1)^2 / (\sqrt{2}\times c^2) \right)^{2c}$. 
For instance, it means $M\geq 46$ for $c=2$ and $M\geq 12~405~543$ for $c=3$.},
we get 
$\dfrac{K^2}{(K-1)^4}=\dfrac{1}{(K-1)^2}+\dfrac{2}{(K-1)^3}+\dfrac{1}{(K-1)^4} \leq
\dfrac{1}{4c_2}+\dfrac{2}{4c_2}+\dfrac{1}{4c_2}=\dfrac{1}{c_2}$, which implies
$\epsilon_3 \leq c_2 \times  \dfrac{K^2}{(K-1)^4} \leq 1$. 
Therefore, we obtain $ x^{1/c}\leq g_2 \leq  x^{1/c}+1$, or, equivalently,
$g_2-1 \leq  x^{1/c}\leq g_2$.

\medskip
This implies that we have
either $\left\lfloor x^{1/c}\right\rfloor=\left\lfloor g_2\right\rfloor=\var{g}_2$ or 
$\left\lfloor x^{1/c}\right\rfloor=\left\lfloor g_2\right\rfloor-1=\var{g}_2-1$. 
This justifies the last line of the $\fun{cthRoot}$ procedure and completes the proof of its correctness \emph{if} we can prove that the hypothesis~$t\geq c_0$ \emph{always} holds when $x\geq M'$.

\paragraph{Why $x\geq M'$ implies~$t\geq c_0$.} 
Let us prove this implication by contradiction. Assume we have both $x\geq M'$ and~$t < c_0$.
We have $x\geq M=K^{2c}$ and $x\geq c_0^{2c}+c_0^c$.
 
Recall that $s=\lf x/K^c \rf$ and $t=\lf s^{1/c} \rf$ imply $x/K^c<s+1$ and $s^{1/c}<t+1$, which means $s<(t+1)^c$. 
With the hypothesis $x\geq M=K^{2c}$, this yields
$K^{c}\leq x/K^c < s+1< (t+1)^c+1$ and therefore $K^c\leq (t+1)^c$.
Finally, using the hypothesis~$t < c_0$, we get $K\leq t+1<c_0+1$ and $K\leq c_0$.

Besides, from the hypothesis $x\geq c_0^{2c}+c_0^c$ and from $K\leq c_0$, we deduce, for $s=\lfloor x/K^c \rfloor$, the inequalities:  
$s> x/K^c -1\geq (c_0^{2c}+c_0^c)/K^c -1\geq (c_0^{2c}+c_0^c)/c_0^c -1=c_0^c$, hence $s>c_0^c$ and $t=\lfloor s^{1/c}\rfloor\geq \lfloor (c_0^c)^{1/c} \rfloor=c_0$. 
Therefore, we get $t\geq c_0$, a contradiction. 

\bigskip
The proof of the correctness of the $\fun{cthRoot}$ procedure is now complete. \hfill $\square$

\paragraph{Proof of the constant time of the $\fun{cthRoot}$ procedure.}
For $x<N^d\leq K^{2cd}$, the algorithm (see lines 5-6) can divide $x$ by $K^c$ at most $2d-2$ times before reaching a value $x<M=K^{2c}$ for which the $c$-th root is given by 
$\arr{cthRoot}[x]$. This means that the $\fun{cthRoot}$ procedure makes at most $2d-2$ recursive calls and therefore, for a fixed $d$, executes in constant time. \hfill  $\square$

\medskip
Finally, we have proved that any root of any ``polynomial" integer can be computed in constant time. More precisely:
\begin{proposition}\label{prop: roots constant time} 
Let $c\ge 2$ and $d\ge 1$ be fixed integers.
There is a RAM with addition such that, for any input integer $N$:
\begin{enumerate}
\item \emph{Pre-computation:} the RAM computes some tables in time $O(N)$;
\item \emph{Computation of the root:} using these pre-computed tables and reading any integer $x<N^d$, the RAM computes in constant time the $c$-th root $\lfloor x^{1/c}\rfloor$. 
\end{enumerate}
\end{proposition}

\subsection{Generalized root: a conjecture partially resolved positively}\label{subsec:genroot}

Although we are unable to prove it, we believe that Proposition~\ref{prop: roots constant time}
can be generalized to $y$th root, for an integer variable $y<N^d$.
\begin{conjecture}\label{conj: genRoot}
For any fixed integer $d\geq 1$, the function $(x,y)\mapsto \lf x^{1/y} \rf$, where $x<N^d$ and $2\leq y <N^d$, is computable in constant time after preprocessing in $O(N)$ time.
\end{conjecture}

In this subsection, we give strong arguments in favor of this conjecture.

\begin{remark} The conjecture is true for the restriction of the function $(x,y)\mapsto \lf x^{1/y} \rf$ to the following subdomains:
\begin{enumerate}
\item $y \geq \log_2 N$: then $0\leq \lf x^{1/y} \rf<2^d$ (because for $z=\lf x^{1/y} \rf$ we have 
$z^{\log_2 N} \leq z^y \leq x < N^d= 2^{d\log_2 N}$, which implies $z<2^d$) and therefore, 
$\lf x^{1/y} \rf$ can be computed in time $O(d)$ by dichotomic search;
\item $y \leq (\log_2 N) /(12 \log_2\log_2 N)$: then the algorithm of subsection~\ref{subsec:croot algo} computing the $c$th root $\lf x^{1/c} \rf$, for $x<N^d$ and any fixed integer $c$, and its justification (subsection~\ref{subsec:croot justif}) can be slightly adapted. 
\end{enumerate}
Since the conjecture is also answered positively for $y < \log_2 N$ and $x< N^{1-\varepsilon}$, for any fixed $\varepsilon>0$, by pre-computing in time $O(N)$ an array $\arr{root}$ of dimensions
$N^{1-\varepsilon}\times \log_2 N$, in fact the only open case is 
$\log_2 N /(12 \log_2\log_2 N) < y < \log_2 N$ with $N^{1-\varepsilon}\leq x < N^d$.
\end{remark}

This yields the following algorithm, which makes free use of the exponential and (rounded) logarithm functions proven to be computed in constant time in Section~\ref{sec:expLog}.


\begin{minipage}{\almosttextwidth}
  \begin{algorithm}[H]
    \caption{Computation of $(x,y)\mapsto \lf x^{1/y}\rf$, for $x<\var{N}^d$ and 
    ($2\leq y < \var{L} / (12 \lc \log_2 \var{L} \rc)$ or $\var{L} < y < \var{N}^d$), 
   where $\var{L}\coloneqq \lfloor \log_2 \var{N} \rfloor$}
    
    \begin{minipage}[t]{0.49\textwidth}

     \begin{algorithmic}[1]
    \Procedure{LinPreprocRoot}{\mbox{}}
       \State $\var{B} \gets \mathtt{max}(\lc \var{N}^{1/4} \rc ^3,\lc \var{N}^{1/(d+1)} \rc^d + 1)$
       \State $\var{L} \gets \lfloor \log_2 \var{N} \rfloor$
       \State $\arr{RootN} \gets $ an array of dimension $\var{L}$
       \State $\arr{root} \gets $ an array of dimensions $\var{B} \times \var{L}$
       \For {$\var{y} \From 1 \To \var{L} $}
          \State $\var{r} \gets 0$
          \While {$ \var{r}^{3\times \var{y}} < \var{N}$}
              \State $\var{r} \gets \var{r}+1$
          \EndWhile
          \State $\arr{RootN}[\var{y}] \gets \var{r}$ 
          \Comment{$r=\lceil N^{1/ (3 \times y) } \rceil$}
       \EndFor
       \For {$\var{y} \From 1 \To \var{L} $}
       \State $\arr{root}[0][\var{y}] \gets 0$
       \For {$\var{x} \From 1 \To \var{B}-1 $}
         \State $\var{s} \gets \arr{root}[\var{x}-1][\var{y}]$
         \If {$\var{x} < (\var{s}+1)^\var{y}$} 
           \State $\arr{root}[\var{x}][\var{y}] \gets \var{s}$
         \Else 
           \State $\arr{root}[\var{x}][\var{y}] \gets \var{s}+1$
         \EndIf
       \EndFor
       \EndFor
     \EndProcedure
\end{algorithmic}
\vspace{1em}
      \begin{algorithmic}[1]
\Procedure{improve}{$\var{g},\var{x},\var{y}$}
  \If {$\var{g}^{\var{y}} \leq \var{x} < (\var{g}+1)^{\var{y}} $}
  \State \Return $\var{g}$ 
  \Else
  \State \Return $(\var{x}+(\var{y}-1) \times \var{g}^{\var{y}})/(\var{y} \times \var{g}^{\var{y}-1})$
  \EndIf
\EndProcedure
      \end{algorithmic}
    \end{minipage}
    \begin{minipage}[t]{0.500\textwidth}
      \begin{algorithmic}[1]
\Procedure{root}{$\var{x},\var{y}$}
  \If{$\var{y} > \var{L}$} \Comment{$0\leq \lf x^{1/y} \rf < 2^d$}
  \State $\var{r} \gets 0$
  \State $\var{u} \gets 2^d$ 
  \While {$\var{u}-\var{r}>1$}
  \State $\var{z} \gets (\var{r}+\var{u})/2$
  \If {$ \var{x} < \var{z}^{\var{y}}$}
  \State $\var{u} \gets  \var{z}$
  \Else
  \State $\var{r} \gets  \var{z}$
  \EndIf
  \EndWhile
  \State \Return $\var{r}$
  \EndIf
 
 \If {$ \var{y} < \var{L} / (12\times \lc \log_2 \var{L} \rc) $}
  \If {$ \var{x} < \var{B}$}
   \State \Return $\arr{root}[\var{x}][\var{y}]$
   \Else
  \State  $\var{s} \gets \var{x} / (\arr{RootN}[\var{y}])^\var{y}$
  \State  $\var{t} \gets \Call{root}{\var{s},\var{y}}$ 
  \State  $\var{g}_0 \gets \var{t} \times \arr{RootN}[\var{y}]$
  \State  $\var{g}_1 \gets  \fun{improve}(\var{g}_0,\var{x},\var{y})$
  \State  $\var{g}_2 \gets  \fun{improve}(\var{g}_1,\var{x},\var{y})$
  \If {$ \var{g}_2^\var{y} \leq \var{x} < (\var{g}_2+1)^\var{y}$} 
  \State  \Return $\var{g}_2$ 
  \Else
  \State  \Return $\var{g}_2-1$
  \EndIf
  \EndIf
 \EndIf
\EndProcedure
      \end{algorithmic}
    \end{minipage}
  \end{algorithm}
\end{minipage}

\paragraph{Justification of the algorithm.} The case $y>\var{L}$, for $\var{L}=\lf \log_2 N \rf$, lines 2-11 of the $\fun{root}$ procedure, is justified by the following invariant condition: $r^y \leq x < u^y$ 
(meaning $0\leq x < 2^{dy}$ at the initialization): 
therefore, the return value $r$, obtained after at most $d$ iterations of the body of the while loop, satisfies $r^y \leq x < (r+1)^y$ as required.

The case $ y < \var{L} / (12\times \lc \log_2 \var{L} \rc) $, lines 12-24 of the $\fun{root}$ procedure,
is a variant of the $\fun{cthRoot}$ procedure, where the fixed integer $c$ is replaced by the integer variable $y$. Moreover, the computations which justify the correctness of this case are the same as those of $\fun{cthRoot}$, see Subsection~\ref{subsec:croot justif}, with the following adaptations. 

\begin{itemize}
\item Instead of setting $K \coloneqq \lc N^{1/(2c)} \rc$, we set $K\coloneqq\lc N^{1/(3y)} \rc$, which we pre-compute and store as an array element: $\arr{RootN}[y]\coloneqq\lc N^{1/(3y)} \rc$.
\item The recursive case of of the $\fun{root}$ procedure is $x\ge B$, for \\
$B\coloneqq  \mathtt{max}(\lc N^{1/4} \rc ^3,\lc N^{1/(d+1)} \rc^d + 1)$, instead of $x\ge M'$, for 
$M'\coloneqq \mathtt{max}(M,c_0^{2c}+c_0^{2c})$ 
(with $M\coloneqq K^{2c}$ and 
$c_0\coloneqq \left\lceil \left\lceil\sqrt{6c-12}\;\right\rceil\times (c^2-1)/(6c)\right\rceil$).
\item Our computations continue to use the numbers $c_0,c_1,c_2$ introduced in Subsection~\ref{subsec:croot justif}, now defined by replacing $c$ by $y$. In particular, we now set
$c_0\coloneqq \left\lceil \left\lceil\sqrt{6y-12}\;\right\rceil\times (y^2-1)/(6y)\right\rceil$ and 
$c_2 \coloneqq (y-1)^7 \times (y+1)^4/(8y^4)$.
\end{itemize}
For the reasoning and computations of Subsection~\ref{subsec:croot justif} to still be valid, it suffices to check that the following three conditions hold under the assumption 
$ y < L / (12\times \lc \log_2 L \rc) $, for sufficiently large integers $y$ and $N$:
\begin{enumerate}
\item $B\geq K^{2y}$; 
\item $B\geq c_0^{2y}+c_0^{y}$; 
\item $K \geq 1 + 2 \sqrt{c_2}$.
\end{enumerate}
It is important to note that, by items 1 and 2, the condition $x\geq B$ implies \linebreak
$x\geq 
\mathtt{max}(K^{2y},c_0^{2y}+c_0^{y})$. 
This corresponds to the condition 
$x\geq M'= \mathtt{max}(K^{2c},c_0^{2c}+c_0^{c})$ of the $\fun{cthRoot}$ procedure.
Recall also that, to conclude the reasoning of Subsection~\ref{subsec:croot justif}, condition~3 ($K \geq 1 + 2 \sqrt{c_2}$) has been assumed.

\begin{proof}[Proof of item 1] Since we have $K=\lc N^{1/(3y)} \rc\geq 2$ for $N\geq 2$, we obtain
$K/2 \leq K-1 < N^{1/(3y)}$. This implies $K < 2  N^{1/(3y)}$ and then 
$K^{2y} < 2^{2y} \times N^{2/3}$.
From the assumption \linebreak
$ y < L / (12\times \lc \log_2 L \rc) $, we get 
$2^{2y} < 2^{L / (6 \lc \log_2 L \rc)}\leq 2^{(\log_2 N) / (6  \lc \log_2 L \rc)} \leq N^{1/(6 \lc \log_2 L \rc)}$ and thus $2^{2y} < N^{1/(6 \lc \log_2 L \rc)}$.
Moreover, for $N\geq 16$, we have $L=\lf \log_2 N \rf\geq 4$ and then $\lc \log_2 L \rc \geq 2$. This gives $2^{2y} < N^{1/12}$.
All in all, we get $K^{2y} < 2^{2y} \times N^{2/3}<N^{1/12}\times N^{2/3}=N^{3/4}\leq B$, and 
$B \geq K^{2y}$ as claimed.
\end{proof}

The following lemma will be useful to prove items 2 and 3. 
\begin{lemma}\label{lemma: inverse log} 
For all numbers $y,z$ and $\alpha$ such that $y\geq 2$, $z>1$ and $0<\alpha \leq 1$, we have the implication
\[
y < \alpha \times z / \log_2 z \Rightarrow y \log_2 y < \alpha \times z
\]
\end{lemma}

\begin{proof}[Proof of the lemma]
Assume $y\leq \alpha \times z / \log_2 z$. 
Using also the inequalities $\log_2\log_2 y\geq 0$ and $\log_2\alpha \leq 0$, we obtain

$\dfrac{y \log_2 y}{\log_2(y\log_2 y)}=\dfrac{y \log_2 y}{\log_2 y+\log_2\log_2 y}\leq y
\leq \dfrac{\alpha \times z} {\log_2 z}\leq \dfrac{\alpha \times z} {\log_2 z+\log_2 \alpha}
= \dfrac{\alpha \times z} {\log_2(\alpha\times z)}$, and thus

$\dfrac{y \log_2 y}{\log_2(y\log_2 y)} \leq \dfrac{\alpha \times z} {\log_2(\alpha\times z)}$. 
Noticing that the function $x\mapsto \dfrac{x}{\log_2 x}$ is strictly increasing on its domain, we obtain $y \log_2 y \leq \alpha \times z$ as required.
\end{proof}

As a direct application of Lemma~\ref{lemma: inverse log} where we take $z=L=\lf log_2 N \rf>1$ and $\alpha=1/12$, we obtain 
\begin{lemma}\label{lemma: appli}
If $ y < L / (12\times \lc \log_2 L \rc) $ then $y \log_2 y < L/12$.
\end{lemma}

\begin{proof}[Proof of item 2 ($B\geq c_0^{2y}+c_0^{y}$)] We have \\
$c_0 = \left\lceil \left\lceil\sqrt{6y-12}\;\right\rceil\times (y^2-1)/(6y)\right\rceil <
\dfrac{(\sqrt{6y-12}+1)\times (y^2-1)}{6y}+1 <
\dfrac{(\sqrt{y-2}+1/ \sqrt{6}) \times y +\sqrt{6}}{\sqrt{6}}$.
If $y\geq 6$, which implies $\sqrt{6} \leq y / \sqrt{6}$ and $2 / \sqrt{6} < 2 \leq \sqrt{y-2}$, we obtain \linebreak
$c_0 < \dfrac{(\sqrt{y-2} + 2/ \sqrt{6})\times y}{\sqrt{6}} < \dfrac{2 \sqrt{y-2} \times y}{\sqrt{6}} < y^{3/2}$. 
This implies $c_0^{2y}+c_0^{y}< 2\times c_0^{2y} < \linebreak
2\times (y^{3/2})^{2y} = 2 \times y^{3y} = 2 \times 2^{3y \log_2 y}< 2^{4 y \log_2 y}$ 
(note that $y\geq 2$ implies $y \log_2 y >1$) and therefore
$c_0^{2y}+c_0^{y}< 2^{4 y \log_2 y}$. 
Because of  $y \log_2 y < L/12$ according to Lemma~\ref{lemma: appli}, 
we get $\log_2 y < (1/12) L/y$ and then $4 y \log_2 y < (1/3) L$.
All in all, we deduce $c_0^{2y}+c_0^{y}< 2^{(1/3) L} \leq N^{1/3}<B$ and then 
$B\geq c_0^{2y}+c_0^{y}$.
\end{proof}

\begin{proof}[Proof of item 3 ($K \geq 1 + 2 \sqrt{c_2}$)] We have
$1 + 2 \sqrt{c_2} = 1 + \dfrac{(y-1)^{7/2}\times (y+1)^2}{\sqrt{2} \times y^2}$. 
Moreover, it is easy to verify that for $y\geq 6$, we have $\dfrac{(y+1)^2}{\sqrt{2} \times y^2}<1$. 
It comes $1 + 2 \sqrt{c_2} < 1 + (y-1)^{7/2} < y^{7/2}$, which implies 
$1 + 2 \sqrt{c_2} < 2^{(7/2) \log_2 y}$. Because of  $y \log_2 y < L/12$ according to Lemma~\ref{lemma: appli}, we get $\log_2 y < (1/12) L/y$ and then 
$1 + 2 \sqrt{c_2} < 2^{(7/2)\times (1/12) \times L/y} < 
2^{L/(3y)}\leq N^{1/(3y)}\leq \lc N^{1/(3y)} \rc = K$. 
We have proved $K \geq 1 + 2 \sqrt{c_2}$.
\end{proof}

\paragraph{Justification of the complexity.} In the $\fun{LinPreprocRoot}$ procedure, for each $y\in[1,L]$, the while loop, lines 8-9, does at most $N^{1/(3y)}$ iterations. 
Because of $N^{1/(3y)}\leq N^{1/3}$, the running time of the for loop, lines 6-10, is 
$O(L\times N^{1/3})=O(N)$.
Since the run time of lines 11-18 (of $\fun{LinPreprocRoot}$) is $O(L\times B)=O(N)$, all preprocessing runs in time $O(N)$.

To analyze the time of the $\fun{root}$ procedure, it suffices to determine its recursive depth. 
We have $(\arr{RootN}(y))^y=\lc N^{1/(3y)}\rc^y \geq  N^{1/3}$. 
Since we cannot divide the integer $x<N^d$ by  $(\arr{RootN}(y))^y$ 
(which is an integer greater than $ N^{1/3}$)
more than $3\times d$ times, the running time of the $\fun{root}$ procedure is $O(d)$, which is constant for fixed $d$.


\medskip
Conveniently, the notations used must be recalled and completed:
\begin{notation}
Let $L$, $\lambda$ and $B$ denote the integers $L \coloneqq \lf \log_2 N \rf$,
$\lambda\coloneqq L / (12\times \lc \log_2 L \rc)$, and, for a fixed integer $d\geq 1$,
$B\coloneqq  \mathtt{max}(\lc N^{1/4} \rc ^3,\lc N^{1/(d+1)} \rc^d + 1)$.
\end{notation}
In addition to the above algorithm, the root $\lf x^{1/y} \rf$ can also be computed in constant time for ``most'' integers $y$ in the ``remaining'' interval $[\lambda, L]$. 
This is the consequence of the following two lemmas.

\begin{lemma}\label{lemma:composeRoots}
For all positive integers $x,p,q$, the identity
$\lf x^{1/(p\times q)} \rf = \lf \lf x^{1/p} \rf^{1/q} \rf$ holds.
\end{lemma}

\begin{proof}
Let $u,v$ denote the successive roots $u \coloneqq \lf x^{1/p} \rf$ and $v \coloneqq \lf u^{1/q} \rf$.
That means $u^p \leq x \leq (u+1)^p-1$ and  $v^q \leq u \leq (v+1)^q-1$, from which we deduce
$v^{p\times q} \leq u^p \leq x \leq (u+1)^p-1 \leq ((v+1)^q-1+1)^p-1=(v+1)^{p\times q}-1$. That gives
$v^{p\times q} \leq x \leq (v+1)^{p\times q}-1$, which means $v=\lf x^{1/(p\times q)} \rf$ as expected.
\end{proof}

\begin{lemma}\label{lemma:forbidden form}
Let $d\geq 1$ be a fixed integer. Suppose $N$ is large enough for the inequality 
$\lambda^2 / d^2> L$ to hold. Let $y$ be an integer in $[\lambda, L]$, which is not of the form, called \emph{forbidden form}\footnote{In particular, for $d=1$ (resp.\ $d=2$), an integer is of the forbidden form iff it is prime (resp. it is prime or twice a prime number). 
In general, by the repartition theorem of prime numbers, there are asymptotically about $d\times L/(\ln L)$ integers of forbidden form among the integers in $[\lambda,L]$, which is only 
$O((\log N) / (\log \log N))$.},  
 $p\times q$ for a prime number $p$ and an integer $q\leq d$. Then, $y$ has a divisor $y_1$ such that 
$d<y_1< \lambda$.
\end{lemma}

\begin{proof}
Let us consider the decomposition of $y$ into prime factors, $y=\prod_{i=1}^k p_i$, with $k\geq 2$ and $p_1\geq p_2 \geq \cdots \geq p_k$. Let $i_0$ (resp.\ $i_1$) be the least index such that 
$\prod_{i=1}^{i_0} p_i >d$ (resp.\ $\prod_{i=1}^{i_1} p_i \geq \lambda$). For~$N$ large enough, we have $\lambda \geq d$, which implies $i_0\leq i_1$. We have two cases:
\begin{enumerate}
\item $i_0<i_1$: taking $y_1=\prod_{i=1}^{i_0} p_i$, we obtain (by definition of $i_0$ and $i_1$) $d<y_1< \lambda$;
\item $i_0=i_1$: then, we get $\prod_{i=1}^{i_0} p_i =\prod_{i=1}^{i_1} p_i \geq \lambda$ and $\prod_{i=1}^{i_0-1} p_i \leq d$ (by definition of $i_0$ and~$i_1$), which gives, by division, 
$p_{i_0}\geq \lambda / d$. It implies $i_0=1$, since otherwise, we would have 
$L\geq y\geq p_{i_0 -1} \times p_{i_0} \geq p_{i_0}^2\geq \lambda^2 / d^2$, which contradicts the inequality $\lambda^2 / d^2> L$, for~$N$ large enough. 
Thus, we obtain the factorization $y=p_1\times y_1$, where $p_1 \geq \lambda$ and
$y_1 \coloneqq \prod_{i=2}^{k} p_i$. 
We get $y_1<\lambda$, since otherwise, we would have 
$y=p_1\times y_1\geq \lambda^2 > L$, a contradiction. Moreover, the inequality $y_1>d$ also holds, since otherwise, $y=p_1\times y_1$ would be of the ``forbidden form''. 
\end{enumerate}
Thus, in both cases, we can exhibit a divisor $y_1$ of $y$ such that $d<y_1<\lambda$.
\end{proof}

Notice that we can compute in time $O(L\times \lambda)= O(N)$ an array 
$\arr{factor}[\lambda..L]$ defined by $\arr{factor}[y]\coloneqq y_1$, where $y_1$ is the least divisor of $y$ in $[d+1,\lambda-1]$, if it exists (by Lemma~\ref{lemma:forbidden form}), and is 0 otherwise, i.e.\ when $y$ is of the forbidden form: it suffices to check, for all $y_1\in [d+1,\lambda-1]$, if 
$y \modop y_1 = 0$.

Now, using the $\arr{factor}$ array and Lemma~\ref{lemma:composeRoots}, it is easy to compute
$\lf x^{1/y} \rf$ in constant-time (with linear-time preprocessing), for every $x<N^d$ and every 
$y \in [\lambda, L]$ not in forbidden form: 
take $y_1=\arr{factor}[y]$ and $y_2=y / y_1$, which gives
$\lf x^{1/y} \rf = \lf \lf x^{1/y_1} \rf^{1/y_2} \rf$; 
the roots \linebreak
$z\coloneqq \lf x^{1/y_1} \rf$ and then 
$\lf z^{1/y_2} \rf = \lf \lf x^{1/y_1} \rf^{1/y_2} \rf$ 
can be computed respectively by the above algorithm, lines~12-24, and by reading the element 
$\arr{root}[z,y]$ of the array $\arr{root}[0..B-1][1..L]$, since we have respectively $y_1< \lambda$ and 
$y_1\geq d+1$, from which $z=\lf x^{1/y_1} \rf \leq (N^d)^{1/(d+1)} \leq \lc N^{1/(d+1)} \rc^d < B$ is deduced.

\begin{remark}
Using Lemmas~\ref{lemma:forbidden form} and~\ref{lemma:composeRoots}, it is easy to verify that Conjecture~\ref{conj: genRoot} is equivalent to its following restriction:
For any fixed integer $d\geq 1$, the function $(x,y)\mapsto \lf x^{1/p} \rf$, where $x<N^d$ 
and $p$ is a prime number such that $\log_2 N /(12 \log_2\log_2 N) < p < \log_2 N$, is computable in constant time after preprocessing in $O(N)$ time. Indeed, an integer in forbidden form being of the form $y=p \times q$, for a prime number $p$ and an integer $q\leq d$, the computation of 
$\lf x^{1/y} \rf =  \lf \lf x^{1/p} \rf^{1/q}\rf$ is reduced to the computation of $z\coloneqq \lf x^{1/p} \rf$ since, obviously, the final result $ \lf z^{1/q} \rf$ can be computed in constant time. 
\end{remark}


%% file: 8_Others.tex
\section{Computing many other operations in constant time}\label{sec:bitwiseOperConstTime}
Real computers treat the contents and addresses of registers as binary
words. In addition to arithmetic operations and for the sake of
efficiency, assembly languages as well as programming languages
(Python, etc.) use string operations and bitwise logical operations:
$\andop$, $\orop$, $\xorop$, etc.

\newcommand{\xorar}{xorAr}
\newcommand{\andar}{andAr}

The objects processed by the RAM model are integers but can be assumed
to be in binary notation, i.e.\ they can be seen as binary words.  In
this section, we show that all the usual string and logical operations
on binary words can be computed in constant time with linear-time
preprocessing, just like on a standard CPU of a real computer.

\subsection{String operations and bitwise logical operations computed in constant time}

We identify each integer $x\in\mathbb{N}$ with the string
$x_{\ell-1}\cdots x_1x_0\in\{0,1\}^{\ell}$, of its binary notation with
$x=\sum_{0\leq i < \ell} 2^i \times x_i$.  
The length $\ell$ of the representation of $x$ (without useless leading zeros: $x_{\ell-1}>0$) is given by the
following function:

\begin{equation*}
\Call{length}{x} =
\begin{cases}
1 & \mathtt{if}\; x=0\\
\lfloor\log_2 x\rfloor +1& \mathtt{otherwise}.
\end{cases}
\end{equation*}

\begin{remark}
We have $\Call{length}{x} \leq \log_2 N+ O(1)$, for each integer $x=O(N)$
contained in a RAM register, and $\Call{length}{x}=O(\log N)$ for each
``polynomial'' integer $x=O(N^d)$, for a fixed~$d$.
\end{remark}

\noindent
Three operations are essential for manipulating binary strings:
\begin{itemize}
\item the concatenation of two strings $X$ and $Y$ is defined as
  $\Call{Conc}{X,Y}\coloneqq X \times 2^{\Call{length}{Y}}+Y$;
\item the function that maps each pair $(X,i)$ of a string $X$ and an index $i$ to the $i$th bit $x_i$ of $X$ is defined as 
$\Call{Bit}{X,i}\coloneqq (X \divop 2^i) \modop 2$; 
\item the substring function which associates to each string
  $X=x_{\ell-1}\cdots x_1x_0$ and two indices $i,j$ with $\ell\ge  i> j \ge 0$ 
  the (nonempty) substring $x_{i-1} x_{i-2}\cdots x_{j}$, is defined as
  $\Call{Substring}{X,i,j}\coloneqq (X \modop 2^{i}) \divop 2^j$.
\end{itemize}
With $\fun{Substring}$, one can express that a string $X$ is a suffix
or a prefix of another string $Y$:
\begin{itemize}
\item $X$ is a suffix of $Y \iff X=\Call{Substring}{Y,\Call{length}{X},0}$;
\item $X$ is a prefix of $Y \iff
  \begin{cases}
    \Call{length}{X} \leq \Call{length}{Y} \text{ and }\\
    X=\Call{Substring}{Y,\Call{length}{Y}, \Call{length}{Y}-\Call{length}{X}}
  \end{cases}$.
\end{itemize}

\begin{remark}
Note that these procedures can be straightforwardly adapted if the
integers are represented in a fixed base $\gamma\geq 2$, i.e.\ are
identified with strings $x_{\ell-1}\ldots
x_1x_0\in\{0,1,\ldots,\gamma-1\}^{\ell}$.
\end{remark}


Processing is similar for each bitwise logical operation. As an
example, let us study the operation $\xorop$ applied to two integers
$X,Y$ which are less than $N^d$ (for a fixed $d$) and can therefore be identified with the
strings of their binary notations $x_{\ell-1}\cdots x_0$ and
$y_{\ell-1}\cdots y_0$ of length \linebreak
$\ell= \Call{length}{\mathtt{max}(X,Y)} = O(\log N)$.

The natural idea is to decompose each operand $X,Y<2^{\ell}$ into its
$\ell$ bits copied in $\ell$ registers, to compute the bits $z_i=x_i
\;\xorop \; y_i$ in $\ell$ registers too, and to get the result $Z=X
\;\xorop\; Y$ by concatenating the $\ell$ bits $z_i$ in one register:
$Z=z_{\ell-1}\cdots z_0$. Obviously, this process takes time
$O(\ell)$, which is, generally, not a constant time.

To compute the integer $X\;\xorop \;Y$ in constant time, for operands $X,Y<N^d$ where $d$ is a fixed integer, our trick is to
notice that if we set $X=X_1 \times p+ X_0$ and $Y= Y_1\times p + Y_0$, where
$p$ is a power of two such that $X_0,Y_0 < p$, then we get
\begin{eqnarray}\label{eqn:decomposebit}
X\;\xorop \;Y = (X_1\;\xorop\;Y_1)\times p + (X_0 \; \xorop \; Y_0)
\end{eqnarray}
Of course, the same ``decomposition'' property holds for the other bitwise operations $\andop$, 
$\orop$: $X\;\andop \;Y = (X_1\;\andop\;Y_1)\times p + (X_0 \; \andop \; Y_0)$
and $X\;\orop \;Y = (X_1\;\orop\;Y_1)\times p + (X_0 \; \orop \; Y_0)$.
Also, note that those equalities hold for all pairs of integers $X,Y>1$, i.e.\ represented with at least 2 bits. 

Now, by using equation~(\ref{eqn:decomposebit}) iteratively for $p=2$, we can pre-compute a 
2-dimensional table, called $\arr{\xorar{}}$, giving  $Z=X\;\xorop\; Y$ 
for all the pairs of small enough operands $X,Y$. 
For large operands, we use equation~(\ref{eqn:decomposebit})
for $p=2^{\ell}$ and $\ell=\lf \frac{1}{2} \times \Call{length}{\mathtt{max}(X,Y)}\rf$.
Thus, we obtain a recursive procedure called $\fun{xor}$: 
at each step the maximum length $L$ of the arguments of $\fun{xor}$, 
$L\coloneqq \Call{length}{\mathtt{max}(X,Y)}$, is divided by two, 
i.e.\ $L$ becomes $\lf L/2 \rf$ or $\lc L/2 \rc$.

The previous arguments justify that the following algorithm correctly computes the $\xorop$ operation.

\begin{minipage}{\almosttextwidth}
\begin{algorithm}[H]
  \caption{Computation of the $\xorop$ operation}
      \begin{minipage}{0.58\textwidth}
        \begin{algorithmic}[1]
\Procedure{LinPreprocXOR}{\mbox{}}
  \State $\var{K} \gets {\Call{sqrt}{N}}$ \Comment{$\var{K}=\lc \sqrt{N} \rc$}
  \State $\arr{\xorar{}} \gets $ array of dimension $\var{K}\times \var{K}$
   \For{$\var{i} \From 0 \To 1$}
      \For{$\var{j} \From 0 \To 1$}
         \State $\arr{\xorar{}}[\var{i}][\var{j}] \gets (\var{i}+\var{j})\modop 2$
      \EndFor
   \EndFor
  
  \For{$\var{i} \From 0 \To \var{K}-1$}
  \For{$\var{j} \From 0 \To \var{K}-1$}
  \If{$\var{i}> 1$ or $\var{j}> 1$}
      \State $\arr{\xorar{}}[\var{i}][\var{j}] \gets 2 \times \arr{\xorar{}}[\var{i}/2][\var{j}/2]$
       \\ \mbox{~}\hspace{4em}
      $+ \; \arr{\xorar{}}[\var{i} \modop 2][\var{j} \modop 2] $
  \EndIf
    \EndFor
  \EndFor
\EndProcedure
 \end{algorithmic}
 \end{minipage}  
 \begin{minipage}{0.41\textwidth}
      \begin{algorithmic}[1]
\Procedure{xor}{$\var{x},\var{y}$}
  \If{$\var{x} < \var{K}$ and $\var{y} < \var{K}$}
    \State \Return $\arr{\xorar{}}[\var{x}][\var{y}]$
  \Else
    \State $\ell \gets \Call{length}{\mathtt{max}(\var{x},\var{y})}/2$
    \State $\var{p} \gets 2^{\ell}$
    \State \Return $\var{p} \times \Call{xor}{\var{x}/\var{p},\var{y}/\var{p}}$
        \mbox{~}\hspace{6em} $+\; \Call{xor}{\var{x} \modop \var{p}, \var{y} \modop \var{p}}$
    \EndIf
    \EndProcedure
      \end{algorithmic}
      \end{minipage}
\end{algorithm}
\end{minipage}

\paragraph{Time complexity.}
We have the linearity of the preprocessing because $K=\lc \sqrt{N} \rc$ and
thus the two nested loops, lines 7-10, take $O(N)$ time. 
The result of the preprocessing is an array such that 
$\arr{\xorar{}}[\var{x}][\var{y}]=\var{x}\;\;\xorop \;\;\var{y}$ for all $\var{x},\var{y} < K$. 
When calling $\Call{xor}{\var{x},\var{y}}$,
either both arguments $\var{x},\var{y}$ are smaller than $K$, or we recurse with two calls, 
$\Call{xor}{\var{x}',\var{y}'}$ and $\Call{xor}{\var{x}'',\var{y}''}$, where the length of the maximum of
$\var{x}'$ and~$\var{y}'$ (resp.\ $\var{x}''$ and~$\var{y}''$) is, at most, half the length of the maximum of $\var{x}$ and $\var{y}$ plus one.
Starting with $\var{x},\var{y}<N^d$, which implies $\var{x},\var{y}<K^{2d}$, we obtain,  
for the number~$L_2$ defined as the smallest power of two greater than or equal to 
$\Call{length}{\mathtt{max}(\var{x},\var{y})}$,
the inequalities\footnote{We use the general inequality 
$\Call{length}{u \times v} \leq  \Call{length}{u} + \Call{length}{v}$, for all integers $u,v$.}
$L_2< 2 \times \Call{length}{\mathtt{max}(\var{x},\var{y})} \leq 2 \times \Call{length}{K^{2d}} 
\leq 4d \times \Call{length}{K}$, from which we deduce 
$L_2/(2^{\lc \log_2 d \rc +2})\leq L_2/(4d) < \Call{length}{K}$.

Therefore, the recursion depth is bounded by the number $\delta\coloneqq \lc \log_2 d \rc +2$ since, after repeatedly dividing $L_2$ by two, $\delta$ times, $L_2$ becomes smaller than $\Call{length}{K}$ so that the operands of the $\fun{xor}$ procedure themselves become less than $K$. 
This means that the $\fun{xor}$ procedure executes in constant time, for any fixed integer $d\geq 1$.

\paragraph{Getting other bitwise logical operators.}
All the bitwise operations can be obtained in the same way. 
To get the operation $\andop$ (resp.\ $\orop$), 
we only have to replace the $\arr{\xorar{}}$ array with an $\arr{\andar{}}$ (resp.\ $\arr{orAr}$) array,
and the assignment $\arr{\xorar{}}[\var{i}][\var{j}] \gets (\var{i}+\var{j})\modop 2$, which defines the Boolean $\xorop$ operation (line 6 of the preprocessing procedure), is replaced by the assignment \linebreak
$\arr{\andar{}}[\var{i}][\var{j}] \gets \var{i}\times\var{j}$ 
(resp.\ $\arr{orAr}[\var{i}][\var{j}] \gets \mathtt{max}(\var{i},\var{j})$), defining the Boolean $\andop$ 
(resp.\ $\orop$) operation. 


\subsection{Constant time and complexity of Turing machines and cellular automata}\label{subsec:TuringCA}

The relationships between the complexity of RAMs and that of other
major computation models, Turing machines and cellular automata, are
wide open issues.  It is easy to verify that addition, subtraction and
all the string or logical operations studied in the previous
subsections are computable in linear time on multi-tape Turing
machines.  On the other hand, it is known that multiplication is
computable in linear time on cellular
automata~\cite{Goyal76,Even90,MazoyerY12}.

A natural question arises: is it true that each operation computable
in linear time on a Turing machine or on a cellular automaton is
computable in constant time on a RAM with addition, with linear-time
preprocessing? In this subsection, we answer this question positively.

\medskip
Since it is known and easy to verify that any problem computable in linear time on a multi-tape Turing machine is computable in linear time on a cellular automaton~\cite{Smith71,Terrier12}, it suffices to prove the expected result for cellular automata.

First, we need some definitions and lemmas about one-dimensional cellular automata~\cite{Kari12, Terrier12}.

\begin{definition}[cellular automaton, configuration, computation] 
• A \emph{cellular automaton} (CA or \emph{automaton}) 
$\mathcal{A}\coloneqq (Q,\delta)$ consists of a finite set of \emph{states} $Q$ and a \emph{transition function} 
$\delta: Q^3\to Q$.

• A \emph{configuration} $C$ of the CA is a bi-infinite word $\cdots c_{-2}c_{-1}c_0 c_1 c_2 \cdots \in Q^{\mathbb{Z}}$.
For $i\in \mathbb{Z}$, we say that the cell $i$ of the configuration $C$ is in the state $c_i$.

• The \emph{successor configuration} of $C$ for the CA is 
$\delta(C)\coloneqq \cdots c'_{-2}c'_{-1}c'_0 c'_1 c'_2 \cdots$ where \linebreak
$c'_i\coloneqq\delta(c_{i-1},c_i,c_{i+1})$, for each $i\in\mathbb{Z}$.

• The \emph{computation} of the CA 
from any initial configuration $C$ is the sequence of configurations $(\delta^j(C))_{j\geq 0}$ defined by 
$\delta^0(C)\coloneqq C$ and $\delta^{j+1}(C)\coloneqq  \delta(\delta^j(C))$.
\end{definition}

We are interested in finite word problems. 
\begin{definition}[word configuration]
The configuration of a cellular automaton 
$\mathcal{A}=(Q,\delta)$ that represents a finite non-empty word $w\in\Sigma^+$ over a finite alphabet 
$\Sigma\subset Q$ is the bi-infinite word $\cdots\sharp\sharp w \sharp\sharp\cdots$, 
simply denoted $\sharp w \sharp$ and called a \emph{word configuration}, where $\sharp\in Q\setminus \Sigma$ is a special state (comparable to the ``blank'' symbol of a Turing machine).
\end{definition}  

Let us describe how ``special'' the state $\sharp$ can be:

\begin{definition}[
permanent state on the right]\label{def:sharp}


A state $\sharp$ is \emph{permanent on the right} for
a cellular automaton $\mathcal{A}=(Q,\delta)$ if we have:
(1) $\delta(q,\sharp,\sharp)=\sharp$, for each $q\in Q$;
\;(2) $\delta(q_{-1},q,q_1)\neq \sharp$, for each $q\in Q\setminus\{\sharp\}$ and all $q_{-1},q_1\in Q$.

\end{definition} 

\begin{remark}
If $\sharp$ is a permanent state on the right and 
$C$ is a configuration having state~$\sharp$ on all cells of position $j\geq i$
(i.e.\ $\forall j\geq i, c_j = \sharp$), then these cells remain in state $\sharp$ in $\delta(C)$. 
Also, note that the condition (2) of Definition~\ref{def:sharp} prevents the computation from creating new $\sharp$ symbols.
\end{remark}

\begin{definition}[active cell, active interval]
A cell $i$ is said to be \emph{active} in a computation of a CA if its state $c_i$ is not $\sharp$ in at least one configuration $C$ of the computation.
By Definition~\ref{def:sharp}, the set of active cells of the computation from a word configuration is an interval, called the \emph{active interval} of the computation.
\end{definition}

We need to define a function computable by a CA in linear time, a notion defined and studied in~\cite{GrandjeanRT12} in the case where the output length is equal to the input length.
We need to adapt this definition to allow the outputs to have any length. 
Our definition below is a bit technical: 
namely, 
the CA must \emph{synchronize} its output production but the synchronization is a constraint a priori difficult to satisfy by a parallel computation.

\begin{definition}[function computed by a CA in linear time]\label{def:lintimeCA}
 Let $\Sigma$ and $\Sigma'$ be two finite alphabets and let $\mathcal{A}=(Q,\delta)$ be a CA with 
 $\Sigma\subset Q$, a non-empty set of \emph{output states} $Q_{out} \subset Q\setminus \Sigma$ and two special states: 
 the state $q_{\sharp}\in Q_{out}$ 
 and the state $\sharp \in Q \setminus (Q_{out} \cup \Sigma \cup \Sigma')$ permanent on the right.
 
$\mathcal{A}$ \emph{computes} a function  $f:\Sigma^+ \to \Sigma'^+$ \emph{in linear time} if there are an integer $c\geq1$ and a projection $\pi: Q_{out}\cup \{\sharp\} \to \Sigma'\cup \{\sharp\}$
such that, for each integer $n>0$ and each word $w\in \Sigma^n$, called the \emph{input} of $\mathcal{A}$, 
there is an integer $T\leq cn$ so that conditions 1-4 are satisfied:

\begin{enumerate}
\item 
$\delta^{T}(\sharp w \sharp)=\sharp (q_{\sharp})^k v \sharp$ for an integer $k\geq 0$ and a word 
$v=v_m\cdots v_1\in (Q_{out}\setminus \{q_{\sharp}\})^+$;
 
\item for each $t<T$, no cell of the configuration $\delta^{t}(\sharp w \sharp)$
is in a state of $Q_{out}$;

\item $\pi(q_{\sharp})=\pi(\sharp)= \sharp$ and $\pi(q)\in\Sigma'$ for each 
$q \in Q_{out}\setminus \{q_{\sharp}\}$;

\item The word $f(w)=\pi(v)\coloneqq \pi(v_m)\cdots \pi(v_1)$ of $\Sigma'^+$ is called the 
\emph{output} of $\mathcal{A}$.

\end{enumerate}

\end{definition}

\noindent
Note that if we set 
$\pi(\sharp (q_{\sharp})^k v \sharp)\coloneqq \sharp(\pi(q_{\sharp}))^k \pi(v)\sharp$, then we get
$\pi(\sharp (q_{\sharp})^k v \sharp)=\sharp f(w) \sharp$ and therefore 
$\pi(\delta^T(\sharp w \sharp))=\sharp f(w) \sharp$.\\
Note also that the length of $f(w)$, which is that of $v$, is bounded by $n+T\leq (c+1)n$.

\begin{remark}
The conjunction of conditions 1 and 2 expresses the synchronization of the output configuration
$\sharp (q_{\sharp})^k v \sharp$ : 
at the exact instant $T$, not before, each of its $k+m$ active cells ``knows" that it is in its final state.
\end{remark}

\begin{remark}\label{lemma:binaryOutput}
When the output is an integer written in binary with $\Sigma' \coloneqq \{0,1\}$, which will be the case in our application (Theorem~\ref{th:linTimeCA}), it is natural and more convenient to take \linebreak
$\pi: Q_{out}\cup\{\sharp\}\to \{0,1,\sharp\}$ with 
$\pi(q_{\sharp})\coloneqq 0$ (alternatively, the $q_{\sharp}$ state can be deleted). 
This will give $\pi(\delta^T(\sharp w \sharp))= \sharp f(w) \sharp$
with the output integer $f(w)\in \{0,1\}^T$ written to $T$ bits, possibly including leading zeros.
\end{remark}

Our proof of Theorem~\ref{th:linTimeCA} will use properties of functions computable in linear time by~CA. 

\begin{lemma}\label{lemma:compose}
The class of functions computable in linear time by cellular automata is closed under composition.
\end{lemma}

\begin{proof}

Let $f:\Sigma\to \Sigma'$ and $f':\Sigma'\to \Sigma''$ be two functions computed in linear time by the cellular automata $\mathcal{A}=(Q,\delta)$ and $\mathcal{A'}=(Q',\delta')$, respectively, with the projections 
\linebreak
$\pi: Q_{out}\cup \{\sharp\} \to \Sigma'\cup \{\sharp\}$ and 
$\pi': Q'_{out}\cup \{\sharp\} \to \Sigma''\cup \{\sharp\}$, respectively.
Without loss of generality, assume $Q\cap Q'=\{\sharp\}$.
To compute the composition function $f'\circ f$, we must build a cellular automaton 
$\mathcal{A''}=(Q'',\delta'')$ whose program is the composition of the programs of 
$\mathcal{A}$, $\pi$ and $\mathcal{A}'$.
Using the fact that the sets of states $Q \setminus Q_{out}$, $Q_{out}$ and $Q' \setminus \{\sharp\}$ are mutually disjoint,
 we define the automaton $\mathcal{A''}=(Q'',\delta'')$ by $Q'' \coloneqq Q \cup Q'$ and,
 for all  $q_{-1},q,q_1\in Q''$,

\begin{equation}\label{eqn:casesCA}
\delta''(q_{-1},q,q_1) \coloneqq
\begin{cases}
\delta(q_{-1},q,q_1)& \mathtt{if}\; q_{-1},q,q_1\in Q\setminus Q_{out}\\
\pi(q) & \mathtt{if}\; q \in Q_{out}\\
\delta'(q_{-1},q,q_1)& \mathtt{if}\; q_{-1},q,q_1\in Q'\\
\mathtt{any}\; \mathtt{value}& \mathtt{otherwise}.
\end{cases}
\end{equation}
The reader can easily convince herself that starting from any word configuration 
$\sharp w \sharp$, $w\in \Sigma^+$, \linebreak
the computation of $\mathcal{A''}$ only meets cases 1, 2 and 3 of equation~(\ref{eqn:casesCA}), in this order\footnote{The only ``exception'' is the case 
$\delta''(\sharp,\sharp,\sharp)=\delta(\sharp,\sharp,\sharp)=\delta'(\sharp,\sharp,\sharp)=\sharp$ common to cases 1 and 3. The other cases are mutually disjoint}.
By hypothesis, there are constants $c,c'$ such that, for each $n>0$ and each word $w\in\Sigma^n$, 
we have $\pi(\delta^T(\sharp w \sharp))=\sharp f(w) \sharp$, for an integer $T\leq cn$, and 
$\pi'(\delta'^{T'}(\sharp f(w) \sharp))=\sharp f'(f(w)) \sharp$, for an integer $T'\leq c'n'$ where $n'$ is the length of $f(w)$.
This implies $T'\leq c'(c+1)n$. 
By construction, the automaton $\mathcal{A''}$ transforms the configuration $\sharp w \sharp$
into the configuration $\delta^T(\sharp w \sharp)$ (in $T$ time), and then into 
$\pi(\delta^T(\sharp w \sharp))=\sharp f(w) \sharp$ (in unit time), and finally into 
$\delta'^{T'}(\sharp f(w) \sharp))$ (in $T'$ time) from which $f'(f(w))$ is obtained by the $\pi'$ projection.
Thus, $\mathcal{A''}$ computes $f' \circ f (w)$, with $\pi'$, in linear time, 
more precisely, time $T+1+T' \leq cn +1 + c'(c+1)n$.
\end{proof}

We are going to show that we can impose that the computation time of a function computable in linear time on a CA be an \emph{exact multiple} of the length of the input. For this we need the following result:

\begin{lemma}\label{lemma:fire}
Let $\Sigma$ be a finite alphabet and $c$ be a positive integer. 
There is a cellular automaton $\mathcal{A}_{\varphi}=(\delta_{\varphi},Q_{\varphi})$, called \emph{firing squad automaton}, with 
$\Sigma\subset Q_{\varphi}$, two special states, $\varphi$, called the ``fire'' state, and $\sharp$, the permanent state on the right, such that, for each $n\geq 1$ and each word $w\in\Sigma^n$, we have
$\delta_{\varphi}^{cn}(\sharp w \sharp) = \sharp\varphi^{cn} \sharp$ whereas, for each $t<cn$, the ``fire'' state $\varphi$ does not appear in the configuration $\delta_{\varphi}^{t}(\sharp w \sharp)$.
\end{lemma}

\begin{proof}[Proof sketch]
It is a variant (slight extension) of the classical ``firing squad'' automaton~\cite{Mazoyer86, Mazoyer87,Mazoyer96} which corresponds to the case $c=1$ of the lemma. 
This means $\delta_{\varphi}^{n}(\sharp w \sharp) = \sharp\varphi^{n} \sharp$ whereas, for each $t<n$, 
the $\varphi$ state does not appear in the configuration $\delta_{\varphi}^{t}(\sharp w \sharp)$. 
Now, let us allow ``extended'' neighborhoods of $c+2$ cells. 
Specifically, allow ``extended'' transition functions $\delta: Q^{c+2}\to Q$ given by
$q'=\delta(q_{-c},q_{-c+1},\dots,q_{-1},q_0,q_1)$ 
where $q'$ denotes the state of a cell numbered $\gamma$ at time~$t$, 
and $q_j$, $j\in[-c,1]$, is the state of the cell numbered $\gamma + j$ at time $t-1$.
By launching in parallel (from an input $w\in \Sigma^n$) $c$ ``firing squad'' automata, the original one and $c-1$ ``shifted copies'', we obtain
$\delta^{n}(\sharp w \sharp) = \sharp\varphi^{cn} \sharp$.
By slowing down time, precisely by simulating one computation step of the ``extended'' transition function thus constructed by $c$ steps of a ``standard'' transition function denoted 
$\delta_{\varphi}:(Q_{\varphi})^3 \to Q_{\varphi}$, we obtain the desired result: 
$\delta_{\varphi}^{cn}(\sharp w \sharp) = \sharp\varphi^{cn} \sharp$.
\end{proof}

From Lemma~\ref{lemma:fire} we deduce the following lemma:
\begin{lemma}\label{lemma:exactlintime}
Let $f: \Sigma^+ \to  \Sigma'^+$ be a function computable by a CA in linear time. 
Then there are a constant integer $c$ and a cellular automaton $\mathcal{A}=(Q,\delta)$ 
that computes $f(w)$ 
in space (width of the active interval) and in time both equal to $cn$ for an input $w\in\Sigma^n$. 
\end{lemma}

\begin{proof}
Essentially, the idea is to launch in parallel on the same input $w\in\Sigma^n$ the original automaton, denoted $\mathcal{A}_0=(Q_0,\delta_0)$, of time complexity $T<(c-1)n$ and therefore of space complexity  bounded by $cn$, 
and the firing squad automaton 
$\mathcal{A}_{\varphi}=(\delta_{\varphi},Q_{\varphi})$ of Lemma~\ref{lemma:fire} which acts as a $cn$-clock.

Define the automaton $\mathcal{A}=(Q,\delta)$ as follows:
$Q \coloneqq (Q_0 \times Q_{\varphi}) \cup \Sigma \cup Q_{out} \cup \{\sharp\}$, 
with the same set $Q_{out}$ of output states (and the same projection $\pi$) for $\mathcal{A}$ as for 
$\mathcal{A}_0$;
the transition function $\delta$ is given by: 
for the first step of the computation of $\mathcal{A}$, set, for all $q_{-1},q,q_1\in \Sigma\cup\{\sharp\}$,
$$\delta(q_{-1},q,q_1) \coloneqq (\delta_0(q_{-1},q,q_1), \delta_{\varphi}(q_{-1},q,q_1));$$
for the other steps, set for all $q_{-1},q,q_1\in Q_0$ and $q'_{-1},q',q'_1 \in Q_{\varphi}$,
\begin{equation*}
\delta((q_{-1},q'_{-1}),(q,q'),(q_1,q'_1)) \coloneqq \!
\begin{cases}
(\delta_0(q_{-1},q,q_1), \delta_{\varphi}(q'_{-1},q',q'_1))& \mathtt{if}\; 
q_{-1},q,q_1 \in Q_0 \setminus Q_{out} \\
(q, \delta_{\varphi}(q'_{-1},q',q'_1))& \mathtt{if}
\; q \!\in \! Q_{out}\!\cup\!\{\sharp\} \;\mathtt{and}\; \delta_{\varphi}(q'_{-1},q',q'_1) \!\neq\! \varphi \\
q & \mathtt{if}\; \delta_{\varphi}(q'_{-1},q',q'_1) = \varphi. 
\end{cases}
\end{equation*}
The first case of this equation applies from the 2nd step of the computation of $\mathcal{A}$ to its $T$-th step, corresponding to the last step of $\mathcal{A}_0$ ($T\leq (c-1)n$); 
the second case applies from the $(T+1)$-th step to the $(cn-1)$-th step of $\mathcal{A}$ ($\mathcal{A}_0$ has finished its computation);
finally, the third case of the equation applies to the $cn$-th step (the last step of $\mathcal{A}_{\varphi}$).

Note that since the second and third cases of the above equation retain the output states \linebreak
$q \!\in \! Q_{out} \cup \{\sharp\}$ the final configuration 
$\delta^{cn}(\sharp w \sharp)$ of $\mathcal{A}$ 
is the same as that of the original \linebreak
automaton $\mathcal{A}_0$.
Finally, observe that the computation time and space of $\mathcal{A}$ are those of $\mathcal{A}_{\varphi}$ which are exactly $cn$.
\end{proof}

All the preceding definitions and lemmas aim to study the operations computed in linear time by cellular automata as defined below.

\begin{definition}[operation computed by a CA in linear time]
A cellular automaton $\mathcal{A}$ \emph{computes} an operation $\op: \mathbb{N}^r \to  \mathbb{N}$ of arity $r\geq 1$ \emph{in linear time} if $\mathcal{A}$ computes in linear time a function 
$X_1;\cdots;X_r \mapsto Y$ from $\{0,1,;\}^+$ to $\{0,1\}^+$,
with integers $X_1,\ldots,X_r$ represented in binary and a binary representation $Y$ of 
$\op(X_1,\ldots,X_r)$, possibly including leading zeros. 
\end{definition}





Here is our general result relating the complexity of cellular automata and RAMs:


\begin{theorem}\label{th:linTimeCA}
Let $d\geq 1$ be a constant integer and let $\mathtt{op}: \mathbb{N}^r\to \mathbb{N}$  be an operation of arity $r\geq 1$ which is computable in linear time on a cellular automaton. 
There is a RAM with addition, such that, for any input integer $N>0$:
\begin{enumerate}
\item \emph{Pre-computation:} the RAM computes some tables in time $O(N)$;
\item \emph{Operation:} from these pre-computed tables and after reading $r$ integers $X_1,\ldots,X_r<N^d$, the RAM computes in constant time the integer $\mathtt{op}(X_1,\ldots,X_r)$. 
\end{enumerate}
\end{theorem}

\noindent 
The formal proof of this theorem is very technical and is given in the appendix, but the basic idea is quite simple.
\begin{proof}[Proof sketch]
Let $\mathcal{A}=(Q,\delta)$ be a cellular automaton that computes the operation 
$\op: \mathbb{N}^r \to  \mathbb{N}$ in linear time according to Lemma~\ref{lemma:exactlintime} and let 
$X\coloneqq X_1;\cdots ; X_r$ be an input word for $\mathcal{A}$.
Let $L$ be the common value of the computation time and space of $\mathcal{A}$ on the input $X$. 
By hypothesis, we have $L=\Theta(\mathtt{length}(X))=O(\log N)$.

The idea is to divide the time-interval $[0,L]$ and the active interval of $L$ cells
into $k$ time-periods of duration $\ell\coloneqq L/k$, called $\ell$-\emph{periods}, and $k$ segments of $\ell$ cells, respectively, for a large enough fixed integer~$k$. 
For $1\leq i,t \leq k$, let $w_{i,t}\in Q^{\ell}$ be the word, called $\ell$-\emph{block}, contained in the $i$-th segment 
of $\ell$ cells at the end of the $t$-th $\ell$-period of the computation of $\mathcal{A}$ on $X$.
Similarly, let $w_{i,0}$ be the $i$-th $\ell$-block of the initial configuration $C_0$.
Thus, we consider the sequence of configurations of $\mathcal{A}$ on~$X$, every $\ell$ computation steps 
(= every $\ell$-period): 
\begin{center}
$C_0=\sharp w_{1,0}\cdots w_{k,0}\sharp, \;C_1=\delta^{\ell}(C_0),\ldots,\;C_k=\delta^{\ell}(C_{k-1})=
\sharp w_{1,k}\cdots w_{k,k}\sharp$. 
\end{center}
Note that the total number of $\ell$-blocks $w_{i,t}$ in this sequence is $k^2+k=O(1)$.
Here are the important points:
\begin{enumerate}

\item for all $i,t$ with $1\leq i \leq k$ and $0\leq t<k$, the value of $w_{i,t+1}$ only depends on the three $\ell$-blocks $w_{i-1,t}$, $w_{i,t}$, $w_{i+1,t}$, which can be written $w_{i,t+1}=\delta^{\ell}(w_{i-1,t},w_{i,t},w_{i+1,t})$ 
with the convention $w_{0,t}=w_{k+1,t}=\sharp^{\ell}$;

\item it suffices to perform $k^2$ assignments of the form $w_{i,t+1}\gets\delta^{\ell}(w_{i-1,t},w_{i,t},w_{i+1,t})$ 
to simulate the whole computation of $\mathcal{A}$ on $X$, since there are exactly $k^2$ words $w_{i,t}$, $1\leq i,t \leq k$;

\item if $k$ is large enough, then the number of \emph{distinct} triplets $(w_{-1},w_0,w_1)\in (Q^{\ell})^3$ is \linebreak
$\lvert Q \rvert^{3\ell} = (\lvert Q \rvert^{3 L})^{1/k}= O(N^{1/2})$ and the array 
$T[w_{-1}][w_0][w_1]\coloneqq \delta^{\ell}(w_{-1},w_0,w_1)$ for \linebreak
$(w_{-1},w_0,w_1)\in (Q^{\ell})^3$,
can be computed in time $O(N^{1/2}\times \ell^{O(1)})=O(N)$.

\end{enumerate}

\noindent
Thus, the RAM desired first pre-computes the table $T$ in time O(N) (by point 3).
Secondly, after reading the operands $X_1,\ldots,X_r < N^d$, the RAM computes the initial configuration 
$C_0 \coloneqq \sharp X \sharp$ of $\mathcal{A}$, 
for $X\coloneqq X_1;\cdots ; X_r$, and the successive configurations 
\begin{center}
$\delta^{\ell}(C_0)=C_1,\ldots, \delta^{\ell}(C_{k-1})=C_{k}=w_{1,k} \cdots w_{k,k}$
\end{center}
by $k^2=O(1)$ consultations of the $T$ table (see points 1 and 2), therefore in constant time, and finally computes, still in constant time, the projection $\pi(w_{1,k})\cdots \pi(w_{k,k})$
which is by definition a binary representation of the integer $\mathtt{op}(X_1,\ldots,X_r)$.
\end{proof}

%% file: 9_SuccRAMweak.tex
\section{Minimality of the RAM with addition}\label{sec: MinAdd}

We have proved that the RAM model using addition as the only operation can perform in constant time -- with linear preprocessing -- any other usual arithmetic operation and, above all, Euclidean division.
The reader may ask the following natural questions which we explore in this section:

\begin{enumerate}
\item Is the addition a ``minimal'' operation for our complexity results? 
Can it be replaced by a finite set of unary operations (such as the successor and predecessor functions)?
\item Can addition be replaced, as a primitive operation, by any other
  usual arithmetic operation: subtraction, multiplication, division,
  or a combination of two of them?
\end{enumerate}

\subsection{Addition is not computable in constant time on a RAM with only unary operations} 

In this subsection, we will establish the negative result given by its
title and formulated precisely and in a stronger form by
Proposition~\ref{prop: AddNotCTonSuccRAM}.

\begin{proposition}\label{prop: AddNotCTonSuccRAM} 
Let $\mathtt{UnaryOp}$ be a finite set of \emph{unary} operations.
There is \emph{no} RAM in $\RAM[\mathtt{UnaryOp}]$ such that, for any input integer $N$:
\begin{enumerate}
\item \emph{Pre-computation:} the RAM computes some tables in \emph{space} $O(N)$;
\item \emph{Addition:} by using these pre-computed tables and by reading any two integers $x,y<N$, the RAM computes in \emph{constant time} their sum $x+y$. 
\end{enumerate}
Further, the result is still valid if the RAM is allowed to use any
(in)equality test, with $=,\neq,<,\leq$, in its branching
instructions.
\end{proposition}

\begin{remark}
By its item 1, Proposition~\ref{prop: AddNotCTonSuccRAM} states a strictly stronger result than what is required. Indeed, the space used by a RAM operating in linear time is by definition linear. 
\end{remark}

\begin{proof}[Proof of Proposition~\ref{prop: AddNotCTonSuccRAM} by contradiction] Suppose that there is a RAM $M$ in $\RAM[\mathtt{UnaryOp}]$
which fulfills the conditions 1,2 of Proposition~\ref{prop: AddNotCTonSuccRAM}. Without loss of generality, the following conditions can be assumed.
\begin{itemize}
\item Step 1 of $M$ (the \emph{pre-computation}) ends with the internal memory registers $R(0),R(1),\ldots$ containing the successive integers 
$p(0),p(1),\ldots,p(cN),0,0,\ldots$, for a constant $c\ge 1$, with $p(i)\leq cN$ in the register $R(i)$, for each $i\leq cN$, and 0 in $R(i)$, for each $i> cN$;
\item Step 2 of $M$ (the \emph{addition}) does not modify the internal memory registers $R(0),R(1),\ldots$ (which are read-only throughout Step 2) but uses a fixed number of read/write registers, denoted~$A$ (the accumulator) and $B_1,\ldots,B_k$ (the buffers), all initialized to 0, and two read-only registers $X$ and~$Y$, which contain the operands $x$ and $y$, respectively. 
The program of Step 2 is a sequence $I_0,\ldots, I_{r-1}$ of labeled instructions of the following forms:
\begin{itemize}
\item $A=0$ ; $A=X$ ; $A=Y$;
\item $A=\mathtt{op}(A)$, for some $\mathtt{op}\in \mathtt{UnaryOp}$;
\item $A=R(A)$ (meaning $A=R(a)$, where $a$ is the current content of $A$);
\item $B_i=A$, for some $i\in\{1,\ldots,k\}$ ; $A=B_i$, for some $i\in\{1,\ldots,k\}$;
\item $\mathtt{if}$ $A\prec B_1$ $\mathtt{then}$ $\mathtt{goto}$ $\ell_0$ 
$\mathtt{else}$ $\mathtt{goto}$ $\ell_1$, for $\prec$  $\in \{=,\neq,<,\leq,>,\geq\}$ and some \linebreak
$\ell_0,\ell_1\in \{0,\ldots,r-1\}$ (branching instruction);
\item $\mathtt{return}$ $A$ (this is the last instruction $I_{r-1}$; recall that the instruction executed after any other instruction $I_{j}$, $j<r-1$, is $I_{j+1}$, except after the branching instruction).
\end{itemize}
\end{itemize}

\noindent
\emph{Justification:} The constraint that the internal memory be read-only throughout Step 2 is made possible by the hypothesis that the time of Step 2 is bounded by some constant. Thus, Step 2 can consult/modify only a constant number of registers $R(i)$. Instead of modifying the contents of a register $R(i)$, we copy it -- at the first instant when it needs to be modified -- in one of the buffer registers $B_1,\ldots,B_k$ so that its contents can be modified as that of the original register $R(i)$ that it simulates. 

More precisely, if $m$ is an upper bound of the constant time bound of
the simulated (original) constant-time program so that this program
modifies (at most) $m$ distinct registers $R(i)$, then $k\coloneqq
2m+2$ buffer registers $B_1,\ldots,B_k$ are sufficient for the
simulation: $B_1$ is the original buffer; the contents of $B_2$ is the
number of original registers $R(i)$ which have been currently modified
(in the simulated/original constant-time program); the sequence of $m$
buffers $B_3,\ldots,B_{m+2}$ (resp.\ $B_{m+3},\ldots,B_{2m+2}$) contain
the current sequence of addresses $i$ (resp.\ sequence of contents) of
these registers $R(i)$~\footnote{When the original RAM tries to write
a value $v$ in a register $R(i)$, then the simulating RAM $M$ consults
the list of contents of the $s$ ``address" buffers
$B_3,\ldots,B_{s+2}$, where $s$ is the contents of the counter $B_2$.
If the address~$i$ is found, that means, for some $j$ ($3 \leq j \leq
s+2$), $B_j$ contains $i$, then the value $v$ is written in $B_{j+m}$
(the buffer that simulates $R(i)$).  Otherwise, the (not found)
address $i$ is copied in the buffer $B_{s+3}$ (the first buffer not
yet used after $B_{s+2}$), the value $v$ is written in the buffer
$B_{s+3+m}$ -- this buffer now simulates $R(i)$ -- and the counter
buffer $B_2$ is incremented by 1: it now contains $s+1$.  }.  The
reader can check that the time of Step 2 of $M$ is still bounded by a
constant $\mu=O(m^2)$.

\bigskip
For $0\leq i \leq \mu$, let $(A(i),B_1(i),\ldots,B_k(i), L(i))$ respectively be the sequence of contents of the read/write registers $A,B_1,\ldots,B_k$ and the value $L(i)$ of the ordinal instruction counter at instant~$i$ of Step~2. At instant 0, we have
$(A(0),B_1(0),\ldots,B_k(0), L(0))=(0,0,\ldots,0,0)$. For convenience, we adopt the following convention
when $L(i)=r-1$, for some $i<\mu$ (recall that the instruction $I_{r-1}$ is $\mathtt{return}$ $A$):
\begin{eqnarray}\label{eqn:return}
(A(i+1),B_1(i+1),\ldots,B_k(i+1),L(i+1))\coloneqq (A(i),B_1(i),\ldots,B_k(i), r-1)
\end{eqnarray}

\begin{claim}\label{claim:unary}
At each instant $i$ ($0\leq i \leq \mu$) of Step 2, an equation of the following form holds:

\begin{eqnarray}\label{eq: unarySucc}
(A(i),B_1(i),\ldots,B_k(i), L(i)) =
\left\{
\begin{array}{l}
\mathtt{if}\; t_1\;\mathtt{then}\;(a_1,b_{1,1},\ldots,b_{k,1},\lambda_1)\\
\ldots\ldots\ldots\ldots\ldots\ldots\ldots\ldots\ldots\ldots \\
\mathtt{if}\; t_q\;\mathtt{then}\;(a_q,b_{1,q},\ldots,b_{k,q},\lambda_q)\\
\end{array}
\right .
\end{eqnarray}

\noindent
Here, each test $t_j$ is a conjunction of (in)equalities 
$\bigwedge_{h} c_h \prec_h d_h$ where $\prec_h$  $\in \{=,\neq,<,\leq,>,\geq\}$
and each term $a_h,b_{h,j},c_h,d_h$ is of the (unary) form $(f_1\circ f_2 \circ \cdots \circ f_s)(u)$ where $u$ is either $X$, or $Y$, or the constant 0, and $s\geq 0$, and each $f_j$ is either 
a unary operation $\op\in\mathtt{UnaryOp}$ or the ``internal memory function" $p:[0,cN]\to [0,cN]$ which sends each address $\alpha$ to the contents $p(\alpha)$ of the corresponding read-only register $R(\alpha)$ (note that the function $p$ depends on $N$).

Moreover, the family of tests $(t_j)_{1\leq j \leq q}$ satisfies the \emph{partition property}: 
for each pair of integers $(x,y)$, \emph{exactly one} test $t_j$ is true for $X\coloneqq x$ and $Y\coloneqq y$.
\end{claim}

\begin{proof}[Proof of the claim] It is a proof by induction on the instant $i$.  Clearly, Claim~\ref{claim:unary} is true for $i=0$.
Assume that the equality~\ref{eq: unarySucc} is true for some $i<\mu$. 
There is nothing to prove if $I_{L(i)}$ is the instruction $\mathtt{return}$ $A$ because of equality~\ref{eqn:return}.
All other instructions are handled by a fairly similar process, with specific processing for the branching instruction.
First, suppose that $I_{L(i)}$ is an assignment of the form $A=v$ where $v\in\{0,X,Y\}$. Then we have

\begin{eqnarray*}
(A(i+1),B_1(i+1),\ldots,B_k(i+1), L(i+1)) =
\left\{
\begin{array}{l}
\mathtt{if}\; t_1\;\mathtt{then}\;(v,b_{1,1},\ldots,b_{k,1},\lambda_1+1)\\
\ldots\ldots\ldots\ldots\ldots\ldots\ldots\ldots\ldots\ldots \\
\mathtt{if}\; t_q\;\mathtt{then}\;(v,b_{1,q},\ldots,b_{k,q},\lambda_q+1)\\
\end{array}
\right .
\end{eqnarray*}

\noindent
Now, suppose that $I_{L(i)}$ is an assignment 
$A=\mathtt{op}(A)$ where $\mathtt{op}\in \mathtt{UnaryOp}$.
Then we have

\begin{eqnarray}\label{eq: Succ}
(A(i+1),B_1(i+1),\ldots,B_k(i+1), L(i+1)) =
\left\{
\begin{array}{l}
\mathtt{if}\; t_1\;\mathtt{then}\;(\mathtt{op}(a_1),b_{1,1},\ldots,b_{k,1},\lambda_1+1)\\
\ldots\ldots\ldots\ldots\ldots\ldots\ldots\ldots\ldots\ldots \\
\mathtt{if}\; t_q\;\mathtt{then}\;(\mathtt{op}(a_q),b_{1,q},\ldots,b_{k,q},\lambda_q+1)\\
\end{array}
\right .
\end{eqnarray}

\noindent
In the case where $I_{L(i)}$ is the assignment $A=R(A)$, replace the terms 
$\mathtt{op}(a_1),\ldots,\mathtt{op}(a_q)$
by the respective terms $p(a_1),\ldots, p(a_q)$ (involving the ``internal memory function" $p$) in equality~\ref{eq: Succ}.

\medskip
\noindent
The cases where $I_{L(i)}$ is $A=B_i$ or $B_i=A$ are similar. For example, for $B_1=A$, we have 

\begin{eqnarray*}
(A(i+1),B_1(i+1),\ldots,B_k(i+1), L(i+1)) =
\left\{
\begin{array}{l}
\mathtt{if}\; t_1\;\mathtt{then}\;(a_1,a_1,b_{2,1},\ldots,b_{k,1},\lambda_1+1)\\
\ldots\ldots\ldots\ldots\ldots\ldots\ldots\ldots\ldots\ldots \\
\mathtt{if}\; t_q\;\mathtt{then}\;(a_q,a_q,b_{2,q},\ldots,b_{k,q},\lambda_q+1)\\
\end{array}
\right .
\end{eqnarray*}

\noindent 
Finally, suppose that $I_{L(i)}$ is 
$\mathtt{if}$ $A\prec B_1$ $\mathtt{then}$ $\mathtt{goto}$ $\ell_0$ 
$\mathtt{else}$ $\mathtt{goto}$ $\ell_1$. 
Clearly, we have 

\begin{eqnarray*}
(A(i+1),B_1(i+1),\ldots,B_k(i+1), L(i+1)) =
\left\{
\begin{array}{l}
\mathtt{if}\; t_1\land a_1\prec b_{1,1}\;\mathtt{then}\;(a_1,b_{1,1},\ldots,b_{k,1},\ell_0)\\
\mathtt{if}\; t_1\land a_1\not\prec b_{1,1}\;\mathtt{then}\;(a_1,b_{1,1},\ldots,b_{k,1},\ell_1)\\
\ldots\ldots\ldots\ldots\ldots\ldots\ldots\ldots\ldots\ldots \\
\mathtt{if}\; t_q \land a_q \prec b_{1,q} \;\mathtt{then}\;(a_q,b_{1,q},\ldots,b_{k,q},\ell_0)\\
\mathtt{if}\; t_q \land a_q \not\prec b_{1,q} \;\mathtt{then}\;(a_q,b_{1,q},\ldots,b_{k,q},\ell_1)\\
\end{array}
\right .
\end{eqnarray*}
Note that each test has the required form since the negation $\not\prec$ of a relation $\prec$ in the set \linebreak
$\{=,\neq,<,\leq,>,\geq\}$ belongs also to this set; 
moreover, by construction, the $2q$ tests still satisfy the ``partition property".

\medskip
In conclusion, the reader can observe that for every possible instruction $I_{L(i)}$ of $M$, an equality of the required form~\ref{eq: unarySucc} with $i$ replaced by $i+1$ is true. This completes the proof by induction of Claim~\ref{claim:unary}.
\end{proof}

\noindent
\emph{End of the proof of Proposition}~\ref{prop: AddNotCTonSuccRAM}. The contradiction will be proven as an application of the ``pigeonhole" principle.
By hypothesis, the Step 2 of $M$ (addition step) computes/returns the sum $x+y$ at time $\mu$, that means $A(\mu)=x+y$. Therefore, by Claim~\ref{claim:unary}, we obtain, for all $(x,y)\in[0,N[^2$,

\begin{eqnarray}\label{eq: PlusSucc}
x+y=A(\mu) =
\left\{
\begin{array}{l}
\mathtt{if}\; t_1(x,y)\;\mathtt{then}\;a_1\\
\ldots\ldots\ldots\ldots \\
\mathtt{if}\; t_q(x,y)\;\mathtt{then}\;a_q\\
\end{array}
\right .
\end{eqnarray}
where each test $t_j$ is a conjunction of (in)equalities 
$\bigwedge_{h} c_h \prec_h d_h$ where $\prec_h$  $\in \{=,\neq,<,\leq,>,\geq\}$
and each term $a_j,c_h,d_h$ is of the (unary) form $f(0)$, $f(x)$ or $f(y)$, in which $f$ is a composition 
$f_1\circ f_2 \circ \cdots \circ f_s$ of functions $f_i\in \mathtt{UnaryOp} \cup \{p\}$.

Without loss of generality, assume $N>q$. Let us define the $q$ parts $D_j$ of $[0,N[^2$ defined by the tests $t_j$:
$D_j\coloneqq \{(x,y)\in [0,N[^2 \;\mid\; t_j(x,y)\}$. 
The partition property implies $\sum_{j=1}^q \mathtt{card}(D_j)=N^2$ from which it comes -- by the pigeonhole principle -- that at least one of the $q$ sets $D_j$ has at least $N^2/q$ elements. For example, assume $\mathtt{card}(D_1)\geq N^2/q$. 

Also, by~\ref{eq: PlusSucc}, the implication $t_1(x,y) \Rightarrow x+y=a_1$  is valid.
Therefore, we get \linebreak
$\mathtt{card}\{(x,y)\in [0,N[^2 \;\mid\; x+y=a_1\}\geq \mathtt{card}(D_1)\geq N^2/q$.
Now, suppose that $a_1$ is of the form $f(x)$ (the case $f(y)$ is symmetrical and the case $f(0)$ is simpler).
By a new application of the pigeonhole principle, we deduce that there is at least one integer $x_0\in [0,N[$ such that 
\[
 \mathtt{card} \{y\in [0,N[\; \;\mid\; x_0+y=f(x_0)\}\geq N/q
 \]
 From the assumption $N>q$, we deduce that there are at least two distinct integers $y_1,y_2\in[0,N[$ such that $x_0+y_1=f(x_0)=x_0+y_2$, a contradiction.
 This concludes the proof of Proposition~\ref{prop: AddNotCTonSuccRAM} by contradiction. 
\end{proof}
 


\begin{remark}
Using the pigeonhole principle again, it is also easily seen that not only addition but also 
each \emph{binary} operation whose result really depends on the two operands cannot be computed in constant time on a RAM \emph{using only unary operations}. 
Intuitively, a constant number of comparisons (branching instructions) generates a constant number of cases: to compute a binary operation, it is not enough to use a constant number of cases defined by (in)equalities over “unary” terms, each of which depends on only one operand.
\end{remark}

\subsection{Can addition be replaced by another arithmetic operation?}

Although addition is the simplest of the four arithmetic operations
and perhaps for this reason, it is among them the only ``fundamental''
operation of the RAM model, as this subsection shows.

\bigskip
The question we attempt to answer in this subsection is: Can addition be replaced by one of the other three standard operations, subtraction, multiplication $\times$ and Euclidean division $\mathtt{div}$? If not, can it be replaced by two of them?
 
Obviously, the addition can be replaced by the subtraction over $\mathbb{Z}$, i.e.\ relative
integers, as \linebreak
$x+y=x-(0-y)$. 
Addition can also be replaced by subtraction
\emph{and} multiplication (or multiplication by~2) on natural numbers
because of the following identity where \linebreak
$x\dot{-}y\coloneqq\mathtt{max}(0,x-y)$ denotes our subtraction on
natural numbers:
\begin{eqnarray*}
x+y=
\left\{
\begin{array}{l}
2*x\dot{-}(x\dot{-}y)\;\mathtt{if}\; y\dot{-}x=0,\;\mathtt{i.e.}\; y\leq x \\
2*y\dot{-}(y\dot{-}x)\;\mathtt{otherwise}\; \\
\end{array}
\right .
\end{eqnarray*}

Obviously also, subtraction and division are not enough to compute the successor function, and a fortiori the addition, because, for all integers $x$ and $y>0$, we have $x\dot{-}y\leq x$ and \linebreak
$x \divop y\leq x$.
More fully, we state the following negative results.

\begin{proposition}
\begin{enumerate}
\item No RAM with subtraction \emph{and} division over natural numbers
  computes the successor function.
\item No RAM with multiplication computes the successor function.
\end{enumerate}
\end{proposition}

\begin{proof}
As we have given the argument for item 1, it suffices to justify item
2. Suppose $M$ is a RAM with multiplication which computes the
function $N\mapsto N+1$. Let us fix an integer $N$ greater than all
the numbers $j_1,\ldots,j_k$ that appear in the program of $M$ and
such that $N+1$ is a prime number. By construction, each number that
appears in the computation of $M$ over the input $N$ must be a product
of integers belonging to the set $\{N,j_1,\ldots,j_k\}$ and therefore
cannot be equal to $N+1$, because $N+1$ is prime and greater than all
the factors, a contradiction.
\end{proof}

We therefore have $\RAM[\{\dot{-},\times\}]=\RAM[\{+\}]$ but the RAM models $\RAM[\{\dot{-},\mathtt{div}\}]$ and $\RAM[\{\times\}]$ are strictly weaker. 

Also, we do not know if addition can be replaced by multiplication
\emph{and} division, i.e.\ if we have
$\RAM[\{\times,\mathtt{div}\}]=\RAM[\{+\}]$, even if we think the
answer is no.

%% file: 10_Conclusion.tex
\section{Final results, concluding remarks and open problems}\label{sec:conclusion}

We hope that the results established in the previous sections have convinced the reader that our RAM model with addition, also called \emph{addition RAM}, or $\RAM[\{+\}]$, is the standard model for the design and analysis of algorithms. 
This means that the complexity classes defined in this model are robust and faithfully reproduce/model the intuitive complexity classes of concrete algorithms. 
This is already true for the ``minimal'' complexity classes defined as follows. 

\paragraph{``Minimal'' complexity classes of addition RAM}

We now recall the complexity classes presented in Section~\ref{sec:RAMcomplexityclasses}:
\begin{itemize}
\item $\lin$ is the class of functions $\mathcal{I}\mapsto f(\mathcal{I})$ computed in \emph{linear time} by an addition RAM.
\item $\clin$ is the class of ``dynamic'' problems $(\mathcal{I},\mathcal{X})\mapsto f(\mathcal{I},\mathcal{X})$ computed in \emph{constant time} by an addition RAM with \emph{linear-time preprocessing}, which means:

\begin{enumerate}
\item \emph{Pre-computation:} from an input $\mathcal{I}$ of size $N$, the RAM computes an ``index" $p(\mathcal{I})$ in time $O(N)$;
\item \emph{Computation:} from $p(\mathcal{I})$ and after reading a second input $\mathcal{X}$ of constant size, the RAM computes in \emph{constant time} the output $f(\mathcal{I},\mathcal{X})$.
\end{enumerate}
\item $\cdlin$ (see~\cite{DurandG07,BaganDG07,Bagan09,Durand20}) is the class of enumeration problems $\mathcal{I}\mapsto f(\mathcal{I})$ computed with \emph{constant delay} and \emph{constant-space} by an addition RAM with \emph{linear-time preprocessing}, 
which means:
\begin{enumerate}
\item \emph{Pre-computation:} from an input $\mathcal{I}$ of size $N$, the RAM computes an ``index" $p(\mathcal{I})$ in time $O(N)$;
\item \emph{Enumeration:} from $p(\mathcal{I})$ and using
  \emph{constant space}, the RAM enumerates without repetition the
  elements of the set $f(\mathcal{I})$ with a \emph{constant delay}
  between two successive elements and stops immediately after
  producing the last one. The constant-space criterion means that
  the enumeration algorithm can read all the memory initialized during
  the pre-computation but can only write to a constant number of fixed registers\footnote{
The $\cdlin$ class naturally generalizes to the case where the sizes of the elements of the set $f(\mathcal{I})$ to be enumerated are arbitrary instead of being constant:
the enumeration problem $f$ is of ``minimal'' complexity if the time of the pre-computation is still
$O(\lvert\mathcal{I}\vert)$ and each enumerated element $S\in f(\mathcal{I})$ is computed 
in time and space $O(\lvert S \rvert)$. 
The time/delay and space used to produce a solution are linear in the size of that solution~\cite{Bagan06,Bagan09,Courcelle09}.}. 
\end{enumerate}
\end{itemize}

\paragraph{More liberal constant-delay complexity classes with linear-time preprocessing.} 
In the literature on enumeration algorithms, the space complexity of
the enumeration phase is often not accurately described or appears to
be larger than our constant space bound, see e.g.~\cite{Segoufin14,
  AmarilliBJM17, Durand20}.  
We think that the most realistic 
extensions of our $\cdlin$ class are the following:

\begin{itemize}
\item Let $\cdlinlinspace$ denote the extension of the class $\cdlin$
  where the space bound of the enumeration phase is $O(N)$, instead of
  constant.  Note that this complexity class can be defined in our RAM
  model which only uses $O(N)$ integers (contents and addresses of
  registers).
\item \emph{Constant delay with polynomial space}: 
One can also allow the constant-delay enumeration phase to use a
  set of registers of size $O(g(N))$ for some $g$. Note that if we
  don't have $g(N)=O(N)$ then this is, by default, forbidden by our
  RAM model but our ``extended'' model supports it for polynomial $g$
  (i.e.\ $g(N)=O(N^d)$ for some $d$), 
  as described in the paragraph ``RAM model with more memory'' of Section~\ref{sec:more_memory}.
\end{itemize}


\paragraph{The operations we have studied.} We have extended the RAM model with the following two kinds of operations:
\begin{itemize}
\item Let $\mathtt{Op}_{\mathtt{Arith}}$ be the set of arithmetic operations that follow: addition, subtraction, multiplication, division, exponential, logarithm, and $c$-th root (for any fixed integer $c\ge 2$), acting on ``polynomial'' integers, i.e.\ integers less than~$N^d$, for a fixed $d$.
\item Let $\mathtt{Op}_{\mathtt{CA}}^{\mathtt{lin}}$ be the set of operations on ``polynomial'' integers computable in linear time on cellular automata.
\end{itemize}

The following corollaries summarize our main results and establish the robustness and expressiveness of the complexity classes thus defined in our RAM model.

\begin{corollary}\label{cor:invOp+}
Let $\mathtt{Op}$ be a finite set of operations such that 
$\{+\}\subseteq \mathtt{Op}\subseteq 
\mathtt{Op}_{\mathtt{Arith}}\cup \mathtt{Op}_{\mathtt{CA}}^{\mathtt{lin}}$. 
Then we have $\RAM{}=\RAM[{+}]$.
\end{corollary}

\begin{corollary}\label{cor:invRAM+}
The classes $\lin$, $\clin$, $\cdlin$ and $\cdlinlinspace$ do not change if the set of primitive operations of the RAMs is any finite set of operations containing addition and included in 
$\mathtt{Op}_{\mathtt{Arith}}\cup \mathtt{Op}_{\mathtt{CA}}^{\mathtt{lin}}$.
\end{corollary}

\begin{proof}
Each of these four complexity classes allows linear-time preprocessing and we have proved that each operation in $\mathtt{Op}_{\mathtt{Arith}}\cup \mathtt{Op}_{\mathtt{CA}}^{\mathtt{lin}}$ can be computed by an addition RAM in constant time after linear-time preprocessing.
\end{proof}

\begin{remark}
For simplicity, we have chosen to state Corollary~\ref{cor:invRAM+} for the ``minimal'' complexity classes of the RAM model. 
However, it is easy to deduce the same result for complexity classes of greater time and/or greater space, e.g.\ $\mathtt{Time}(N^2)$, etc.
\end{remark}

\paragraph{Addition of randomness?}
Our RAM model is purely deterministic but we can easily add to our
model a source of randomness, for instance by having an instruction
\texttt{RANDOM} that fills the accumulator $A$ (of an $AB$-RAM) with
some random value in $[0,cN[$, for some constant integer~$c$.

Adding randomness to a RAM model does not change the worst case
complexity of well-defined problems as, in the worst case, the random
instruction might behave exactly like the instruction \texttt{CST} 0. 
For that reason, the addition of randomness is only useful when
talking about complexity classes where the answer is allowed to be
only \emph{probably} correct (as opposed to \emph{always} correct) or about
programs terminating in \emph{expected time} $T(N)$, for some function $T$.

In computational complexity theory, there are open conjectures about
whether the addition of a source of randomness to Turing machines
might make them exponentially faster.  Since the classes of problems
computable in polynomial or exponential time in our RAM or in a Turing
machine are the same, there is little hope to determine whether a RAM
machine with randomness is equivalent to a RAM without randomness.

\paragraph{Weakness of the RAM with successor and predecessor?}

Let $\linsp$, $\clinsp$, and $\cdlinsp$ denote the complexity classes similar to the $\lin$, $\clin$, and \linebreak
$\cdlin$ classes, respectively, when the successor and predecessor functions, $x\mapsto x+1$ and $x\mapsto x-1$, where $x<cN$, for some constant integer $c>0$, are the only primitive operations of the RAM model instead of the addition operation.

\medskip
In contrast with Corollary~\ref{cor:invRAM+}, the following ``negative'' result holds.

\begin{corollary}\label{cor:strictIncConstLinSucc} 
We have the strict inclusion $\clinsp\subsetneqq \clin$.
\end{corollary}

\begin{proof}
The inclusion is a direct consequence of what we have proved before. 
So, it suffices to prove that it is strict.
Let us consider the ``dynamic'' problem $(\mathcal{I},\mathcal{X}=(x,y))\mapsto x+y$, for $x,y\in[0,N[$, where $N$ is the size of $\mathcal{I}$. This problem trivially belongs to $\clin$ but it does not belongs to $\clinsp$ by Proposition~\ref{prop: AddNotCTonSuccRAM}.
\end{proof}

We ignore if Corollary~\ref{cor:strictIncConstLinSucc} can be ``extended'' to the other ``minimal'' complexity classes.

\begin{openpb}\label{conj:succ moins que plus pour lin}
Is the inclusion $\linsp\subseteq \lin$ strict?
\end{openpb}

\begin{openpb}\label{conj:succ moins que plus pour cdlin}
Is the inclusion $\cdlinsp\subseteq \cdlin$ strict?
\end{openpb}

The difficulty comes from the fact that RAM with successor, also known as successor RAM, is much more powerful than it seems at first glance. 
As a striking example, Schönhage~\cite{Schonhage80} showed as early as 1980 that integer-multiplication, which is the problem of computing the product of two integers in binary notation, i.e., with one bit per input register~\footnote{Note that in our formalism, the input of this problem is of the form \linebreak
$\mathcal{I}\coloneqq (N,I[0]=a_0,\ldots, 
I[n-1]=a_{n-1},I[n]=2,I[n+1]=b_0,\ldots,I[2n]=b_{n-1})$, where $N=2n+1$ and the $a_i,b_i$ belongs to $\{0,1\}$, and its output is $(2n,p_0,\ldots,p_{2n-1})$, where $p=a\times b$, for $a=\overline{a_{n-1}\ldots a_1a_0}$,  $b=\overline{b_{n-1}\ldots b_1b_0}$, and $p=\overline{p_{2n-1}\ldots p_1p_0}$.}, is computed in linear time on a successor RAM and therefore belongs to $\linsp$. 
Moreover, Schönhage~\cite{Schonhage80} proved that the successor RAM is equivalent to what he calls the ``Storage Modification Machine'', also called ``Pointer Machine" by Knuth~\cite{Knuth68} and Tarjan~\cite{Tarjan75}, who proves that in this computational model, Union-Find operations can be performed in amortized quasi-constant time $\alpha(n)$, where $\alpha(n)$ is the inverse of Ackermann's function $A(n,n)$.

\paragraph{Why does our RAM model model the algorithms in the literature?}

Our RAM model faithfully mimics/reproduces the computations of the
algorithms on the combinatorial structures: trees, formulas, graphs, hypergraphs, etc.
The reason is twofold: 

\begin{enumerate}

\item Such a structure $\mathcal{S}$, e.g.\ a tree or a graph, is naturally encoded by a RAM input \linebreak 
$\mathcal{I}=(N,I[0],\ldots,I[N-1])$, with $I[j]=O(N)$ and $N=\Theta(L)$, where $L$ is the ``natural length'' of the structure $\mathcal{S}$, which is $L=n$ if $\mathcal{S}$ is a tree of $n$ nodes, or $L=m+n$ if $\mathcal{S}$ is a graph of $m$ edges and $n$ vertices, etc.

\item The structure and acting of the RAM model are homogeneous: 
on the one hand, its input array $\mathcal{I}=(N,I[0],\ldots,I[N-1])$ is of the same nature as the array $R[0],R[1],\ldots$ of its registers of contents also $O(N)$;
on the other hand, a read/write instruction is a random access instruction like all other instructions.



\end{enumerate}

\paragraph{Why RAM linear time is intuitive linear time under unit cost criterion?}
By our conventions, any \emph{linear-time} (and linear-space\footnote{As is the case with almost all linear-time algorithms in the literature and as required for RAMs, we require that a linear-time algorithm use linear space.}) algorithm from the literature on combinatorial structures (trees, graphs, etc.) is modeled/implemented by a RAM algorithm that uses $O(N)$ registers and executes $O(N)$ instructions, each involving $O(1)$ registers of bit length $\Theta(\log N)$. 
Consequently, its time complexity under unit cost criterion is $O(N)$, which is $O(L)$, where $L$ is the ``natural length'' of the input structure, by item 1 of the above paragraph. 

\medskip
What about linear time under logarithmic cost criterion? 
Since a linear-time algorithm performs $O(N)$ instructions, each of cost $O(\log N)$, its time is 
$O(N \log N)$ under logarithmic cost criterion. 
Is it a linear time for this criterion?

\paragraph{Linear-time algorithms on combinatorial structures are also linear under logarithmic cost criterion.}
Let us justify this assertion for graph algorithms. (The proof is similar for algorithms on other structures: trees, formulas, etc.) 
A graph without isolated vertex whose vertices are $1,\ldots, n$ and  edges are 
$(a_1,b_1),\ldots, (a_m,b_m)$ is ``naturally represented'' by the standard input 
$\mathcal{I}=(N,I[0],\ldots,I[N-1])$ where $N\coloneqq 2m+2$,  $I[0]\coloneqq m$,  $I[1]\coloneqq n$,  and finally $I[2j]\coloneqq a_j$ and $I[2j+1]\coloneqq b_j$ for $j=1,\ldots, m$. 
Recall that we have $n\leq 2m=N-2$.


If $\mathcal{I}$ is encoded by a binary word, its bit length, called
$\Call{length}{\mathcal{I}}$, is
$\Theta(\Call{length}{N}+\sum_{j=1}^m\Call{length}{\Call{binarycode}{a_j,b_j})}$,
which is $O(N\log N)$. 
On the other hand, since the $m$ edges $(a_j,b_j)$ are pairwise distinct, we have
\begin{eqnarray}\label{eqn:ineq_shortest_length_words}
\sum_{j=1}^m \Call{length}{\Call{binarycode}{a_j,b_j}}\geq \sum_{j=1}^m \Call{length}{w_j}
\end{eqnarray}
where $w_1,\ldots,w_m$ is the list of the $m$ shortest nonempty binary words numbered in increasing length, that means $\Call{length}{w_j}\leq \Call{length}{w_{j+1}}$, for each $j<m$.
The inequality~(\ref{eqn:ineq_shortest_length_words}) is justified by the following points:
\begin{itemize}
\item without loss of generality, assume that the $m$ edges $(a_1,b_1),\ldots, (a_m,b_m)$ are numbered in increasing length of their (pairwise distinct) encodings $\Call{binarycode}{a_j,b_j}$, $j=1,\ldots,m$;
\item consequently, we have $\Call{length}{\Call{binarycode}{a_j,b_j}}\geq \Call{length}{w_j}$, for each $j\leq m$.
\end{itemize}
Moreover, the (in)equality $\sum_{j=1}^m \Call{length}{w_j}=\Omega(m \log m)$ 
is intuitive and not hard to prove\footnote{
Note the sequence of equivalences $\Call{length}{w_j}=\ell$ $\iff$ 
$\sum_{i=1} ^{\ell-1} 2^i < j \leq \sum_{i=1} ^{\ell} 2^i$ $\iff$
$2^{\ell}-2<j \leq  2^{\ell+1}-2$ $\iff$
$2^{\ell}<j+2 \leq  2^{\ell+1}$ $\iff$
$\ell=\lc\log_2(j+2)\rc -1$.
Hence, $\Call{length}{w_j}=\lc\log_2(j+2)\rc -1$.
Therefore, we get $\sum_{j=1}^m \Call{length}{w_j} \geq
\sum_{j=\lc m/2 \rc}^m \Call{length}{w_j} \geq
(m/2) \Call{length}{w_{\lc m/2 \rc}} = (m/2) \lc \log_2(\lc m/2 \rc+2) \rc -1$, from which
$\sum_{j=1}^m \Call{length}{w_j}=\Omega(m \log m)$ is deduced.
}
.
By~(\ref{eqn:ineq_shortest_length_words}), this implies
$\Call{length}{\mathcal{I}}=\Omega(m\log m)=\Omega(N\log N)$.
Thus, we have proved\footnote{
Note that the truth of~(\ref{eqn:unitcost_versus_logcost}) is largely independent of the representation of the graph provided it is ``natural''. E.g., if the vertices of the graph are $1,\ldots,n$, some of which may be isolated, and its edges are $(a_1,b_1),\ldots, (a_m,b_m)$, the graph is ``naturally'' represented by the standard input $\mathcal{I}=(N,I[0],\ldots,I[N-1])$ where $N\coloneqq 2m+n+1$,  $I[0]\coloneqq m$,  $I[j]\coloneqq j$ for $j=1,\ldots, n$, and finally $I[2j+n-1]\coloneqq a_j$ and $I[2j+n]\coloneqq b_j$ for $j=1,\ldots, m$. 
We still have $\Call{length}{\mathcal{I}}=O(N\log N)$. 
Moreover, by the same reasoning as above, we get
$\Call{length}{\mathcal{I}}=\Omega(n \log n + m \log m)=\Omega(N \log N)$ 
(because $N=\Theta(max(m,n))$) and therefore 
$\Call{length}{\mathcal{I}}=\Theta(N \log N)$.
}
\begin{eqnarray}\label{eqn:unitcost_versus_logcost}
\Call{length}{\mathcal{I}}=\Theta(N\log N)
\end{eqnarray} 
Consequently, a graph algorithm running in linear time under unit cost criterion, therefore in time $O(N\log N)$ under logarithmic cost criterion, runs in time $O(\Call{length}{\mathcal{I}})$ under the same criterion, according to~(\ref{eqn:unitcost_versus_logcost}). 
This is true linear time under the logarithmic cost criterion!

Thus, algorithms on graphs or similar combinatorial structures (trees, formulas, hypergraphs, etc.) running in linear time under the unit cost criterion also run in linear time, i.e.\ time $O(\Call{length}{\mathcal{S}})$, under the logarithmic cost criterion, where $\Call{length}{\mathcal{S}}$ is the length in bits of the input structure $\mathcal{S}$. 

However, our argument above for combinatorial structures does \emph{not} apply to \emph{word structures}.

\paragraph{Which linear time complexity for algorithms on words?}
As seems most natural, it is the rule for a text algorithm (see e.g.\ the reference book~\cite{CrochemoreR94} or Chapter 32 of~\cite{CormenLRS09}) to represent an input word (text, string) 
$w=w_0\ldots w_{n-1} \in \Sigma^n$ over a finite fixed alphabet $\Sigma$, identified with the set of integers
$\{0,1,\ldots,b-1\}$ for $b=\mathtt{card}(\Sigma)$, by an array 
$\mathtt{W}[0..n-1]$ where $\mathtt{W}[j]=w_j$.
Thus, the text algorithm is faithfully simulated by a RAM whose input is 
$\mathcal{I}=(N,I[0],\ldots,I[N-1])$ where $N\coloneqq n$ and $I[j]\coloneqq w_j$; 
this means $0\leq I[j] < b = O(1)$.

Many word algorithms of the literature, e.g.\ the well-known string-matching algorithm of Knuth, Morris and Pratt~\cite{KnuthMP77,CrochemoreR94,Gusfield1997}, are executed in the RAM model in time $O(N)$ under the unit cost criterion, where $N$ is the length of the input word $w$ (or the maximum of the lengths of the two input words). 
Therefore, such an algorithm runs in time $O(N\log N)$ under the logarithmic cost criterion, for $N=\Call{length}{w}$. 
A priori, it is not $O(\Call{length}{w})$...
This corresponds to the fact that the input elements $I[j]$ are in $O(1)$ while the contents of the memory registers are in $O(N)$: 
here, there is \emph{heterogeneity between input and memory}.

To recover homogeneity, we can adopt a technical solution already introduced in \cite{Grandjean96, 
Schwentick97, GrandjeanSchwentick02}: 
partition any input word $w\in\Sigma^n$ 
(where $\Sigma=\{0,\ldots,b-1\}$) into
$N=\Theta(n/\log n)$ factors of length $\Theta(\log n)=\Theta(\log N)$. 
More precisely, using the parameters 
$\ell \coloneqq \lc (\log_b n)/2 \rc$ and $N \coloneqq \lc n/\ell \rc$,
we partition $w$ into the concatenation $w=f_0\ldots f_{N-1}$ of $N$ factors $f_j$, 
with $\Call{length}{f_j}=\ell$ for $j<N-1$ and 
$1\leq \Call{length}{f_{N-1}}\leq \ell$.
The word $w\in\Sigma^n$ is represented by the input $\mathcal{I}=(N,I[0],\ldots,I[N-1])$ defined by 
$I[j]\coloneqq f_j$ for each $j<N$. 
Many other variants of this representation can also be adopted, the important point being the orders of magnitude obtained: 
\begin{center}
$N=\Theta(n/\log n$) ; $I[j]=O(b^{\ell})=O(b^{(\log_b n)/2})=O(n^{1/2})=O(N)$, for each $j<N$.
\end{center}
With such a word representation, the time complexity of a RAM
operating in time $O(N)$, which is $O(n/\log n)$, under the unit cost
criterion, is $O(N\log N)=O(n)$, i.e.\ \emph{truly} linear,
under the logarithmic cost criterion.  
It would be a large and challenging goal to determine which word problems known to be computable in
``linear time'' can be computed in time $O(n/\log n)$ under the unit cost criterion using our word representation. 
For example, it is
true for the following ``sorting''
problem~\cite{Grandjean96} from $(\Sigma\cup\{\sharp\})^+$ to $(\Sigma\cup\{\sharp\})^+$,
for $\sharp\not\in\Sigma$, defined as
\[
x_1\sharp x_2\sharp\cdots \sharp x_q \mapsto x_{\pi(1)}\sharp x_{\pi(2)}\sharp\cdots \sharp x_{\pi(q)}
\]
where $x_1,x_2,\ldots,x_q$ are words in $\Sigma^+$, and $\pi$ is a permutation of $\{1,2,\ldots,q\}$ such that \linebreak
$x_{\pi(1)}\leq x_{\pi(2)}\leq\ldots \leq x_{\pi(q)}$ for the lexicographic order $\leq$ of $\Sigma^+$.

\paragraph{Is multiplication more difficult than addition in the RAM model?}
We mentioned above that Schönhage~\cite{Schonhage80} surprisingly proved that the successor RAM, equivalent to his ``storage modification machine'' (but seemingly weaker than our addition RAM!) can compute in time $O(n)$ the product of two integers of length $n$ in binary notation, each contained in $n$ input registers (one bit per input register). 

On this point, in his lecture on receiving the ACM Turing award~\cite{Cook83} in 1982, Stephen Cook expressed his astonishment as follows:
`` Schönhage~\cite{Schonhage80} recently showed$\dots\dots$that his storage modification machines can multiply in time $O(n)$ (linear time!). 
We are forced to conclude that either multiplication is easier than we thought or that Schönhage’s machines cheat.''

Although we agree that Schönhage's result is very beautiful and strong (it is all the more powerful that it uses a successor RAM instead of an addition RAM!), we believe that the paradox highlighted by Cook is only apparent. 
The product of two integers of binary length $n$ is computed by a successor RAM in time $O(n)$ under the unit cost criterion, and therefore in time $O(n \log n)$ under the logarithmic cost criterion. 
This is not real linear time!

At the opposite, we have noticed, see Remark~\ref{rem:complexitySumDiff}, that our addition RAM can compute the sum (resp.\ difference) of two integers of binary length $n$ in time $O(n/ \log n)$ under the unit cost criterion
(by representing each input integer by the array of binary words of length $\lc (\log_2 n)/2 \rc$ whose concatenation is its binary notation, like the word representation we suggest in the previous paragraph).
This is $O(n)$ time under the logarithmic cost criterion, which is \emph{really} linear time.
This is not the case with the Schönhage multiplication algorithm or the Fast Fourier multiplication and we conjecture that our RAM model cannot compute the product of two integers of binary length $n$ in time $O(n/ \log n)$ under the unit cost criterion as its does for their sum and difference.
Multiplication continues to seem harder than addition and subtraction!

\paragraph{Linear space or larger space?} The RAM model we have chosen uses \emph{linear space}. This means that the addresses and contents of registers $R[0],R[1],\dots$ are integers in $O(N)$; 
in particular, no $k$-dimensional array $N_1\times N_2\times\dots\times N_k$ is allowed unless it uses
$N_1\times N_2\times\dots\times N_k= O(N)$ registers so that it can be represented by a one-dimensional array of size $O(N)$. 
This is justified by two arguments:
\begin{itemize}
\item our RAM model models the vast majority of linear-time algorithms in the literature: they use a linear space;
\item our RAM model is very ``robust'': its complexity classes are invariant for many variations, in particular with respect to the primitive operations of the RAM, as established at length in this paper.
\end{itemize}
Nevertheless, we have noticed that several recent papers about queries in logic and databases, see e.g.\ ~\cite{DurandSS14, Segoufin14, BerkholzKS17,BerkholzKS17b}, present algorithms which use $k$-dimensional arrays $\mathbf{A}$ with $O(N^k)$ available registers, such that for given $(n_1,\dots,n_k)\in\mathbb{N}^k$, with each $n_i=O(N)$, the entry $\mathbf{A}[n_1,\dots,n_k]$ at position $(n_1,\dots,n_k)$ can be accessed in constant time. 
Although this ``multidimensional'' RAM model is ``unrealistic for real-world computers'', as~\cite{BerkholzKS17} writes, it is robust in two respects:
\begin{itemize}
\item it inherits all the invariance properties of our (stricter) RAM model according to the primitive operations;
\item it is invariant according to the dimensions of arrays: Theorem~\ref{thm:mem} states that the RAM model with arrays of arbitrary dimensions is equivalent to the RAM model using only 2-dimensional arrays, and even, arrays of dimensions $N \times N^\epsilon$, for any $\epsilon>0$, hence with 
$O(N^{1+\epsilon})$ registers available.
\end{itemize}
It remains the open and difficult question of whether the RAM model
with $k$-dimensional arrays but only using $O(N)$ cells is equivalent
to the RAM model with only $1$-dimensional arrays of~$O(N)$ size.
Answering this question is crucial to determine whether the algorithms
of~\cite{DurandSS14, Segoufin14, BerkholzKS17,BerkholzKS17b} can be
implemented in linear space in ``real-world computers''.
\paragraph{Which space complexity to choose for enumeration algorithms?} In accordance with and in the spirit of the definition of the complexity class $\clin$, whose second phase uses constant time and therefore constant space, the ``minimal" enumeration complexity class that we have chosen as the standard (when each solution to be enumerated is of constant size\footnote{Note that if the solutions are not of constant size, e.g.\ if they are subsets of vertices of an input graph, then the ``minimal''  (standard) enumeration complexity class is the class of problems, called $\lindlin$, whose preprocessing phase uses $O(N)$ time and space and enumeration phase uses $O(\vert S \vert)$ delay and space, for each solution $S$, see~\cite{Bagan06,Bagan09,Courcelle09}. Of course, belonging to $\lindlin$ amounts to belonging to $\cdlin$ for an enumeration problem whose size of solutions is constant.}
) is the class $\cdlin$ whose second phase uses constant time \emph{and also} constant space. This is justified by two points:
\begin{itemize}
\item some basic enumeration algorithms in combinatorics, logic and database theory~\cite{Feder94, BaganDG07, DurandG07, Bagan09, Durand20} belong to $\cdlin$;
\item this complexity class is as ``robust'' (and for the same reasons) as the classes $\lin$ and $\clin$.
\end{itemize}
The a priori larger complexity class $\cdlinlinspace$, which allows the enumeration phase to use $O(N)$ space, instead of constant space, is also realistic (its two phases use a linear space) 
and as robust as $\cdlin$.

Again, the question of whether the inclusion $\cdlin \subseteq \cdlinlinspace$ is strict remains an open and difficult problem.

\paragraph{What other operations are computed in constant time?}
We have the feeling that our addition RAM can compute in constant time
(with linear-time preprocessing) any operation on ``polynomial''
integers defined from the usual arithmetic operations (whose versions
on $\mathbb{R}$ are continuous), plus the rounding functions $x\mapsto
\lf x \rf$ and $x\mapsto \lc x \rc$, and the compositions of such
operations, provided that the intermediate results are ``polynomial''.
We managed to prove this for all the operations we studied in
Sections~\ref{sec:SumDiffProdBase} to~\ref{sec:roots}, except for the
general root operation $(x,y) \mapsto \lf x^{1/y} \rf$ for which
Subsection~\ref{subsec:genroot} gives only a partial proof.  However,
we have failed to prove any general result -- except for the set of
linear-time computable operations on cellular automata --- and our
proofs are ``ad hoc'' even though a discrete analog of Newton's
approximation method used in Section~\ref{sec:roots} might be repeated
or ``generalized'' in some way.

\paragraph{Final remarks on the RAM model.} 
As observed and discussed in the introduction to this paper, on the one hand the RAM model is generally considered to be the standard model in the foundations of algorithm design and algorithm analysis, and on the other hand, surprisingly, the literature does not agree on a standard definition of the RAM model and its complexity classes.

Although in the past some papers~\cite{Grandjean94,Grandjean94-cc, Grandjean96,GrandjeanSchwentick02, GrandjeanOlive04,
DurandG07,Bagan09} have defined and started to study the $\lin$ class of our RAM model with addition and its subclass of decision problems called $\dlin$, so far no systematic work has investigated from scratch the robustness and extent of computational power of this model and its complexity classes.

Having such a systematic and unified study was all the more urgent as the many recent studies on the algorithms and the complexity of ``dynamic'' problems, and especially of enumeration problems~\cite{Segoufin13,Segoufin14,AmarilliBJM17,BerkholzKS18,KazanaS19,Strozecki19,CarmeliZBKS20,BerkholzGS20,Durand20}, essentially refer to the reference book~\cite{AhoHU74} (excellent but almost half a century old!) and to the scattered and incomplete papers cited above. 
This gives an impression of imprecision and approximation and can lead to ambiguities or even errors\footnote{Most of the papers cite the book~\cite{AhoHU74} which considers the four operations $+,-,\times$ and $/$ as primitive operations of RAMs without limiting the values of their operands. However, Schönhage~\cite{Schonhage79} proved that the class of languages recognized by such RAMs in polynomial time under the unit cost criterion contains the entire class NP!}.

We hope that this self-contained article meets this need for precision. We have designed our paper as a progressive, detailed and systematic study. 

We consider that the main arguments in favor of adopting the addition RAM as the basic algorithmic model are as follows:
\begin{itemize}
\item its simplicity/minimality and its homogeneity;
\item its ability to represent usual data structures and in particular tables, essential for preprocessing;
\item the robustness and computing power of its complexity classes (Corollary~\ref{cor:invRAM+}): 
the most important and surprising result that this paper establishes is that the addition RAM can compute the Euclidean division in constant time with linear-time preprocessing; 
as we also proved, this allows to compute many other arithmetic operations on ``polynomial'' integers in constant time with linear-time preprocessing.
\end{itemize}

\paragraph{Open problems.} Even if we have positively solved in this paper the main (old) open question, which is the equivalence of the addition RAM with the RAM equipped with the four operations $+,-,\times$ and $ / $ acting on ``polynomial'' integers, for the complexity classes under the unit cost criterion, there remain many other interesting open questions regarding the RAM model, several of which have been presented throughout this paper.

We have tried to establish the most complete list of integer
operations that an addition RAM is able to compute in
constant time with linear-time preprocessing.  In fact, many
algorithms deal with combinatorial structures such as trees, circuits,
graphs, etc.  The next step would therefore be to undertake a similar
systematic study of operations on combinatorial structures computable
in constant time (or quasi-constant time) on addition RAMs with
linear-time preprocessing.  For example, although a lot of work has
been done in this direction, it would be very useful to establish the
``largest'' set of ``atomic'' tree operations computable in constant
(or quasi-constant) time after a linear (or quasi-linear) time
preprocessing\footnote{For instance, as soon as 1984, Harel and
Tarjan~\cite{HarelT84} proved that the closest common ancestor of any
two nodes $x,y$ of a tree $T$ of $N$ nodes can be computed in constant
time after a preprocessing of time $O(N)$ and depending only on $T$.}.
This would constitute a toolbox to optimize algorithms for ``dynamic''
problems such as those studied in~\cite{AmarilliBJM17,BerkholzKS17,
  BerkholzKS18b, BerkholzKS18, AmarilliJP21}.

%% file: Appendix.tex
\section{Appendix: Proof of Theorem~\ref{th:linTimeCA} of Subsection~\ref{subsec:TuringCA}}

Let us first recall the result.

\medskip \noindent
{\bf Theorem}~\ref{th:linTimeCA}.
\emph{
Let $d\geq 1$ be a constant integer and let $\mathtt{op}: \mathbb{N}^r\to \mathbb{N}$ be an operation of arity 
$r\geq 1$, computable in linear time on a cellular automaton (CA). 
There is a RAM with addition, such that, for any input integer $N>0$:
\begin{enumerate}
\item \emph{Pre-computation:} the RAM computes some tables in time $O(N)$;
\item \emph{Operation:} from these pre-computed tables and after reading $r$ integers $X_1,\ldots,X_r<N^d$, the RAM computes in constant time the integer $\mathtt{op}(X_1,\ldots,X_r)$. 
\end{enumerate}
}

\begin{proof}

Let $\mathcal{A}=(Q,\delta)$ be a CA that computes in linear time an operation $\op$, i.e.\ the function 
$f_{\op}: X_1;\cdots;X_r\mapsto \op(X_1,\ldots,X_r)$. 
Before building a RAM that simulates the CA, let us slightly tweak the CA input configuration. 
In particular, we add the integer $N$ to the input as a reference to the upper bound $N^d$ of the operands $X_1,\ldots,X_r<N^d$.

Let us represent the $(r+1)$-tuple of integers $(N,X_1,\ldots,X_r)$, for $X_1,\ldots,X_r<N^d$, by the word, called
$\mathtt{code}(N,X_1,\ldots,X_r)\coloneqq u_{L-1}\cdots u_1 u_0\in \Gamma^L$, on the alphabet 
$\Gamma \coloneqq \{0,1,\ldots,\gamma-1\}$ where $\gamma \coloneqq 2^{r+1}$, which is defined as follows:
\begin{description}
\item [Length of the code:] $L\coloneqq \mathtt{length}(N^d)$; 
\item[Letters of the code:]  $u_j\coloneqq \sum_{i=0}^{r} x_{i,j} 2^i$, for $0\leq j<L$,
where $x_{i,L-1}\dots x_{i,1} x_{i,0}$ denotes the binary representation of $X_i$ in $L$ bits, 
for $i\in\{0,1,\ldots,r\}$ and $X_0\coloneqq N$. In other words, $u_j$ is the integer whose binary representation is column $j$ of the 0/1 matrix $(x_{i,j})_{0\leq i \leq r, 0 \leq j <L}$ when the columns are numbered from $L-1$ (left) downto 0 (right).
\end{description}
\emph{Example:} For $r=2$, $d=1$, and the size integer $N=13$, which has the binary representation $1101$ of length $L=4$, and for $X_1=7$ and $X_2=9$, which have the binary representations (of length $L=4$) $0111$ and $1001$, respectively, we obtain $\gamma=2^{r+1}=8$,
$\Gamma=\{0,1,\ldots,\gamma-1=7\}$ and 
$\mathtt{code}(N,X_1,\ldots,X_r)=5327\in\Gamma^4$, see Table~\ref{table: ExEncodeInput}. 

\begin{table}[H]
  \centering
\begin{tabular}{|l||c|c|c|c|}
\hline
$X_2=9$ & 1 & 0 & 0 & 1 \\
\hline
$X_1=7$ & 0 & 1 & 1 & 1\\
\hline
$X_0=N=13$ & 1 & 1 & 0 & 1\\ 
\hline
\hline
$\mathtt{code}(N,X_1,X_2)$ &  $u_3=5$ & $u_2=3$ & $u_1=2$ & $u_0=7$ \\
\hline
\end{tabular}
\caption{Encoding the input $(N,X_1,\ldots,X_r)$}
\label{table: ExEncodeInput}
\end{table}

It is tedious but routine to build a cellular automaton that computes in linear time the function $g: \mathtt{code}(N,X_1,\ldots,X_r) \mapsto X_1;\cdots;X_r$. 
By Lemma~\ref{lemma:compose}, the composition of functions \linebreak
$f_{\op}\circ g: \mathtt{code}(N,X_1,\ldots,X_r) \mapsto \op(X_1,\ldots,X_r)$ is computed in linear time on a CA, also called for simplicity $\mathcal{A}=(Q,\delta)$, with the set of output states $Q_{out} \subset Q$. 
By Definition~\ref{def:lintimeCA} supplemented by Lemma~\ref{lemma:exactlintime} and Remark~\ref{lemma:binaryOutput}, this means that there exist a constant integer $c\geq 1$ and a projection 
$\pi: Q_{out}\cup\{\sharp\} \to \{0,1,\sharp\}$ such that, 
for all integers $N>0$ and $X_1,\ldots,X_r<N^d$, 
we have $\delta^{cL}(\sharp  \mathtt{Input} \sharp)=\sharp Y \sharp$,
\begin{itemize}
\item with $\mathtt{Input}\coloneqq \mathtt{code}(N,X_1,\ldots,X_r)$ 
and $L\coloneqq \mathtt{length}(\mathtt{Input})=\mathtt{length}(N^d)$,
\item and for a word $Y=y_{cL-1}\cdots y_0\in (Q_{out})^{cL}$ such that 
$\mathtt{Output} \coloneqq \pi(y_{cL-1})\cdots \pi(y_0)\in \{0,1\}^{cL}$
is the binary notation of the integer $\op(X_1,\ldots,X_r)$ written to $cL$ bits. 
\end{itemize}

We can assume $Q\coloneqq \{0,1,\ldots s-1\}$ with $\Gamma \coloneqq \{0,1,\ldots,\gamma-1\}$, 
$\gamma \coloneqq 2^{r+1}$, $\sharp\coloneqq \gamma$, 
and  \linebreak
$\Gamma\cup \{\sharp\} \subseteq~Q$, which implies $s\geq \gamma+1$.
For convenience, the active cells of the computation are numbered $cL-1,\ldots,1,0$, increasing from right to left (like the representation of integers in any base).
Thus, at any time $t\in[0,cL]$ of the computation, the configuration has the form 
$\sharp q_{cL-1}\cdots q_1q_0 \sharp$, 
where $q_i\in Q$ is the state of cell~$i$ (possibly $\sharp$) at time $t$.

\medskip
We are now ready to present the simulation of the CA by a RAM. 
Before describing it with all the necessary details, let us first give the guiding ideas of this simulation.

\paragraph{A simplified view of the simulation.} The trick is to divide the active interval of the computation, of length~$cL$, into a constant number of consecutive blocks of ``small'' length $\ell$, called \linebreak
$\ell$-blocks, numbered $\ldots 2,1,0$ (increasing from right to left!).
We take $\ell \coloneqq \epsilon \log_2 N$ for a ``small'' constant~$\epsilon$. (Precisely, we will define $\ell \coloneqq \lc L/D \rc$, for a large fixed integer $D$.) 
Similarly, we divide the computation, of time $cL$, into a constant number of computation intervals of duration~$\ell$, called ``global $\ell$-transitions''. 

By abuse of language, the string of states $B_j^t \in Q^{\ell}$ of an $\ell$-block at some time $t$ is called an $\ell$-block at time $t$. 
Therefore, the configuration at any time $t$ is the concatenation $\sharp B_{c_0}^{t} \cdots B_{1}^{t}\sharp$, for the constant $c_0 \coloneqq \lc cL/\ell \rc$.

The essential point is the ``locality'' of the cellular automaton: the state of any cell $x$ at any time $t+1$ (resp.\ $t+\ell$) is entirely determined by the states of the three cells $x-1,x,x+1$ (resp. the state string of the cell interval $[x-\ell,x+\ell]$) at time~$t$. 
As a consequence, an $\ell$-block $B_j^{t+\ell}$, i.e.\ the $\ell$-block $j$ at time $t+\ell$, is determined by the three 
$\ell$-blocks, $B_{j+1}^{t}$, $B_j^{t}$, $B_{j-1}^{t}$, which are the $\ell$-block itself and its left and right neighboring $\ell$-blocks at time $t$.
In other words, the $\ell$-block $B_j^{t+\ell}\in Q^{\ell}$ 
is determined by the concatenation $B_{j+1}^{t}B_{j}^{t}B_{j-1}^{t}\in Q^{3\ell}$. 
Such a concatenation is called a ``local configuration'' and the transformation 
$B_{j+1}^{t}B_{j}^{t}B_{j-1}^{t}\mapsto B_j^{t+\ell}$ from $Q^{3\ell}$ to $Q^{\ell}$ is called a ``local $\ell$-transition''. 
There are two crucial points:
\begin{description}
\item[What preprocessing in linear time?] The number of possible ``local configurations'', therefore the number of possible ``local $\ell$-transitions'' for the parameter $N$, called $\mathtt{NbLocalConf}(N)$, is the cardinality of the set $Q^{3\ell}$, which is ``small'', precisely $s^{3\ell}$;  
this will make it possible to pre-compute in $O(N)$ time a table of all possible ``local $\ell$-transitions'' for the integer $N$;
\item[How to compute $\op$ in constant time?] Since the time (resp.\ length of the active interval) of the computation is $cL$, then the computation (resp. the active interval) is the concatenation of $cL/\ell=O(1)$ global $\ell$-transitions (resp. $\ell$-blocks), and all the computation on \linebreak
$\mathtt{Input}\coloneqq \mathtt{code}(N,X_1,\ldots,X_r)$ can be done by performing $(cL/\ell)^2=O(1)$ local $\ell$-transitions; it is a constant time since each local $\ell$-transition is performed by consulting a single element in the table of local $\ell$-transitions.
\end{description}

\begin{remark}[processing the input and output in constant time]
In this simplified view, we have not mentioned the reading of the operands $X_1,\ldots,X_r<N^d$, with the encoding process giving the initial configuration 
$\sharp\mathtt{Input}\sharp$, for $\mathtt{Input}\coloneqq \mathtt{code}(N,X_1,\ldots,X_r)$, nor the final output process building $\op(X_1,\ldots,X_r)$ by the projection of the final configuration.
We will explain below how these input and output processes can also be performed in constant time.
\end{remark}

\begin{remark}
Our simplified view also overlooked the fact that in general, the time $cL$ of the computation is not a multiple of $\ell$.
In general, we have $cL=(c_0-1)\ell +\rho$, for the ``constant'' $c_0 \coloneqq \lc cL/\ell \rc$ and an integer $\rho \in [1,\ell]$.
Therefore, the computation breaks down into $c_0-1$ global \linebreak
$\ell$-transitions, each composed of $c_0$ local $\ell$-transitions,
and one global $\rho$-transition composed of $c_0$ local $\rho$-transitions of time $\rho$. Like a local $\ell$-transition, a  local $\rho$-transition produces, from a local configuration $(B_2,B_1,B_0)\in (Q^{\ell})^3$ of adjacent $\ell$-blocks at some time $t$, the value $B'_1\in Q^{\ell}$ of the intermediate $\ell$-block at time $t+\rho$ (instead of time $t+\ell$ for a local $\ell$-transition).

\end{remark}

\paragraph{The number of local configurations.} Let us now determine the value of the ``cutting'' constant~$D$ which determines the block length $\ell(N) \coloneqq \lc L(N)/D \rc$, 
for $L(N)\coloneqq \mathtt{length}(N^d)$, 
so that the number $\mathtt{NbLocalConf}(N)$
of the possible local configurations, for an input $\mathtt{code}(N,X_1,\ldots,X_r)$, 
is $O(N^{\sigma})$, for some $\sigma<1$. 
By definition of local configurations, we have
\[
\mathtt{NbLocalConf}(N) = s^{3\ell} = s^{3\lc L/D \rc}<s^3 s^{3L/D}
\]
because of $\lc L/D \rc < L/D +1$. 
We also have $L=\mathtt{length}(N^d)\leq 1+\log_2 N^d = 1+d\log_2 N$. \linebreak
It comes 
$s^{3L/D} \leq s^{3/D} s^{(3d/D)\log_2 N}= s^{3/D} N^{(3d/D)\log_2 s}$.
from which we deduce
\[
\mathtt{NbLocalConf}(N) <  s^{3+3/D} \times N^{(3d/D)\log_2 s}
\]
This implies $\mathtt{NbLocalConf}(N) < c_1 N^{\sigma}$, for the constant factor $c_1\coloneqq s^{3+3/D}$
and the constant exponent $\sigma\coloneqq (3d/D)\log_2 s$, 
which is less than 1 if we fix the value of our constant integer $D$ to more than $3d \log_2 s$. 
Therefore, setting $D\coloneqq 1+\lf 3d \log_2 s \rf$ is suitable.

\medskip
It is time to present in detail how a RAM can simulate in constant time the cellular automaton $\mathcal{A}=(Q,\delta)$ thanks to pre-computed tables in linear time. 
It is convenient to confuse the active part (resp.\ any $\ell$-block) of a computation of $\mathcal{A}$, which is a non-empty word over the alphabet $Q=\{0,1,\ldots,s-1\}$, with the integer represented by this word in the base~$s$. 
As we have already seen, we can easily go from a string of $s$-digits to the integer represented by this string in the base $s$, by Horner's method, and vice versa.

\paragraph{Computation of local transition tables.} First, to represent the transition function $\delta$ of the CA, let us define the 3-dimensional array 
$\mathtt{DELTA}[0..s-1][0..s-1][0..s-1]$  by
\[
\mathtt{DELTA}[x_2][x_1][x_0]
\coloneqq \delta(x_2,x_1,x_0)
\]
The $\mathtt{DELTA}$ array can be computed in constant time $s^3$ 
since for each $(x_0,x_1,x_2)\in[0..s-1]^3$, $\delta(x_2,x_1,x_0)$ is an explicit integer.


\medskip

Let us now represent the table of local $\rho$-transitions $B_{j+1}^{t}B_{j}^{t}B_{j-1}^{t}\mapsto B_j^{t+\ell}$ from $Q^{3\ell}$ to $Q^{\ell}$, for all the $\rho\in [1,\ell]$, by the 2-dimensional array
$\arr{LT}[1..\ell][0..s^{3\ell}-1]$ defined, for 
$1\leq \rho \leq \ell$, $B_0,B_1,B_2\in Q^{\ell}$ and $B=B_2 s^{2\ell}+B_1s^{\ell}+B_0$, by
$ \arr{LT}[\rho][B] \coloneqq R $,
which is interpreted as follows:  if, at any time $t$, three consecutive $\ell$-blocks numbered $j+1$, $j$, $j-1$ are respectively $B_2,B_1,B_0$, then at time $t+\rho$, the middle $\ell$-block (numbered $j$) is $R$.

\medskip 
The following code computes the $\arr{LT}$ array thanks to the $\mathtt{DELTA}$ array:

\begin{nosAlgos}[Computation of the \arr{LT} array]

\For {$\var{B} \From 0 \To \var{s}^{3\times \ell}-1$}
    \For {$\rho \From 1 \To \ell$}
    
         \State $\var{R} \gets \var{B}$
         \For {$\var{x} \From 0 \To 3\times \ell-1$}
           \State $\arr{C}[\var{x}][0] \gets \var{R} \modop \var{s}$
           \State $\arr{R} \gets \var{R} \divop \var{s}$
        \EndFor
 
    \For {$\var{t} \From 1 \To \rho$}
        \For {$\var{x} \From \ell - \rho + \var{t} \To 2\times \ell - 1 + \rho - \var{t}$}
           \State $\arr{C}[\var{x}][\var{t}] \gets 
           \arr{DELTA}[\arr{C}[\var{x}+1][\var{t}-1]][\arr{C}[\var{x}][\var{t}-1]][\arr{C}[\var{x}-1][\var{t}-1]]$
        \EndFor
    \EndFor
    
    \State $\var{R} \gets 0$
    \For {$\var{x} \From 2\times \ell -1 \Downto \ell$} 
            \State $\var{R} \gets \var{R}\times \var{s}+\arr{C}[\var{x}][\ell]$
    \EndFor
    \State $\arr{LT}[\rho][\var{B}] \gets \var{R}$      
    
    \EndFor
 \EndFor 
\end{nosAlgos}

\noindent
\emph{Justification:}
First, note that the main loop (line 1) goes through the integers $B<s^{3\ell}$ encoding local configurations (triplets of adjacent $\ell$-blocks).
Of course, $C[x,t]$ represents the state of a cell $x\in[0,3\ell-1]$ of the local configuration at an instant numbered $t\in [0,\rho]$ of a time interval of duration $\rho$, and the main assignment (line 9) is a paraphrase of the transition function 
\begin{eqnarray}\label{eqn:delta}
C[x][t] = \delta(C[x+1][t-1],C[x][t-1],C[x-1][t-1])
\end{eqnarray}
The loop condition $x\in[\ell - \rho + t, \;2\ell -1 + \rho - t]$ of line 8 for the application of~(\ref{eqn:delta}) is justified inductively by the locality of the transition function: obviously, the sequence of states at time~$t$ of the cell interval $[\ell - \rho + t, \;2\ell -1 + \rho - t]$
(which is equal to the middle $\ell$-block interval $[\ell, \;2\ell -1]$ for $t=\rho$) are determined by the sequence of states at time $t-1$ of the cell interval $[\ell - \rho + (t-1), \;2\ell -1 + \rho - (t-1)]$, which has one more cell, numbered $\ell - \rho + t - 1$ 
(resp.\ $2\ell -1 + \rho - t + 1$), at each end.

In addition, lines 3-6 of the algorithm transform an integer $B<s^{3\ell}$ (representing a local configuration) into the array of its $3\ell$ $s$-digits: the initial string $C[x][0]$, for $0\leq x \leq 3 \ell-1$.  

Conversely, lines 10-12 transform the final string $C[x][\rho]$, for $\ell \leq x \leq 2 \ell-1$, 
into the integer $R<s^{\ell}$ it represents: 
this is the resulting middle $\ell$-block of the local $\rho$-transition (line 13).

\paragraph{Representing the computation by a two-dimensional array of $\ell$-blocks.} 
We represent the computation of the CA by the sequence of its $c_0+1$ global configurations 
$\mathcal{C}_{t\ell} \coloneqq \sharp B_{c_0}^{t\ell} \cdots B_{1}^{t\ell}  \sharp$ at time $t\ell$, for $0\leq t\leq c_0$, 
where $B_{i}^{t\ell}$ is the $i$th $\ell$-block, $1\leq i \leq c_0$, of $\mathcal{C}_{t\ell}$ from the right.
Since we have $cL=(c_0-1)\ell +\rho$, for the constant $c_0 \coloneqq \lc cL/\ell \rc$ and an integer $\rho$ such that
$1\leq \rho\leq \ell$, note that, in case $\rho < \ell$, we need to add $\ell-\rho$ symbols $\sharp$ at the left of the $\rho$-block numbered~$c_0$ to make it an $\ell$-block.

For convenience, we add to each global configuration $\mathcal{C}_{t\ell}$ two $\ell$-blocks, 
$B_{0}^{t\ell}\coloneqq \sharp^{\ell}$ at its right, and $B_{c_0+1}^{t\ell}\coloneqq \sharp^{\ell}$ at its left. 
Thus, the global configuration $\mathcal{C}_{t\ell}$ is now written 
$B_{c_0+1}^{t\ell} B_{c_0}^{t\ell} \cdots B_{1}^{t\ell} B_{0}^{t\ell}$.
Note that the string $\sharp^{\ell}$ is assimilated to the integer represented in base $s$ by $\ell$ digits $\gamma$, 
$\overline{\gamma\cdots\gamma}$, which, by our convention $\sharp\coloneqq \gamma\coloneqq 2^{r+1}$, 
is equal to $2^{r+1}(s^{\ell-1}+\cdots + s +1)= 2^{r+1}(s^{\ell}-1)/(s-1)$.

\paragraph{Computing the initial configuration.} 
This is the most technical and most tedious part of the simulation.
Recall $cL=(c_0-1)\ell +\rho$, for the ``constant'' 
$c_0 \coloneqq \lc cL/\ell \rc$ and the integer $\rho$, which satisfies
$1\leq \rho\leq \ell$. Similarly, $L=(c_1-1)\ell +\lambda$, for the ``constant'' $c_1 \coloneqq \lc L/\ell \rc$ and the integer $\lambda$, which satisfy $1\leq \lambda\leq \ell$ and $c_1\leq c_0$. 
(We write that $c_0$ and $c_1$ are ``constant'' in quotation marks because, even if they depend on $N$, they are bounded by constants: precisely, we have $c_1=\lc L/\ell \rc=\lc L/ \lc L/D \rc \rc \leq \lc L/ (L/D) \rc=D$ and, likewise, $c_0\leq cD$.)
Also, recall the definition $\mathtt{code}(X_0,X_1,\ldots,X_r)\coloneqq u_{L-1}\cdots u_1 u_0\in \Gamma^L$
where $u_j\coloneqq \sum_{i=0}^{r} x_{i,j} 2^i$, for $0\leq j <L$, and $x_{i,L-1}\cdots x_{i,1} x_{i,0}$ denotes the binary representation of $X_i$ in $L$ bits. In other words, the $j$th $\gamma$-digit $u_j$ of 
$\mathtt{code}(X_0,X_1,\ldots,X_r)$, 
where $\gamma\coloneqq 2^{r+1}$, is
\begin{center} 
$\sum_{i=0}^r (j\mathtt{th}\;\mathtt{bit}\;\mathtt{of}\; X_{i}) \times 2^i$.
\end{center} 
We obtain the initial configuration $\mathcal{C}_{0}=B_{c_0+1}^{0} B_{c_0}^{0} \cdots B_{1}^{0} B_{0}^{0}$ by the four following steps.

\begin{enumerate}
\item Decompose each of the binary strings $X_i\in\{0,1\}^L$, $0\leq j \leq r$, as a concatenation \linebreak
$X_i=X_{c_1,i}X_{c_1-1,i}\cdots X_{1,i}$ of, from right to left,  $c_1-1$ $\ell$-blocks $X_{1,i},\ldots,X_{c_1-1,i}\in\{0,1\}^{\ell}$ and 
a $\lambda$-block $X_{c_1,i}\in\{0,1\}^{\lambda}$; 
\item Compute the decomposition of $\mathtt{Input}\coloneqq\mathtt{code}(X_0,X_1,\ldots,X_r)\in \Gamma^L$ as a concatenation 
$U=U_{c_1}U_{c_1-1}\cdots U_{1}$ of $c_1-1$ $\ell$-blocks $U_{1},\ldots,U_{c_1-1}\in \Gamma^{\ell}$ and a $\lambda$-block $U_{c_1}\in \Gamma^{\lambda}$; here again, for each $h\in[1,c_1]$, we write $U_h=\mathtt{code}(X_{h,0},X_{h,1},\ldots,X_{h,r})$ to mean that for each rank $i$, $0\leq i <\ell$, 
the $i$th $\gamma$-digit of $U_h$ is
$
\sum_{j=0}^r (i\mathtt{th}\;\mathtt{bit}\;\mathtt{of}\; X_{h,j}) \times 2^j 
$;
\item Compute the $\ell$-blocks $B_{j}^{0}$, for $j \in [1, c_1-1]$, of the global configuration 
$\mathcal{C}_{0}=B_{c_0+1}^{0} B_{c_0}^{0} \cdots B_{0}^{0}$ at time 0;
note that, for each $j\in[1,c_1-1]$, the representation of $U_j$ in base $\gamma$, 
denoted $\overline{a_{\ell-1}\cdots a_1 a_0}^{\gamma}$, ``is'' the representation of $B_j^0$ in base $s$, i.e.\ the same string of digits: $B_{j}^{0}=\overline{a_{\ell-1}\cdots a_1 a_0}^s$;
in terms of integers, this means that if we have $U_j=\sum_{i=0}^{\ell-1}a_i \gamma^i$ ($0\leq a_i\leq \gamma-1$), 
then we obtain $B_{j}^{0}=\sum_{i=0}^{\ell-1}a_i s^i$, which we write 
$B_{j}^{0}=\mathtt{convert}_{\gamma\to s}(U_j)$;

\item Compute the other $\ell$-blocks of $\mathcal{C}_{0}$: 
\begin{enumerate}
\item as seen above, we have 
$B_{j}^{0}=\sharp^{\ell}=2^{r+1}(s^{\ell}-1)/(s-1)$, for $j\in \{0\}\cup[c_1+1,c_0+1]$ 
(recall $\sharp\coloneqq \gamma \coloneqq 2^{r+1}$); 
\item finally, the representation of the integer $B_{c_1}^{0}$ in the base $s$ corresponds to the concatenation of strings $\sharp^{\ell-\lambda}$ and $U_{c_1}$.
Thus we obtain
\begin{center}
$ B_{c_1}^{0}=
2^{r+1}(s^{\ell}-s^{\lambda})/(s-1) + U'_{c_1} $
\end{center}
with $U'_{c_1}\coloneqq \mathtt{convert}_{\gamma\to s}(U_{c_1})$, which means that  
$U_{c_1}=\sum_{i=0}^{\lambda-1}a_i \gamma^i$ ($0\leq a_i\leq \gamma-1$) implies 
$U'_{c_1}=\sum_{i=0}^{\lambda-1}a_i s^i$.
\end{enumerate}


\end{enumerate}
To compute the initial configuration in constant time by the previous algorithm (1-4), we need to pre-compute two arrays in linear time.

\paragraph{Computing the tables $\mathtt{CODE}$ and $\mathtt{CONVERT}$.}
Let $\mathtt{CODE}[1..\ell][0..2^{\ell}-1][0..2^{\ell}-1]\cdots [0..2^{\ell}-1]$ be the array defined, for each $\lambda\in [1,\ell]$ and all $X_0,X_1,\ldots,X_r \in [0,2^{\lambda}-1]$, by
\begin{center}
$\mathtt{CODE}[\lambda][X_0],[X_1]\cdots[X_r]\coloneqq \mathtt{code}(X_0,X_1,\ldots,X_r)$
\end{center}
where $ \mathtt{code}(X_0,X_1,\ldots,X_r)$ is the integer in $[0,\gamma^{\lambda} -1]$ ($\gamma=2^{r+1}$) defined by item~2 above. 
The $\mathtt{CODE}$ array is computed by the following algorithm:

\begin{nosAlgos}[Computation of the \arr{CODE} array]
\For {$\lambda \From 1 \To \ell $}
  \For {$(\var{X}_0,\var{X}_1,\ldots,\var{X}_{\var{r}}) \in [0, 2^{\lambda} - 1]^{\var{r}+1}$}
  
    \For {$ \var{j} \From 0 \To \var{r} $}
      \For {$\var{i} \From 0 \To \lambda - 1 $}
        \State $ \arr{BIT}_{\var{j}}[\var{i}] \gets \var{X}_{\var{j}} \modop 2 $
         \State $ \var{X}_{\var{j}} \gets \var{X}_{\var{j}}  \divop 2 $
       \EndFor
     \EndFor
        
    \For {$ \var{i} \From 0 \To \lambda - 1 $}
       \State $ \arr{DIGIT}[\var{i}] \gets 0 $
        \For {$ \var{j} \From \var{r} \Downto 0 $} 
            \State $\arr{DIGIT}[\var{i}] \gets  2 \times  \arr{DIGIT}[\var{i}] + \arr{BIT}_{\var{j}}[\var{i}]$
        \EndFor
     \EndFor
      
    \State $\var{V} \gets 0$
    \For {$ \var{i} \From \lambda - 1 \Downto 0 $}
      \State $\var{V} \gets \gamma \times \var{V} + \arr{DIGIT}[\var{i}]$
    \EndFor
    \State $\arr{CODE}[\lambda][\var{X}_0][\var{X}_1]\cdots[\var{X}_{\var{r}}] \gets \var{V}$
    
     \EndFor
\EndFor
\end{nosAlgos}

\medskip \noindent
\emph{Comment:} 
Lines 3-6 of the algorithm compute the $\lambda$ bits of each integer
$X_0,\dots,X_r<2^{\lambda}$. 
Lines~7-10 and Lines 11-14 compute respectively the $\lambda$ $\gamma$-digits
of the integer $\arr{CODE}[\lambda][\var{X}_0][\var{X}_1]\cdots[\var{X}_{\var{r}}]$
and its value. 

\medskip
Similarly, let $\mathtt{CONVERT}[1..\ell][0..\gamma^{\ell}-1]$ be the array defined, for each length $\lambda\in [1,\ell]$ and each integer $B \in [0,\gamma^{\lambda}-1]$, by
\;$\mathtt{CONVERT}[\lambda][B]\coloneqq \mathtt{convert}_{\gamma\to s}(B)$\;
where $\mathtt{convert}_{\gamma\to s}(B)$ is the integer in $[0,s^{\lambda} -1]$ defined by the following implication (see items 3 and 4 above):
$B=\sum_{i=0}^{\lambda-1}a_i \gamma^i$, for $0\leq a_i <\gamma$, implies 
$\mathtt{convert}_{\gamma\to s}(B) \coloneqq \sum_{i=0}^{\lambda-1}a_i s^i$.
We can easily verify that the $\mathtt{CONVERT}$ array is computed by the following algorithm:

\begin{nosAlgos}[Computation of the \arr{CONVERT} array]
\For {$ \lambda \From 1 \To \ell $}
  \For {$ \var{B} \From 0 \To \gamma^{\lambda} - 1 $}
  
    \State $\var{U} \gets \var{B}$
    \For {$ \var{i} \From 0 \To \lambda - 1 $}
      \State $\arr{DIGIT}[\var{i}] \gets \var{U} \modop \gamma$ 
      \State $\var{U} \gets \var{U} \divop \gamma$
    \EndFor
      
    \State $\var{U} \gets 0$
    \For {$ \var{i} \From \lambda - 1 \Downto 0 $} 
       \State $\var{U} \gets \var{s} \times \var{U} + \arr{DIGIT}[\var{i}]$
     \EndFor
    \State $\arr{CONVERT}[\lambda][\var{B}] \gets \var{U}$
    
  \EndFor 
\EndFor   
\end{nosAlgos}

\medskip
\paragraph{Computing the output by projecting the final configuration.}
This step is similar but much simpler than the computation of the initial configuration from the input. 

For simplicity, we define $\pi(x)$ 
for each $x\in Q$ and not only for 
$x\in Q_{out}\cup \{\sharp\}$ as it is necessary. 
It also convenient to set $\pi(\sharp)\coloneqq 0$ (instead of $\pi(\sharp)\coloneqq \sharp$) to take account of the possible $\sharp$ symbols at the left end of the leftmost $\ell$-block, the $c_0$-th $\ell$-block.
More generally, we set $\pi: Q\to \{0,1\}$ with $\pi(q)\coloneqq 0$ for each $q\in Q\setminus Q_{out}$.

Here again, by its locality, the projection is done for each symbol $x\in Q$ of each $\ell$-block of the final configuration. It uses the array $\mathtt{PROJECT}[0..s^{\ell}-1]$ defined by 
\;$\mathtt{PROJECT}[B]\coloneqq \sum_{i=0}^{\ell-1} \pi(x_i)\times 2^i$
where $B$ is any $\ell$-block $x_{\ell-1}\cdots x_1x_0\in Q^{\ell}$, or equivalently, $B= \sum_{i=0}^{\ell-1} x_i\times s^i$, for $x_i \in [0,s-1]$.
The $\mathtt{PROJECT}$ array is computed by the following algorithm where $\mathtt{PI}[0..s-1]$ denotes the array defined by 
$\mathtt{PI}[x] \coloneqq \pi(x)$, for $x\in Q=[0..s-1]$.

\begin{nosAlgos}[Computation of the \arr{PROJECT} array]
\For {$ \var{B} \From 0 \To \var{s}^{\ell} - 1$} 
    \State $\var{R} \gets \var{B}$
    \For {$ \var{i} \From 0 \To \ell - 1 $} 
         \State $\arr{DIGIT}[\var{i}] \gets \arr{PI}[\var{R} \modop \var{s}]$ 
         \State $\var{R} \gets \var{R} \divop \var{s}$  
     \EndFor
     \State $\var{R} \gets 0$ 
    \For {$ \var{i} \From \ell - 1 \Downto 0 $}  
         \State $\var{R} \gets 2 \times \var{R} + \arr{DIGIT}[\var{i}]$
     \EndFor
    \State $\arr{PROJECT}[\var{B}] \gets \var{R}$
 \EndFor
\end{nosAlgos}


\paragraph{Putting things together: simulation of the cellular automaton.}

Recall that, for \linebreak
$cL=(c_0-1) \ell + \rho$ with $1\leq \rho\leq \ell$, i.e.\ for
$c_0 \coloneqq \lc cL/\ell \rc$ and $\rho\coloneqq cL-(c_0-1)\ell$,  
the computation of the automaton $\mathcal{A}=(Q,\delta)$ decomposes into 
\begin{itemize}
\item $c_0-1$ global $\ell$-transitions, each composed of $c_0$ local $\ell$-transitions, and 
\item one global $\rho$-transition composed of $c_0$ local $\rho$-transitions of time $\rho$.
\end{itemize}
 

Let $B_{x}^{T}$ be the $\ell$-block numbered $x$, $0\leq x \leq c_0+1$, 
of the global configuration $\mathcal{C}_{T}$ at time~$T$, $0\leq T \leq cL$.
The $\mathtt{op}$ operation is computed by the following procedure, called $\fun{op}$, which simulates the automaton $\mathcal{A}=(Q,\delta)$ with its projection $\pi: Q \to \{0,1\}$. 
The main data structure of the $\fun{op}$ procedure is a two-dimensional array, called 
$\mathtt{BK}[0..c_0+1][0..c_0]$ and defined by 
\begin{center}
$\mathtt{BK}[x][t] \coloneqq B_{x}^{t\ell}$ \;for $0\leq t \leq c_0-1$, and
\;$\mathtt{BK}[x][c_0] \coloneqq B_{x}^{cL}$.
\end{center}

\begin{nosAlgos} [Constant-time computation of the $\mathtt{op}$ operation 
]
\Procedure{op}{$\var{X}_1,...,\var{X}_\var{r}$} 
  \State $ \var{X}_0 \gets \var{N} $
  \For {$ \var{j} \From 0 \To \var{r} $}
     \For {$ \var{x} \From 1 \To \var{c}_1 $}
       \State $ \arr{X}[\var{j}][\var{x}] \gets \var{X}_{\var{j}} \modop 2^{\ell} $ 
       \State $ \var{X}_{\var{j}} \gets \var{X}_{\var{j}} \divop 2^{\ell} $
     \EndFor
  \EndFor
     \For {$ \var{x} \From 1 \To \var{c}_1 - 1 $} 

       \State $ \arr{U}[\var{x}] \gets \arr{CODE}[\ell][\arr{X}[0][\var{x}]]\cdots [\arr{X}[\var{r}][\var{x}]] $
   \EndFor
 \State $ \arr{U}[\var{c}_1] \gets \arr{CODE}[\lambda][\arr{X}[0][\var{c}_1]]\cdots [\arr{X}[\var{r}][\var{c}_1]] $
 \State $ \arr{BK}[0][0] \gets 2^{\var{r}+1}\times (\var{s}^{\ell} - 1) / (\var{s}-1)  $
   \For {$ \var{x} \From 1 \To \var{c}_1 - 1 $} 
      \State $ \arr{BK}[\var{x}][0] \gets \arr{CONVERT}[\ell][\arr{U}[\var{x}]] $
   \EndFor
 \State $ \arr{BK}[\var{c}_1][0] \gets \arr{CONVERT}[\lambda][\arr{U}[\var{c}_1]] + 
 2^{\var{r}+1}\times (\var{s}^{\ell} - \var{s}^{\lambda}) / (\var{s}-1)  $
   \For {$ \var{x} \From \var{c}_1+1 \To \var{c}_0+1 $} 
      \State $ \arr{BK}[\var{x}][0] \gets 2^{\var{r}+1}\times (\var{s}^{\ell} - 1) / (\var{s}-1) $
   \EndFor

   \For {$ \var{t} \From 1 \To \var{c}_0 - 1 $}
     \For {$ \var{x} \From 1 \To \var{c}_0 $}
       \State $ \arr{BK}[\var{x}][\var{t}] \gets \arr{LT}[\ell][\fun{Conc}(\arr{BK}[\var{x}+1][\var{t}-1],\arr{BK}[\var{x}][\var{t}-1],\arr{BK}[\var{x}-1][\var{t}-1])] $
     \EndFor
  \EndFor
  \For {$ \var{x} \From 1 \To \var{c}_0 $} 
     \State $ \arr{BK}[\var{x}][\var{c}_0] \gets \arr{LT}[\rho][\fun{Conc}(\arr{BK}[\var{x}+1][\var{c}_0-1],\arr{BK}[\var{x}][\var{c}_0-1],\arr{BK}[\var{x}-1][\var{c}_0-1])] $
  \EndFor

  \State $ \var{R} \gets 0 $
  \For {$ \var{x} \From \var{c}_0 \Downto 1 $} 
      \State $ \var{R} \gets \var{s}^{\ell} \times \var{R} + \arr{PI}[\arr{BK}[\var{x}][\var{c}_0]]  $
  \EndFor
  \State $ \Return\; \var{R} $
\EndProcedure 
\end{nosAlgos}


\smallskip \noindent
\emph{Comments:} Lines 2-15 (the major part of the $\fun{op}$ procedure!) construct the part of the 
$\arr{BK}$ array, $BK[x][0]$, $0\leq x \leq c_0+1$, 
that represents the initial configuration $\mathcal{C}_0$ of $\mathcal{A}$
on the input $\mathtt{code}(N,X_1,\ldots,X_r)$. More precisely, the construction decomposes into the following steps:
\begin{itemize}
\item from the integers $X_0\coloneqq N,$ $X_1,\ldots,X_r$, lines 2-6 of the $\fun{op}$ procedure compute the array $\arr{X}[0..r][1..c_1]$ such that $\arr{X}[i][x]$ is the $x$-th bit of $X_i$;
\item from the array $\arr{X}[0..r][1..c_1]$,  lines 7-9 compute the array $\arr{U}[1..c_1]$ such that $\arr{U}[x]$ is the $x$-th $\gamma$-digit of 
$\mathtt{code}(X_0,X_1,\ldots,X_r)$ (for $\gamma\coloneqq 2^{r+1}$), which is $\arr{U}[x]=\sum_{i=0}^r 2^i \times \arr{X}[i][x]$;
\item from the $\arr{U}[1..c_1]$ array, lines 10-15 construct the part $\mathtt{BK}[0..c_0+1][0]$ of the 
$\mathtt{BK}$ array that encodes the initial configuration 
$\mathcal{C}_0=\sharp B^0_{c_0+1}\cdots B^0_{1} B^0_{0} \sharp$ with the $\ell$-blocks
$B^0_{x}$ among which 
\begin{itemize}
\item $B^0_{0}$ and $B^0_{c_1+1},\ldots, B^0_{c_0+1}$ are equal to the string $\sharp^{\ell}$ encoded by the integer \linebreak
$2^{r+1}(s^{\ell}-1)/(s-1)$ (line 10 and lines 14-15, respectively, corresponding to item 4.a above), and
\item corresponding to the ``input'', each $\mathtt{BK}[j]$, $1\leq j \leq c_1-1$, is the integer whose string of 
($\ell$) digits in base $s$ is equal to $B^0_j$ (lines 11-12 corresponding to item 3),
and $\mathtt{BK}[c_1]$ is the integer whose string of 
($\ell$) digits in base $s$ is $B^0_{c_1}=\sharp^{\ell-\lambda}U_{c_1}$ (line 13 corresponding to item 4.b; recall that $\lambda$ is the length of $U_{c_1}$).
\end{itemize}

\end{itemize}

Lines 16-20 simulate the computation of $\mathcal{A}$: the $c_0-1$ global $\ell$-transitions (lines 16-18) and the final global $\rho$-transition (lines 19-20).

Finally, lines 21-24 compute and return the output integer by projecting the final configuration~$\mathcal{C}_{cL}$ with the $\pi$ projection.
 
It is immediate to check that the body of each ``for'' loop of the $\fun{op}$ procedure is repeated a number of times bounded by $r+1$, $c_0$ or $c_1$, therefore by a constant, and therefore the $\fun{op}$ procedure executes in constant time as claimed. 

\paragraph{Why is the preprocessing time linear?} 
Our constant-time procedure $\fun{op}$ uses five pre-computed tables: $\mathtt{PI}$, computed in constant time,
and $\arr{LT}$, $\arr{CODE}$, $\arr{CONVERT}$ and $\arr{PROJECT}$.
By examining the structure of the loops of the above algorithms, it is easy to verify that 

\begin{itemize}
\item $\arr{LT}$ is computed in time $O(\mathtt{NbLocalConf}(N)\times \ell^3)= 
O(N^{\sigma}\times (\log N)^3)= O(N)$
(recall $\mathtt{NbLocalConf}(N)=s^{3\ell}=O(N^{\sigma})$, for a fixed $\sigma<1$), 
\item $\arr{CODE}$ and $\arr{CONVERT}$ are computed in time 
$O((2^{\ell})^{r+1}\times \ell^2)=O(\gamma^{\ell}\times \ell^2)$, 
and $O(\gamma^{\ell}\times \ell^2)$, respectively, and therefore in time $O(s^{\ell}\times \ell^2)$, and 
$\arr{PROJECT}$ is computed in time $O(s^{\ell}\times \ell$), and therefore, a fortiori, all these tables are computed in time $O(N)$.
\end{itemize}
This completes the proof of Theorem~\ref{th:linTimeCA}.
\end{proof}
\section{Appendix: Another constant-time algorithm for division}

While multiplication can be computed in linear time on a cellular automaton, so Theorem~\ref{th:linTimeCA} applies to it, it is unknown if division is.  
However, the intuitive principle of the usual division algorithm is also ``local'': to compute the most significant (leftmost) digit of the quotient of integers $\lf A/B \rf$, the first step is to ``try'' to divide the integer $a$ consisting of the 1, 2, 3, or 4 leftmost digits of $A$ by the integer~$b$ consisting of the 1, 2, or 3 leftmost digits of $B$. 

We are going to design a new constant-time algorithm for division by a careful adaptation of the usual division principle combined with a judicious choice of the base $K$ in which the integers are represented.
To avoid any vicious circle here, we will only use the addition, subtraction, multiplication and base change of ``polynomial'' integers, and the procedure $\fun{DivBySmall}(A,B)$ of Subsection~\ref{subsec:divbysmallint}, which returns $\lf A/B \rf$ for $A<\beta^d$, $0<B<\beta$, for a fixed integer $d\geq 1$ and $\beta\coloneqq \lc N^{1/6} \rc^3$.

\paragraph{Analyzing and optimizing the usual division algorithm.} Assume the dividend $A$ and the divisor $B$ are represented in any fixed base $K\ge 2$ and let $a$ (resp.\ $b$) be the integer represented by the 1, 2, 3, or 4 (resp. 1, 2, or 3) leftmost digits of $A$ (resp. $B$). Each positive integer is assimilated to the string $\overline{x_{p-1}\cdots x_1 x_0}$ of digits $x_i$, $0\leq x_i<K$, which represents it in base~$K$, with leading digit $x_{p-1}>0$. Clearly, we get the inequalities
$a K^m \leq A < (a+1) K^m$ and $b K^n \leq B < (b+1) K^n$, for some integers $a,b\geq 1$ and $m,n$. 
Because of $A\geq B$, we can assume $m\geq n$. It implies
$\frac{a}{b+1}K^{m-n} < A/B < \frac{a+1}{b}K^{m-n}$.
To do \emph{one} test instead of several tests to get the \emph{most significant} digit of the quotient 
$A/B$ in base $K$, we want the difference between the bounds to be less than $K^{m-n}$, it means
$\frac{a+1}{b}-\frac{a}{b+1}\leq 1$.
Noticing the identity $\frac{a+1}{b}-\frac{a}{b+1}= \frac{a}{b(b+1)}+\frac{1}{b}$, this can be rephrased as 
\begin{eqnarray}\label{eqn:identity}
\frac{a}{b(b+1)}+\frac{1}{b}\leq 1.
\end{eqnarray}

\begin{lemma}\label{lemma:q or q+1}
Assume $a K^m \leq A < (a+1) K^m$ and $b K^n \leq B < (b+1) K^n$, for some integers $K\geq 2$, and $A,B$ such that 
$A> B \geq 1$, and $m,n,a,b$ such that $m\geq n \geq 0$, and $a > b \geq 1$, and $\frac{a}{b(b+1)}+\frac{1}{b}\leq 1$.
We then get
\begin{center}
$q K^{m-n} \leq A/B < (q+1)K^{m-n}$ for $q=\lf \frac{a}{b+1} \rf$ or $q=\lf \frac{a}{b+1} \rf+1$,
\end{center}
which implies $A = B q K^{m-n}+R$ for some integer
$R< B K^{m-n}$, and therefore 
\begin{eqnarray}\label{eqn: recursiveDiv}
\lf A/B \rf = q K^{m-n}+ \lf R/B \rf 
\end{eqnarray}
with $q\geq 1$ and $\lf R/B \rf < K^{m-n}$.

\end{lemma}

\begin{proof}
From the inequalities $\frac{a}{b+1}K^{m-n} < A/B < \frac{a+1}{b}K^{m-n}$ we deduce
$\lf\frac{a}{b+1}\rf K^{m-n} < A/B < \lc \frac{a+1}{b}\rc K^{m-n}$.
Besides, the inequality $\frac{a+1}{b}-\frac{a}{b+1}\leq 1$ implies $\lc\frac{a+1}{b}\rc\leq \lf\frac{a}{b+1}\rf +2$.
This yields $q K^{m-n} < A/B < (q+2) K^{m-n}$ for $q=\lf \frac{a}{b+1} \rf \geq 1$, which proves the lemma.
\end{proof}

\begin{remark}
Lemma~\ref{lemma:q or q+1} gives and justifies the principle of a recursive division algorithm in any fixed base $K\geq 2$: 
to compute $\lf A/B \rf$, compute the integers $q$ and $\lf R/B \rf$ in base $K$; 
the representation of the integer $\lf A/B \rf$ in base $K$ is the concatenation of the string representing the integer $q\geq 1$ and the string of $m-n$ digits representing the integer $\lf R/B \rf< K^{m-n}$ (possibly with leading zeros).
\end{remark}

\paragraph{How to use Lemma~\ref{lemma:q or q+1}?} How many leftmost (most significant) digits of $A$ and $B$, written in base $K$, should we take to obtain integers $a$ and $b$, with $a\geq b\geq 1$, that satisfy the condition $\frac{a}{b(b+1)}+\frac{1}{b}\leq 1$ of Lemma~\ref{lemma:q or q+1}? This is expressed precisely by the following question: 
for which integers $i,j$ is the implication 
$K^i \leq a < K^{i+1} \land K^j \leq b < K^{j+1} \Rightarrow \frac{a}{b(b+1)}+\frac{1}{b}\leq 1$ 
true for all~$a,b$?

\begin{lemma}\label{lemma:i<2j}
Let $i,j$ be two positive integers. Then the condition $i< 2j$ is equivalent to the condition
\[ 
\forall (a,b)\in[K^i,K^{i+1}[\times[K^j,K^{j+1}[ \;\; \left(\frac{a}{b(b+1)}+\frac{1}{b}\leq 1 \right)
\] 
\end{lemma}

\begin{proof}
The maximum of the numbers 
$\frac{a}{b(b+1)}+\frac{1}{b}$ for all pairs of integers $(a,b)\in[K^i,K^{i+1}[\times[K^j,K^{j+1}[$ 
is $\frac{K^{i+1}-1}{K^j(K^j+1)}+\frac{1}{K^j}=\frac{K^{i+1}+K^j}{K^{2j}+K^j}$, which is at most 1 if and only if $i+1\leq 2j$.
This proves the lemma.
\end{proof}

We are now ready to give our new division algorithm using the principle given by Lemma~\ref{lemma:q or q+1}. Here again, we must choose carefully the base $K$.

\paragraph{The choice of the base.} We take $K\coloneqq\lc N^{1/7} \rc$
and we assume that each ``polynomial'' operand $X<N^d$, for a fixed integer $d$, is written 
$\overline{x_{m-1}\cdots x_1 x_0}$, $0\leq x_i\leq K-1$, in base~$K$, with one digit $x_i$ per register, 
and is therefore stored in $m\leq 7d$ registers, since  $X<N^d\leq K^{7d}$. In particular, the following functions can be computed in constant time for their argument $X$ so represented: 
$\mathtt{length}_{\mathtt{K}}(X)\coloneqq m$;  
$\mathtt{Substring}_{\mathtt{K}}(X,i,j)\coloneqq\overline{x_{i-1}\cdots x_j}$, for $m\geq i > j \geq 0$.

\begin{notation}
For an integer $X$ written $\overline{x_{m-1}\cdots x_1 x_0}$ in base $K$, $0\leq x_i\leq K-1$,  \linebreak 
and a positive integer $p\leq m$, denote by $X_{(p)}$ the prefix of the $p$ most significant digits of $X$: 
\linebreak
$X_{(p)} \coloneqq \overline{x_{m-1} \cdots x_{m-p}}= X \divop K^{m-p}$.

\end{notation}

\paragraph{The division algorithm.} Let $A,B$ be two ``polynomial'' integers such that $0<B\leq A < K^d$, for a fixed integer $d$. To divide $A$ by $B$, we consider four cases: 

\begin{enumerate}
\item  $B<K^2$;
\item $K^p\leq B \leq A< K^{p+1}$, for an integer $p\geq 2$, and $A_{(3)} = B_{(3)}$; 
\item $K^2\leq B \leq A$ and $A_{(3)} > B_{(3)}$;
\item $K^p\leq B < K^{p+1} \leq A$, for an integer $p\geq 2$, and $A_{(3)} \leq B_{(3)}$.
\end{enumerate}
Obviously, these four cases are mutually exclusive.
To ensure that there is no other case than these four cases, just note that if condition 2 is false when $K^2\leq B\leq A < K^d$, then either we have $A_{(3)} > B_{(3)}$ (case 3), or we have 
$A_{(3)} \leq B_{(3)}$ \emph{and} $A$ has \emph{more} digits than $B$ in base $K$ (case~4).

\begin{itemize}
\item Case~1 is handled by the procedure $\fun{DivBySmall}(A,B)$ of Subsection~\ref{subsec:divbysmallint} 
(note that \linebreak
$K^2=\lc N^{1/7} \rc^2 \leq \lc N^{1/6} \rc^3=\beta$).

\item In case 2, 
we clearly get $A-B<K^{p-2}<B$, which implies $A<2B$ and therefore $\lf A/B \rf=1$.

\item  Cases 3 and 4 are handled by Lemma~\ref{lemma:q or q+1}.
In case 3 (resp.\ case~4), take $a=A_{(3)}$ and $b=B_{(3)}$ 
(resp.\ $a=A_{(4)}$ and $b=B_{(3)}$), which implies 
$K^2\leq b < 
a < K^3$ (resp.\ \linebreak
$K^2\leq b <K^3 \leq a < K^4$). 
\end{itemize}

The integers $i=2$ and $j=2$ (resp.\ $i=3$ and $j=2$) satisfy the condition  $i<2j$ of Lemma~\ref{lemma:i<2j}. Thus, in cases 3 and 4, the condition  $\frac{a}{b(b+1)}+\frac{1}{b}\leq 1$ of Lemma~\ref{lemma:q or q+1} is satisfied. 
Therefore, Lemma~\ref{lemma:q or q+1} proves the correctness of the following $\fun{Division}$ procedure which implements cases~1-4 for \linebreak
$A\geq B$. More precisely, the lines 17-21 of the algorithm compute 
$\lf A/B \rf$ in cases 3 and 4 by using the pre-computed array $\arr{DIV}[0..K^4-1][2..K^3]$ defined by $\arr{DIV}[x][y] \coloneqq \lf x / y \rf$, for $0\leq x < K^4$ and $2\leq y \leq K^3$ 
(see the pre-computation code of the $\arr{DIV}$ array presented after the $\fun{Division}$ procedure). 

\begin{nosAlgos}[Constant-time computation of the $\mathtt{division}$ operation]
\Procedure{Division}{$\var{A},\var{B}$}
    \If {$\var{A} < \var{B}$} 
      \State \Return $0$
    \EndIf
    \If {$\var{B} < \var{K}^2$} \Comment{case 1}
       \State \Return $\fun{DivBySmall}(\var{A},\var{B})$ 
    \EndIf 
    \State $\ell_{\var{A}} \gets \mathtt{length}_{\var{K}}(\var{A})$ 
    \State $\var{a} \gets \mathtt{Substring}_{\var{K}}(\var{A},\ell_{\var{A}},\ell_{\var{A}}-3)$ 
    \Comment{$a \gets A_{(3)}$} 
    \State $\var{m} \gets \ell_{\var{A}}-3$
    \State $\ell_{\var{B}} \gets \mathtt{length}_{\var{K}}(\var{B})$ 
    \State $\var{b} \gets \mathtt{Substring}_{\var{K}}(\var{B},\ell_{\var{B}},\ell_{\var{B}}-3)$  
    \Comment{$b \gets B_{(3)}$}
    \State $\var{n} \gets \ell_{\var{B}}-3$
    \If {$\ell_{\var{A}}  = \ell_{\var{B}}~\mathbf{and}~\var{a} = \var{b}$}
        \State \Return $1$
        \Comment{$\mathtt{length}_{\var{K}}(\var{A})=\mathtt{length}_{\var{K}}(\var{B})$ and 
         $A_{(3)}= B_{(3)}$: case 2}
    \EndIf 
    \If {$\var{a} \leq \var{b}$}   \Comment{$A_{(3)} \leq B_{(3)}$: case 4}
        \State $\var{a} \gets \mathtt{Substring}_K(\var{A},\ell_{\var{A}},\ell_{\var{A}}-4)$ 
        \Comment{$a \gets A_{(4)}$, so $b=B_{(3)} < a < (b+1)K$}
         \State $\var{m} \gets \ell_{\var{A}}-4$
    \EndIf
    \State $\var{q} \gets \arr{DIV}[\var{a}][\var{b}+1]$ 
    \Comment{$b+1\leq a < (b+1)K$ and therefore $1\leq q < K$}\\
    \Comment{common treatment of cases 3 and 4 justified by Lemma~\ref{lemma:q or q+1}}
    \If {$\var{A} \geq \var{B}\times (\var{q}+1)\times \var{K}^{\var{m}-\var{n}}$} 
         \State $\var{q} \gets \var{q}+1$
    \EndIf
    \State $\var{R} \gets \var{A} - \var{B}\times \var{q}\times \var{K}^{\var{m}-\var{n}}$ 
    \State \Return $\var{q}\times \var{K}^{\var{m}-\var{n}}+ \fun{Division}(\var{R},\var{B})$
    \Comment{justified by equality~\ref{eqn: recursiveDiv}}
\EndProcedure
\end{nosAlgos}

\paragraph{Comments:} In case 4 ($a=A_{(3)}\leq B_{(3)}=b$), after the assignment $a \gets A_{(4)}$ of line 15, we have 
$a=A_{(4)} < (A_{(3)}+1)K$. 
Now, by the hypothesis $A_{(3)}\leq B_{(3)}=b$ (line 14), we get  \linebreak
$b=B_{(3)}<a=A_{(4)} < (A_{(3)}+1)K \leq (b+1)K$, and thus $b<a <(b+1)K$, as the comment of line 15 claims.

In case 3, i.e.\ for $a=A_{(3)}> B_{(3)}=b$, we also have $a=A_{(3)}<B_{(4)}<(B_{(3)}+1)K=(b+1)K$.
This justifies, for cases 3 and 4, the assertion given as a comment of line 17: $b+1\leq a<(b+1)K$.

\bigskip 
It remains to describe the preprocessing and to prove that the claimed complexity is achieved.

\paragraph{Pre-computation 
of the $\mathtt{DIV}$ array.} Note that in the 
$\fun{Division}$ procedure, we compute $q=\lf a/(b+1) \rf$ when $2\leq b+1 \leq K^3$ and $1\leq q \leq K-1$. Therefore, the following algorithm constructs 
the $\mathtt{DIV}$ array.

\begin{nosAlgos}[Pre-computation of the \arr{DIV} array]
\For {$\var{y} \From 2 \To \var{K}^3$}
  \For {$\var{q} \From 1 \To \var{K}-1$}
     \For {$\var{r} \From 0 \To \var{y}-1$}
      \State $\var{x} \gets \var{q}\times \var{y} + \var{r} $
      \State $\arr{DIV}[\var{x}][\var{y}] \gets \var{q}$
    \EndFor
  \EndFor
\EndFor
\end{nosAlgos}

\paragraph{Complexity of the preprocessing and division algorithm.} We immediately see that the time of the above algorithm which constructs the array $\arr{DIV}[0..K^4-1][2..K^3]$
is $O(K^3\times K \times K^3)=O(K^7)=O(N)$.

The complexity of the $\fun{Division}$ procedure is determined by its recursive depth (number of recursive calls). 
By equation~\ref{eqn: recursiveDiv} (implemented by line~21), which is $\lf A/B \rf = q K^{m-n}+ \lf R/B \rf$
with $q\geq 1$ and $\lf R/B \rf < K^{m-n}$, the quotient  $\lf R/B \rf $ has at least one digit less than $\lf A/B \rf$ when represented in base $K$ (the digit $q$, multiplied by $K^{m-n}$). Therefore, the recursive depth is bounded by the number of digits of the representation of $\lf A/B \rf$ in base $K$, which is less than~$7d$ because of $\lf A/B \rf<K^{7d}$.
Since the most time consuming instruction is the multiplication, executed $O(d)$ times (lines 18, 20, 21) and whose time cost is $O(d^2)$, the time of the $\fun{Division}$ procedure is $O(d^3)$, at compare to the time $O(d^4)$ of the $\fun{Division}$ procedure of Section~\ref{section:division}, see Theorem~\ref{theorem: division constant time}. 